\documentclass[12pt,english]{article}
\usepackage{lmodern}
\usepackage{booktabs}
\usepackage[table]{xcolor}
\usepackage{float}
\usepackage{caption}
\usepackage{subcaption}
\definecolor{mixedrow}{gray}{0.92}
\usepackage[T1]{fontenc}
\usepackage[utf8]{inputenc}
\usepackage{geometry}
\usepackage{color}
\usepackage{babel}
\usepackage{mathtools}
\usepackage{amsmath}
\usepackage{amsthm}
\newcommand{\E}{\mathbb{E}}
\usepackage{mathrsfs}
\usepackage{amssymb}
\usepackage{enumitem}
\usepackage{graphicx}
\usepackage[authoryear]{natbib}
\usepackage[unicode=true,
 bookmarks=true,bookmarksnumbered=false,bookmarksopen=false,
 breaklinks=false,pdfborder={0 0 1},backref=false,colorlinks=true]
 {hyperref}
\hypersetup{
 pdfborderstyle=,pdfpagelayout=OneColumn,pdfnewwindow=true,pdfstartview=XYZ,plainpages=false,urlcolor=[rgb]{0.0430,0,0.5},linkcolor=[rgb]{0.0430,0,0.5},citecolor=[rgb]{0.0430,0,0.5},hypertexnames=false}
 \usepackage{setspace}

\allowdisplaybreaks[1]

\usepackage[font=small,skip=4pt]{caption}

\usepackage{titlesec}
\titlespacing*{\section}{0pt}{1.2ex plus .2ex minus .2ex}{0.6ex plus .2ex}
\titlespacing*{\subsection}{0pt}{1.0ex plus .2ex minus .2ex}{0.4ex plus .2ex}
\titlespacing*{\subsubsection}{0pt}{0.8ex plus .2ex minus .2ex}{0.3ex plus .2ex}

\usepackage{microtype}
\setlist{itemsep=0pt, topsep=4pt, parsep=0pt, partopsep=0pt}

\makeatletter
\theoremstyle{plain}
\newtheorem{prop}{\protect\propositionname}
\theoremstyle{definition}
\newtheorem{defn}{\protect\definitionname}
\theoremstyle{plain}
\newtheorem{lem}{\protect\lemmaname}
\theoremstyle{plain}
\newtheorem{thm}{\protect\theoremname}
\newtheorem{cor}{\protect\corollaryname}

\usepackage{chngcntr}
\setcitestyle{round}
\usepackage{breakcites}
\usepackage[all]{hypcap}
\usepackage{dcolumn}

\providecommand{\corollaryname}{Corollary}
\providecommand{\definitionname}{Definition}

\providecommand{\lemmaname}{Lemma}
\providecommand{\propositionname}{Proposition}
\providecommand{\theoremname}{Theorem}
\begin{document}
\title{Learning from Viral Information\thanks{We thank Leonie Baumann, Michel Bena\"{i}m, Aislinn Bohren, Tommaso Denti, Glenn Ellison, Mira Frick, Drew Fudenberg, Co-Pierre Georg, Ben Golub, Benjamin Hebert, Ryota Iijima, Bart Lipman, George Mailath, Suraj Malladi, Chiara Margaria, Meg Meyer, Evan Sadler, Philipp Strack, Heidi Thysen, Fernando Vega-Redondo, Rakesh Vohra, Yu Fu Wong, numerous seminar participants, and the editor and anonymous referees for valuable comments and discussions. Byunghoon Kim, Luis Henrique Linhares, Matt Murphy, Stephan Xie, and Tyera Zweygardt provided excellent research assistance. We gratefully acknowledge financial support from NSF Grants SES-2214950 and SES-2215256.}}
\author{Krishna Dasaratha\thanks{Boston University. Email: \texttt{\protect\href{mailto:krishnadasaratha@gmail.com}{krishnadasaratha@gmail.com}}}
\and Kevin He\thanks{University of Pennsylvania. Email: \texttt{\protect\href{mailto:hesichao@gmail.com}{hesichao@gmail.com}}}}

\date{{\normalsize{}}%
\begin{tabular}{rl}
First version: & August 20, 2022\tabularnewline
This version: & June 25, 2026 \tabularnewline
\end{tabular}}
\maketitle
\begin{abstract}
{\normalsize{}\thispagestyle{empty}
\setcounter{page}{0}}{\normalsize\par}

Motivated by social media, we study an equilibrium model  of  agents  interacting with and learning from each other's signals. Rational agents arrive sequentially, observe a signal (corresponding to a news story) and a sample of predecessors' signals (corresponding to a news feed), and decide which of these signals to endorse. The observed sample is jointly determined by predecessors' endorsement behavior and  a sampling rule (capturing a platform algorithm). We focus on how often the sampling rule selects more viral (i.e., widely endorsed) signals. Showing agents viral signals can increase information aggregation, but it can also generate steady states where most endorsed signals are wrong. These misleading steady states self-perpetuate, as agents who observe wrong signals develop wrong beliefs, and thus rationally continue to endorse them. We highlight several consequences of our results for social-media platforms.

\bigskip{}

\end{abstract}
\makeatother

\providecommand{\definitionname}{Definition}
\providecommand{\lemmaname}{Lemma}
\providecommand{\propositionname}{Proposition}
\providecommand{\theoremname}{Theorem}

\onehalfspacing

\newpage

% change initialization
% discuss initial score = 1 versus 0?
% asymmetric priors

\section{Introduction}

People are often exposed to information spreading through societies, and learning outcomes depend substantially on what content spreads widely. A leading example is social-media platforms, where the content viewed by users is largely mediated by platform-generated news feeds. Whether a story spreads widely or fades from view depends jointly on the algorithms that curate these feeds and on users' endogenous actions, such as retweeting, sharing, and upvoting.

How does the design of the news feed affect how users learn on such platforms? Consider a platform deciding how much to push widely shared (or highly upvoted) content into users' news feeds. On the one hand, a
news feed that primarily shows users widely shared stories can create
a social version of  confirmation bias: incorrect stories that become popular early on  can shape later users' beliefs, even though most subsequent information points the other way.   One might expect such  feedback loops  with naive users, but we show they can also arise in an equilibrium model with rational users. The idea is that when stories supporting an incorrect position are shared more, later users tend to see these incorrect stories in their news feeds due to the stories' popularity, and hence  form incorrect beliefs through Bayesian updating. If users derive utility from sharing accurate content and thus share stories that agree with their beliefs, they will rationally share  these false stories and further increase their popularity.  Users have less exposure   to
the true stories: even if these stories are more numerous, they are shared less than the false stories and therefore shown less by the news-feed algorithm.

But on the other hand, selecting news stories based on their popularity
may help aggregate more information. Seeing a particular story in a
news feed that selects widely shared content gives a user more information than the realization of a
single signal. The popularity of this story also tells the user about the past sharing decisions
of their predecessors, and thus lets the user draw inferences about
the many stories that these predecessors saw in their news feeds. In some circumstances,  seeing just a few stories in a news feed that primarily shows viral content can lead to strong Bayesian beliefs about the state of nature, even if individual stories are imprecise signals about the state. This is because   sophisticated users can use the selection of these stories to infer much more about sharing on the platform.

To formalize and explore these tradeoffs, we develop a social-learning model where agents sample and interact with others' signals.  We will draw on mathematical techniques from a stochastic approximation literature to track the dynamics of signals and interactions in a society. We now describe a set of modeling assumptions that will let us apply these techniques to learning dynamics under equilibrium behavior.

%To formalize these tradeoffs, we develop a model of content sharing on a platform, using tools from stochastic approximation to characterize the dynamics of content under equilibrium sharing behavior.

A large number of agents arrive in turn and learn about a binary state. Each agent receives a conditionally independent binary signal about the state (corresponding to a news story) and observes a sample of signals from predecessors (corresponding to a news feed). The sampling rule interpolates between choosing a uniform sample of the past signals and  choosing each signal with probability proportional to its \emph{popularity score}, which increases as agents \emph{endorse} the signal.   Sampling rules are parametrized by a \emph{virality weight}  $\lambda$ that captures the weight placed on  popularity: higher $\lambda$ corresponds to showing more popular signals.  Agents are Bayesians and know the sampling rule, so they appropriately account for  selection in the signals they see.\footnote{An alternative approach would be to assume agents are naive and fail to account for this selection. Many of the main forces we highlight in our equilibrium framework would also appear in this behavioral model.}  Agents then choose which of these sampled signals to endorse. We assume agents prefer to endorse signals that match the true state, given their endogenous beliefs. This simple utility specification, which one might think is conducive to learning, can nevertheless generate rich learning dynamics such as persistent learning failures.

We next describe our results. The evolution of the system is described by a stochastic process in $[0,1]$ we call \emph{viral accuracy}, which measures the relative popularity of the signals that match the true state in each period. We show viral accuracy almost surely converges to a steady-state value, but there can be multiple steady states and which steady state is reached can depend on the realizations of private signals and sampling. In equilibrium, there is always an \emph{informative steady state} where most sampled signals match the state. But when the virality weight is high enough, equilibrium also induces a \emph{misleading steady state}, where most sampled signals do not match the state (so viral accuracy is less than $\frac12$). At a misleading steady state, agents tend to see false signals, and therefore believe in the wrong state and endorse these false signals. The misleading steady states correspond to the socially generated confirmation bias described above.

These misleading steady states emerge when $\lambda$ crosses a threshold, which we call the \emph{critical virality weight} $\lambda^*$. Misleading steady states always exist in equilibrium  when the virality weight is at or above this threshold, but not below it.  A key finding is that this emergence is \emph{discontinuous}: at the  threshold virality level $\lambda^*$ where the misleading steady state first appears, the probability that learning converges to this bad steady state is strictly positive. As a consequence, expected accuracy jumps downward at this threshold. Below the critical virality weight, however, the unique informative steady state becomes monotonically more accurate as $\lambda$ increases. This result formalizes the intuition mentioned above that a more viral sampling rule helps aggregate more information.  Increasing $\lambda$ therefore leads to a trade-off between facilitating more information aggregation and preventing the possibility of a misleading steady state in equilibrium.

After characterizing equilibrium steady states, we give two consequences for social-media platforms. First, the model predicts a popularity score distribution, which corresponds to an empirical distribution of the number of shares (or likes, retweets, etc.) on social-media platforms. At steady state, these popularities converge to a stationary distribution, which we solve for explicitly. These distributions have power-law tails whose thickness is a simple function of model parameters and the realized steady state; empirical work on distributions of shares on social media platforms also finds power-law structure (e.g., \cite{kwak2010twitter} and \cite{garg2025political}).

Second, we ask a natural question with potential implications for the regulation of platforms: to what extent can changes to platforms improve accuracy? We  describe a content-neutral change to algorithms that leads to better learning outcomes: letting the virality weight $\lambda$ vary over time. Consider generating initial agents' news feeds with a low virality weight  but later agents' news feeds with a high virality weight. We show there is a simple equilibrium that achieves high viral accuracy without producing misleading steady states. Intuitively, one way to improve learning is to let independent information accumulate early in the discussion of a new issue and then exploit the advantages of showing viral content later in the discussion.

%In our main model, agents observe the realizations of stories but not how viral these stories are. Motivated by existing social-media platforms, where users usually observe some information about the popularity of posts, we also consider a modified model where agents can distinguish between \emph{viral} and \emph{regular} stories in their news feeds. Each time a user shares a regular story, there is a chance of it ``going viral'' and  creating a corresponding viral copy. A consistent conclusion that we find in both the main model and this modified model is that there must be misleading steady states when users see enough viral stories.

We close this introduction by describing the techniques underlying our modeling and analysis. We study a model where society always converges to a steady-state distribution of sampled signals, but the model allows the possibility of multiple steady states. Agents do not know whether they are at a misleading or informative steady state, and instead form Bayesian beliefs based on the probabilities of reaching the various steady states. We therefore cannot begin our model at a steady state, as these probabilities would be indeterminate,  and must instead track the stochastic evolution of the society from its initial conditions. To analyze this evolution, we rely on stochastic approximation techniques from mathematics (rather than more traditional tools from the social-learning literature, such as martingale convergence theorems applied to the public belief process).

Our paper applies these stochastic approximation  tools to an equilibrium model where agents respond optimally to the evolution of a stochastic process. The same techniques have been used in economics to study dynamics under behavioral heuristics (e.g., \cite{benaim2003deterministic} in evolutionary game theory or \cite*{arieli2022aggregate} in naive social learning). By contrast, applying these tools to a setting where agents use equilibrium strategies is more complex (see also \cite{iijima2025mean} in the evolutionary game theory setting).  A key challenge is that there is no closed-form expression for the probabilities of reaching different steady states, even under a fixed strategy. We are nevertheless able to characterize the possible steady states under equilibrium behavior qualitatively. To do so, we show that outcomes under a specific simple strategy (namely, sharing signals that match the majority of one's observations) tell us about  equilibrium outcomes (which cannot be characterized directly). In particular, a misleading steady state exists with positive probability when agents choose equilibrium sharing strategies if and only if one exists when agents follow this simple strategy.

\subsection{Related Literature}

We begin by describing connections with observational social-learning models. There is a large literature beginning with \cite{banerjee1992simple} and \cite*{bikhchandani1992theory}, and we mention a few relevant threads. Perhaps closest to our model, several papers  assume agents observe a random sample of predecessors' actions (including \cite{banerjee2004word},  \cite{kabos2021welfare}, and \cite*{levy2022stationary}). Our techniques, meanwhile, are based on the same mathematics literature as \cite*{arieli2022aggregate}, who model the distribution of actions taken by agents as a generalized P{\'o}lya urn. Finally, the basic insight that early independent information can improve social learning (Section~\ref{sec:time_varying}) also appears in other settings, including models where a subset of agents act with no observations (\cite{sgroi2002optimizing}, \cite*{peres2020fragile}).

A high-level distinction is that  agents observe signals directly in our model, rather than actions incorporating signals, but the observed signals are endogenously selected. At a  more theoretical level, one can interpret our model as interpolating between a version of observational learning and a case where agents observe unbiased signals (see Section~\ref{sec:discussion}). This leads to several new dynamics relative to the classical herding literature. First, agents learn imperfectly in the long run even without herding-type behavior, so we can quantitatively compare how long-run learning outcome changes across different sampling parameters  (e.g., Proposition \ref{prop:informative}). Such comparisons are key to the main trade-off between more information aggregation and misleading steady states. Second, misleading steady states can persist in our model even when new private information continues to arrive and play a pivotal role in later agents' sharing decisions.\footnote{For example, when agents see the same number of positive and negative signals in the sample  (which happens with positive probability in every steady state when the sample size is even),  they always share the sampled signals that match their private signals.} By contrast, classical results on information cascades rely sharply on later agents' private signals having no impact on any agents' actions.

Our model also relates to a recent literature on learning from shared signals. As we discuss in detail below, existing work focuses on the dissemination of a single signal, or on settings where signals are shared once with network neighbors but not subsequently re-shared.  Our model differs on both of these dimensions. First, we allow many signals about the same
state to circulate simultaneously. These signals interact: a user's social information consists of the multiple signals, so the probability that they share a given signal depends on whether their other signals  corroborate it or contradict it.\footnote{\cite*{jackson2023grapevineWP} study a model with many signals circulating where transmission is exogenous and the main friction instead comes from messages mutating.} Second, we allow signals to be shared widely and shown to many users. The combination of these two model features generates the social version of confirmation bias that we outlined earlier.

% A key feature of our model is that we allow many stories about the same state of nature which can be shared widely through a central platform algorithm that generates news feeds for all users, and this mechanism is responsible for the misleading steady states in our model. Existing work instead focuses  on settings where signals are shared once with network neighbors but not subsequently re-shared, or on the dissemination of a single signal. In our model, a signal can become popular due to early agents' sharing decisions and thus gets pushed into a later agent's news feed, and this later agent can re-share the same signal.  The possibility of re-sharing, together with the existence of many signals (i.e., news stories) that can reinforce or contradict each other, gives the social version of confirmation bias that we outlined earlier. 

\cite*{bowen2021learning} study a model where signals are selectively shared at most once with network neighbors, but agents are misspecified and partially neglect this selection. This bias leads to mislearning,  and it also generates polarization in social networks with echo chambers. By contrast, we focus on rational agents who make endogenous sharing decisions in equilibrium. \cite{bowen2021learning} note that ``the Internet has also brought an abundance of
information, which should lead people to learn quickly and beliefs to converge (not diverge)
according to standard economic models.'' Our results imply that even if people observe a large (but finite) amount of information and rationally account for selection, they can converge to a misleading steady state.

Another group of papers in operations research and economics study settings with ``fake news'' where people decide whether to share a story depending on the outcome of a (possibly noisy) fact check (e.g., \cite{papanastasiou2020fake}, \cite{kranton2020social}, and \cite*{merlino2022debunking}) or depending on their prior beliefs about the story's likelihood (e.g., \cite*{bloch2018rumors}, \cite*{acemoglu2021misinformation} and \cite*{hsu2021persuasion}). Most of these papers consider the diffusion of a single signal that can be re-shared through a network, while \cite{kranton2020social} look at the supply-side decisions of information producers when consumers can share their stories at most once with network neighbors.\footnote{\cite*{merlino2022debunking}'s model features one true and one false message.} We focus on a different dimension of platforms. Instead of asking about the network structure that connects users on the platform (e.g., echo chambers) or fact-checking technologies, we analyze the impact of showing users more viral content. 

% recent operations research literature focuses on the interaction between sharing and fact-checking in settings with ``fake news'': see, for example,  \cite{papanastasiou2020fake}, \cite*{acemoglu2021misinformation}, and \cite*{hsu2021persuasion}.\footnote{Within economics, \cite*{merlino2022optimal} also model fact-checking decisions. Their model features one true and one false message.} Their models consider the diffusion
% of a single signal, while we consider many signals about the same
% state. These signals interact with each other, because observing multiple corroborating stories that support the same state in a news feed can increase the probability that each such story is shared. We also focus on a different dimension of platform-design choices. Instead of asking about the network structure that connects users on the platform (e.g., echo chambers) or fact-checking technologies, we consider the platform's choice in terms of showing its users more or less viral content.

\cite*{buechel2022misinformation}, like our work,  consider an environment where agents can share and re-share copies of a signal. In this model, agents' sharing behavior resembles the DeGroot heuristic. In particular, their agents' sharing is  independent of beliefs, while we study sharing rules that seek to share correct signals and therefore depend on beliefs.

%other papers to consider discussing
%banerjee 1993 rumours
%vielle and co (Stationary social learning in a changing environment, https://arxiv.org/pdf/2201.02122.pdf), and meg meyer (“A Welfare Analysis of a Steady-State Model of Observational Learning” (with Eszter Kabos), https://www.youtube.com/watch?v=QEpFJGMSlCQ) changing state observational learning models with sampling, banerjee and fudenberg 2004
%older kranton paper

\section{Model}\label{sec:model}

We consider a finite society with $n$ agents learning in sequence
about an unknown state of nature $\omega\in\{-1,1\}$. Everyone starts
with the common prior that both states are equally likely. Each agent
receives a binary private signal $s_{i}\in\{-1,1\}$ about the state. Call $s_{i}=-1$ a negative signal and $s_{i}=1$
a positive signal. We assume signals are conditionally independent
and symmetric, so that $\mathbb{P}[s_{i}=-1|\omega=-1]=\mathbb{P}[s_{i}=1|\omega=1]=q$
for some  \emph{signal precision} $0.5<q<1$. We also keep track of the
\emph{popularity score} of each signal $s_{i}$, denoted $\rho(s_{i})$.
Each of these $n$ signals $s_1 ,..., s_n$ starts with a score of 1 when it arrives.

After the state of nature $\omega$ realizes but before any agent acts, we  initialize the society's pool of signals  with a set of $n_0 \ge 2$ binary \emph{seed signals} $s_0^{1}, ..., s_0^{n_0}$.  Conditional on $\omega$, the realizations and initial popularity scores of the $n_0$ seed signals   are drawn i.i.d. from a joint distribution $\psi$ which is invariant under permuting the two states. (The seed signals may start with scores larger than 1.) For example, we could begin by drawing $n_0$ signals of precision $q$ with popularity scores of $1$.

We fix a \emph{sample size} $K$ with $2 \le K \le n_0$ and \emph{capacity} $0<C<K$. Each agent $1 \le i \le n$
sees a sample of $K$ signals from society's current pool of signals.
Agents only observe the realizations of the $K$ sampled signals,
and not their popularity scores or arrival times. Then, agent $i$ \emph{endorses}
$C$ out of the $K$ signals from their sample, increasing
each endorsed signal's popularity score by 1, and their signal $s_i$ is added to the pool. Agent $i$ gets utility
$u>0$ for each endorsed signal that matches the state $\omega.$

A \emph{virality weight }$\lambda\in[0,1]$ determines how $i$ samples from the current pool of signals. Each of the $K$ signals in $i$'s sample has a $\lambda$ chance of being drawn with probabilities proportional to the current popularity scores. With the complementary
probability, it is drawn uniformly at random from the pool of signals. We assume for simplicity that all signals are sampled with
replacement (as we approach the steady state, the effect of replacement
vanishes). All draws are independent.

The sampling rule includes two special cases: 
\begin{enumerate}
\item \textbf{Popularity-based sampling} ($\lambda=1)$:\textbf{ }A signal with twice the popularity score of another has twice the probability of being
sampled. 
\item \textbf{Uniform sampling }($\lambda=0$): Predecessors' endorsement decisions
do not affect sampling. 
\end{enumerate}
More generally, sampling rules with $\lambda$ between zero and one
interpolate between these two cases. The virality weight $\lambda$ measures how much the sampling rule selects more popular signals relative to random signals.

The $n$ agents are uniformly randomly placed into the $n$ positions,
and do not know their positions. Before observing their sample, each agent (correctly)
believes that they are in each position $1,\hdots,n$ with equal probabilities.
 The informational
environment is common knowledge.

\subsection{Discussion and Interpretation}\label{sec:discussion}

We begin by explaining connections to social-media platforms. We then discuss our assumptions about agents' behavior and information and compare our model to sequential social learning settings.

\textbf{Social-Media Interpretation.} Our primary interpretation of the model concerns learning on social-media platforms. Signals $s_i$ correspond to news stories that users discover
from external sources and post on the social-media platform
(e.g., X, Reddit, or Facebook). The samples of previous stories represent news feeds shown to users by a social-media platform. What we generically refer to as ``endorsing'' in our model corresponds
to platform-specific user interactions with content, such as retweeting on X, upvoting on Reddit,  re-sharing friends' posts on Facebook,  and so forth.

Signals arrive exogenously and start with a popularity score of $1$ in our model,  meaning that agents always post the stories they discover. We make this assumption to ensure that new private information continues to arrive and spread on the platform, but could easily adapt our results to other assumptions about the information arrival process.

The news-feed algorithm determines what content is shown to users. It can focus on showing more viral content (larger
$\lambda$) or more ``random'' content (smaller $\lambda$). Displaying
``random'' content could represent, for example, showing a user
the most recent stories that their friends posted without regard for
the stories' popularity score. The virality of the news feed is a design choice that social-media companies devote substantial attention to in practice. Over the years, different iterations of the X (Twitter) feed gave different levels of emphasis to the trending or most popular tweets on the platform. Reddit's ordering algorithm for displaying posts on the front page similarly evolved over many years.\footnote{A 2009 entry on Reddit's company blog discusses tradeoffs in prioritizing more popular comments, including concerns about feedback loops resembling those we will see in our analysis: ``Once a comment gets a few early upvotes, it's moved
to the top. The higher something is listed, the more likely it is to be read (and voted on), and the more votes the comment gets. It's a feedback loop'' \citep{munroe_2009}.}

\textbf{Discussion of Assumptions.} We next discuss several of our assumptions about agents' behavior and information. We assume that agents are rational and want to endorse signals that match the true state. We will see that even under these assumptions, which we view as relatively conducive to learning, there are often misleading steady states. In the context of our social-media application, this assumption is also motivated by empirical evidence on content-sharing behavior. In laboratory experiments, \citet{pennycook2020fighting,pennycook2021shifting} find people have an intrinsic preference for sharing news from more trustworthy sources, which are more likely to accurately
reflect the state. Our analysis is robust to including other agent objectives, provided agents also  care enough about endorsing accurate signals (as discussed in the conclusion).

We assume an explicit capacity constraint $C$ on how many signals people can endorse.  Even if the agents in our model were not forced to endorse exactly $C$ signals out of the $K$ in their sample, they would still find it optimal to always endorse $C$ signals because there is no penalty for endorsing incorrect signals. This improves the model's tractability, as the analysis is considerably cleaner when the number of signals endorsed does not depend on the sample realization. In our social-media application, the capacity constraint captures the fact that people tend to only interact with a small fraction of the content that they consume.

In our model, agents do not see the current popularity scores or the arrival times of the signals in their samples. This assumption is motivated by the difficulty of inferring the state from the popularity or age of observed posts in the social-media application.  We also assume that people do not know their order in the sequence. This is arguably  more realistic than assuming that everyone knows their precise order. From a technical perspective, it is also the more tractable assumption that lets us focus on analyzing long-run
steady states.\footnote{In our model where agents hold a uniform prior over positions, we will be able to analyze changes in the society over time without needing to account for time-varying strategies. If agents knew their positions, strategies and the signal popularities would  both change over time, and even basic convergence properties would be unclear.}

The functional form of the sampling rule has the convenient property that the total popularity scores of positive and negative signals are a sufficient statistic for the distribution of sampled signals,  but other functional forms are also possible. In particular, one could analyze more extreme sampling rules where sampling probabilities depend more heavily on popularity scores than in popularity-based sampling (e.g., the probability of sampling a signal is proportional to a superlinear function of its popularity).

\textbf{Relationship to Observational Learning.} A slight variant of our model clarifies its relationship with sequential social-learning models in which agents observe predecessors' actions. Suppose  new signals start with a popularity score of $0$ (instead of $1$) and suppose $C=1$. Then we can interpret ``endorsing a negative signal'' and ``endorsing a positive signal'' as the two possible actions in an observational social-learning model with binary signals and binary actions. For each of the $K$ observations in the sample, we can think of  our sampling rule as selecting a uniformly random predecessor and observing their binary action with probability $\lambda$ and observing their binary signal with probability $1-\lambda$. When $\lambda =0$, agents observe predecessors' signals in an unbiased way,  so they learn as if they observe $K+1$ private signals. When $\lambda=1$, the model is similar to an observational social-learning model where agents observe $K$ predecessors chosen uniformly at random. There is a positive probability of an information cascade in which all late enough agents choose the incorrect action (as in \cite*{acemoglu2011bayesian}).

Our analysis essentially interpolates between these two cases. In our model, recent signals can always be sampled (since signals start with a positive score) and so agents never converge to choosing the same action. There is nevertheless a form of herding on incorrect actions, which we show emerges discontinuously at an intermediate value of $\lambda$ (Theorem~\ref{thm:equilibrium_ss_by_lambda}). Because both actions will be taken infinitely often, there is also interesting variation in the action distribution as we vary $\lambda$ (see Proposition~\ref{prop:informative}). Finally, by attaching endorsements to specific signals instead of modeling them as binary actions, we can track the popularity evolution of different signals and study the distribution of popularity scores across signals (see Section~\ref{sec:power_law}). 

% We mention the issues with two natural alternative
% approaches to studying steady state behavior of the model. 
% \begin{enumerate}
% \item One approach would be to directly model steady states without modeling
% convergence to those steady states. This creates an indeterminacy
% problem, however. When the social learning process stochastically
% converges to one of multiple steady states, merely analyzing the set
% of possible steady states does not pin down the probabilities of converging
% to different steady states. 
% \item Another approach would be to consider people $i=1,\hdots,n$ acting
% in sequence and assume each person knows their position in the sequence.
% Then as the index $i$ increases, the distribution of shared signals
% and the beliefs of person $i$ given their observations evolve jointly.
% It is not clear how one would prove these two stochastic processes
% converge, as they could cycle jointly. 
% \end{enumerate}
% By assuming that person $i$ does not know their position in the sequence,
% we reduce to a single sharing strategy independent of $i$, as agents
% cannot assign different interpretations to the signals they see based
% on their positions. Then given a sharing strategy, we can consider
% the evolution of the popularity of ``0'' and ``1'' stories. This
% gives a single stochastic process instead of two joint stochastic
% processes, allowing steady state analysis.

\subsection{Strategy, Symmetric BNE, and Limit Equilibrium}

We  define a mixed strategy in the game to be $\sigma:\{-1,1\}\times\{0,...,K\}\to\Delta(\{0,1,...,C\})$,
so that $\sigma(s,k)$ gives the distribution over the number of positive
signals endorsed when the agent has the private signal $s$ and sees a
sample with $k$ positive signals out of $K$.\footnote{Agents cannot distinguish between different positive (or negative) signals in their sample. Moreover, which positive (or negative) signals they endorse does not affect subsequent agents' observations under the family of sampling rules we consider.} We will regard the space of mixed strategies as a subset of $\mathbb{R}^{2(K+1)(C+1)}$ with the standard Euclidean norm. Mixed strategies
must satisfy feasibility constraints in terms of the available numbers
of positive and negative signals to endorse.  For each $0\le k\le K$, the support of $\sigma(s,k)$ is contained in 
\[
\{L_k,L_k+1,\ldots,U_k\},\qquad L_k:=\max\{0,C+k-K\},\quad U_k:=\min\{k,C\},
\] where $L_k$ and $U_k$ are the lowest and highest feasible numbers of positive signals endorsed in a sample with $k$ positive signals. Note that we only need to discuss positive signals
since the agent must always endorse $C$ signals in total.

A simple strategy, which will play a central role in our analysis, is to follow the majority of the signals in the sample as much as possible (breaking ties in favor of the private signal) :
\begin{defn}
The \emph{majority rule} is the pure strategy that endorses the maximal feasible number of signals matching the sample majority. Formally, $\sigma^{\text{maj}}(s,k)(U_k)=1$ if either $k>K/2$ or $k=K/2$ and $s=1$, and $\sigma^{\text{maj}}(s,k)(L_k)=1$ otherwise.
\end{defn}
When $C \le K/2$, the majority rule either endorses $C$ positive signals or $C$ negative signals. For larger $C$, majority rule still always endorses $C$ signals, so it may endorse signals from both sides.  The majority rule need not be an equilibrium strategy in general --- intuitively, it is only optimal when sampled signals are more informative than  private signals. Nevertheless, it turns out the majority rule will help us understand the qualitative properties of equilibrium outcomes even when it is not itself an equilibrium. 

We apply the solution concept of Bayesian Nash equilibrium (BNE).
Note that all possible observations are on-path given any strategy
profile. We focus on player-symmetric and state-symmetric BNE: that
is, a BNE where each agent uses the same strategy $\sigma$, and $\sigma$
treats positive and negative signals symmetrically.\footnote{More precisely, state symmetry means that for every $s\in\{-1,1\}$
and $0\le k\le K,$ we have $\sigma(s,k)(z)=\sigma(-s,K-k)(C-z)$
for each $0\le z\le C.$} We abbreviate this refinement as ``symmetric BNE.''

As we will restrict to state-symmetric strategies throughout, we provide a brief intuition for the implications of this restriction. Asymmetric strategies could allow agents to, for example,  treat positive signals as more informative: they endorse positive signals more than they would endorse negative signals at the corresponding ``mirrored'' observation. If this happens across observations, however, then negative signals would become stronger indicators of the true state, giving a contradiction. We expect that equilibria violating state symmetry, if any exist, would require delicate constructions in which agents treat positive signals as more informative at some observations but negative signals as more informative at others.

We are mainly interested in analyzing the limits of symmetric BNE when
the number of agents in the society grows large and in studying the accuracy of the resulting samples in the long run. Such a limit is well defined because for
fixed parameters $q,K,C,\lambda, n_0, \psi$, the space of strategies stays constant
as the number of agents $n$ grows.
\begin{defn}
For fixed $q,K,C,\lambda, n_0, \psi$ parameters, a mixed strategy $\sigma^{*}$
is a \emph{limit equilibrium} if there exists a sequence of symmetric
BNE $(\sigma^{(j)})_{j=1}^{\infty}$ for finite societies with
$n_{j}$ agents and the same $q,K,C,\lambda, n_0, \psi$ parameters, where $n_{j}\to\infty$ and $\lim_{j\to\infty}\sigma^{(j)}=\sigma^{*}$.
\end{defn}

Symmetric BNE and limit equilibria both exist for all parameter values:

\begin{prop}\label{prop:existence}
For any finite $n$ and parameters $q,K,C,\lambda, n_0, \psi$, there exists
a symmetric BNE. For any parameters $q,K,C,\lambda, n_0, \psi$, there exists
a limit equilibrium.
\end{prop}

The proof shows that there is a symmetric BNE for all $n$ 
via a standard
fixed-point argument. Since the space of feasible mixed strategies can be viewed as a compact subset of a finite-dimensional Euclidean space, a limit equilibrium
must exist.

\section{Steady States and Equilibrium Steady States}

We begin this section by analyzing the evolution of a society under an arbitrary fixed  strategy. We show the society converges to a steady state and distinguish informative and misleading steady states. We then apply this analysis to characterize equilibrium strategies and the structure of steady states under equilibrium behavior. Our main result shows there is at least one misleading steady state at equilibrium if and only if the virality weight is above a threshold level.

\subsection{Definition and Characterization of Steady States}

Suppose everyone uses the same strategy $\sigma$, which need not be an equilibrium, and the true state
is $\omega.$ How will the total popularity score of the correct signals that match the state
compare with that of the incorrect signals in the long run? We define the concept of steady states to study this question.

Given the true state of nature $\omega$, let $\mathcal{P}_t$ denote the pool of signals after $t$ agents have acted (including the $n_0$ seed signals and the $t$ signals added by the agents). Let $\rho_t(s)$ be the popularity score of signal $s \in \mathcal{P}_t$ at that time. The \emph{viral
accuracy} of the society is defined to be $$x(t) = \frac{\sum_{s\in \mathcal{P}_t:s=\omega}\rho_t(s)}{\sum_{s\in \mathcal{P}_t}\rho_t(s)}.$$
Viral accuracy measures the relative popularity of the signals that match
the true state.

Imagine a society with infinitely many agents, with all agents 
using the strategy $\sigma.$ This induces a stochastic process $(x(t))_{t=0}^{\infty}$ where $x(t)$ is the
viral accuracy of the society after  $t$ agents have acted (and $x(0)$ is the initial viral accuracy of the seed signals). We refer to limit points reached by this process as steady states:
\begin{defn}
A point $x^{*}$ such that $x(t)\to x^{*}$ with positive probability is a \emph{steady state} of the strategy $\sigma$.
\end{defn}
When viral accuracy converges to a steady state $x^*$, roughly $x^*$ fraction of the total popularity score in the society
is associated with correct signals in all late enough periods. This fraction persists as new signals arrive in each period and agents use $\sigma$ to decide which signals to endorse from their random samples. The next result tells us that for any state-symmetric strategy, viral
accuracy almost surely converges, and the set of steady states $X^{*}$
is finite.
\begin{prop}
\label{prop:conv}Given a state-symmetric strategy $\sigma$, there is a finite set of steady states $X^{*}\subseteq(0,1)$ such that when all agents use $\sigma$, almost surely $x(t)\to x^{*}$
for some $x^{*}\in X^{*}$. 
\end{prop}
When $X^*$ contains at least two elements, the
limit steady state $x^{*}\in X^{*}$ is random  and  can depend on early signal realizations and the sampling process. The proof of the result, as well as subsequent results, relies on stochastic approximation tools. We largely defer discussion of proof intuitions to Section~\ref{sec:techniques}, where we sketch proofs of the main results and describe the key techniques underlying the analysis.

In light of Proposition \ref{prop:conv}, we write $\pi(\cdot\mid\sigma)$ for the distribution over steady states
generated by a state-symmetric strategy $\sigma$.  A substantial challenge in analyzing our model is that we cannot obtain closed-form expressions for $\pi(\cdot \mid \sigma)$ as these probabilities depend on a complicated stochastic process. We focus instead on understanding when the support of  $\pi(\cdot \mid \sigma)$, which is the set of steady states, contains certain values of $x$. We will see this question is already highly non-trivial, and the answers will have interesting implications for understanding equilibrium.

Our next result will characterize the support of the distribution $\pi(\cdot\mid\sigma)$ over steady states in terms of the fixed points of an \emph{inflow accuracy function}, which is a deterministic approximation of the change in the stochastic viral-accuracy process. 
Suppose today's viral accuracy is $x,$ and exactly $q$ fraction
of the signals in the pool are correct. A new agent increases
the total popularity score in the society by $C+1$, as they add a new signal and endorse
$C$ existing signals. We define the inflow accuracy function  $\phi_{\sigma}(x)$
to be the expected fraction of the incoming $C+1$ popularity score that
is allocated to signals matching the state. 
\begin{defn}
\label{def:inflow_acc}The \emph{inflow accuracy function} is 
\[
\phi_{\sigma}(x):=\frac{q+\sum_{k=0}^{K}P_{k}(x,\lambda)\cdot[q\cdot\mathbb{E}[\sigma(1,k)]+(1-q)\cdot\mathbb{E}[\sigma(-1,k)]]}{1+C}
\]
where $P_{k}(x,\lambda):=\mathbb{P}[\text{Binom}(K,\lambda x+(1-\lambda)q)=k]$
and $\text{Binom}(K,p)$ is the binomial distribution with $K$ trials
and success probability $p$.
\end{defn}
To understand the formula in the definition, note that $\lambda x+(1-\lambda)q$ is the \emph{sampling
accuracy}: the probability of each sampled signal being correct, given viral accuracy $x$ and virality weight
$\lambda$. We can use sampling accuracy to express the probability of getting $k$ positive signals out of $K$ in the sample
when $\omega=1$ for every $0\le k \le K$, then consider how the strategy $\sigma$ combines the private signal $s_i$ and the number of positive
signals in the sample to make an endorsement decision. Finally, $\phi_{\sigma}(x)$ also takes into account that the agent's private signal $s_i$, which starts with a popularity score of 1, has
$q$ chance of matching the state. While $\phi_{\sigma}(x)$ is defined in terms of the expected fraction of the new popularity score assigned to correct signals when $\omega=1$, the symmetry of the environment and of $\sigma$  implies that it also describes the same fraction  when $\omega=-1$.

We always have $\phi_{\sigma}(0)>0$ and $\phi_{\sigma}(1)<1.$ The
idea is that if $x\approx0$ and almost all of the popularity score
are associated with the wrong signals, then the arrival of new signals tends to increase $x$, as a majority of these signals
match the state. If on the other
hand $x\approx1$, then these new signals will on average lower $x,$
since they have a non-zero probability of mismatching the state. So
$\phi_{\sigma}$ must have a fixed point by continuity.

A fixed point of the inflow accuracy function $\phi_{\sigma}$ is
a natural candidate for a steady state induced by $\sigma,$ as it
intuitively represents a level of viral accuracy that tends to be
exactly maintained on average by the inflow of new popularity score,
in a society with sufficiently many signals so that approximately $q$
fraction of them match the true state. The next result establishes
this formally, provided the fixed point is not unstable from
both sides. This lets us extend the argument from the previous paragraph to existence of a steady state: the left-most (right-most) fixed point must be stable on the left (right).
\begin{thm}
\label{thm:ss_half_stable}We have $\pi(x^{*}\mid\sigma)>0$ if $\phi_{\sigma}(x^{*})=x^{*}$ and there exists some
$\epsilon>0$ so that either (a) $\phi_{\sigma}(x)<x$ for all $x\in(x^{*},x^{*}+\epsilon)$,
or (b) $\phi_{\sigma}(x)>x$ for all $x\in(x^{*}-\epsilon,x^{*}).$ Conversely, for $x^{*}\in[0,1],$ we have $\pi(x^{*}\mid\sigma)>0$ only if $\phi_{\sigma}(x^{*})=x^{*}$. 
\end{thm}

\begin{figure}
\begin{centering}
\includegraphics[scale=0.5]{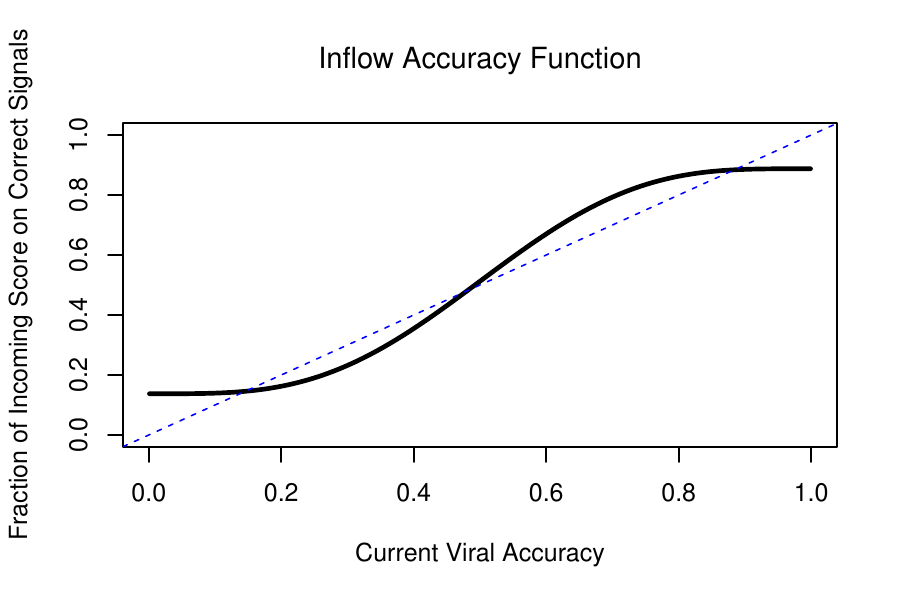}
\par\end{centering}
\caption{\label{fig:maj}The inflow accuracy function for the majority rule
with $K=7,$ $C=3,$ $q=0.55,$ $\lambda=1$.}

\end{figure}

We first discuss fixed points which are stable from both sides. As an example, Figure \ref{fig:maj} plots the inflow accuracy function for the majority rule when $K=7,$ $C=3,$ $q=0.55,$
and $\lambda=1$. There are two fixed points that are stable from both sides, and Theorem \ref{thm:ss_half_stable} implies both are steady states. At the upper fixed point, signals matching the state are more popular. At the lower fixed point, however, incorrect signals are more popular than correct signals. Under the majority rule, such a misleading state is reached with positive probability: if enough initial signals are incorrect, the majority rule will continue endorsing incorrect signals. We will see in Section~\ref{sec:eq_ss} these misleading steady states can also arise under equilibrium behavior.

The more subtle case is a fixed point of $\phi_{\sigma}$ that is unstable from one side (see Figure \ref{fig:touchpoint} for an illustration). A \emph{touchpoint} of $\phi_{\sigma}$ is a fixed point $x^{*}=\phi_{\sigma}(x^{*})$
where exactly one of condition (a) or condition (b) from Theorem \ref{thm:ss_half_stable} holds 
(so $x^{*}$ is only stable from one side). Theorem \ref{thm:ss_half_stable} says that if $\phi_{\sigma}$ has a touchpoint $x^*$, then viral
accuracy converges to $x^*$ with positive probability. It turns out  the convergence of viral accuracy to fixed points (including touchpoints) of its inflow accuracy function  plays a  role in proving the existence of misleading steady states under equilibrium behavior (see Section \ref{sec:equilibrium_intuition}). This convergence also implies that the distribution over steady states is discontinuous in the strategies that agents use and discontinuous in parameters of the model such as $\lambda$ and $q$, and we will discuss some consequences of the discontinuity below.

\begin{figure}
\begin{center}
\includegraphics[scale=0.23]{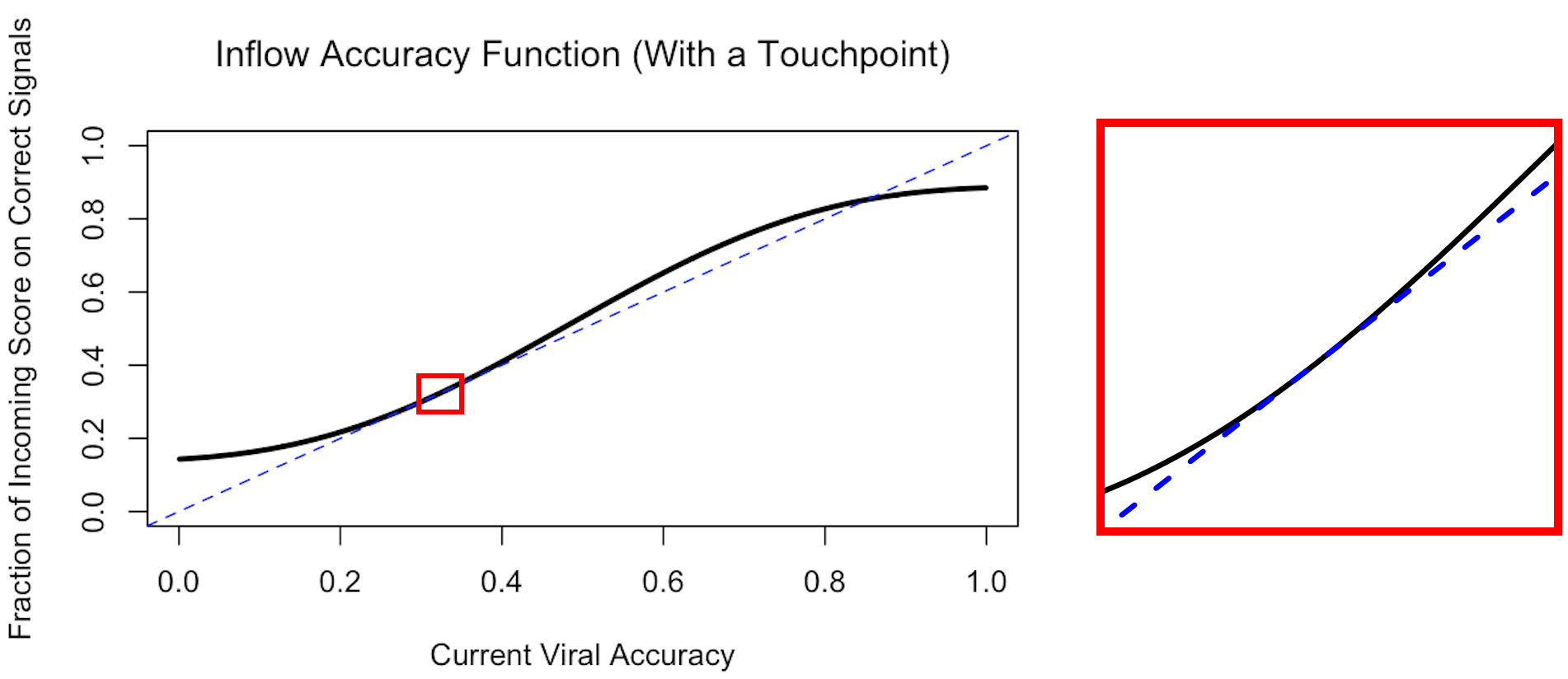}
\end{center}%

\caption{\label{fig:touchpoint}The inflow accuracy function for the majority rule
with $K=7,$ $C=3,$ $q=0.55,$ $\lambda\approx 0.76$. Here $\phi_{\sigma^{\text{maj}}}$ has two fixed points: the left fixed point is a touchpoint that is only stable from the left side (the red box shows a zoomed-in view). The right fixed point is stable from both sides. Theorem \ref{thm:ss_half_stable} implies viral accuracy has a positive probability of converging to each of these two fixed points.}

\end{figure}

\subsection{Informative and Misleading Steady States}

We may classify steady states into two types. One type is an informative
steady state, where sampling accuracy is at least 1/2 and it is more
likely that sampled signals are true. The other type is a
misleading steady state, where the opposite happens.
\begin{defn}
A steady state $x$ is \emph{informative} if $\lambda x+(1-\lambda)q\ge1/2,$
and \emph{strictly informative} if this inequality is strict. A steady
state $x$ is \emph{misleading} if $\lambda x+(1-\lambda)q\le1/2$,
and \emph{strictly misleading} if this inequality is strict.
\end{defn}
Even reasonable strategies like the majority rule $\sigma^{\text{maj}}$  can generate misleading steady
states. Recall from Figure \ref{fig:maj} and the discussion after Theorem~\ref{thm:ss_half_stable} that $\sigma^{\text{maj}}$
has two steady states with the parameters $K=7,$ $C=3,$ $q=0.55,$
and $\lambda=1$. One is informative, but the
other is misleading.

In a misleading steady state, the virality of false signals becomes
self-sustaining. The state might be $\omega=1$ but most people see
negative signals in their samples, as the society's virality weight
implies the popular false signals tend to get shown to agents. This
happens even though there are more positive signals than negative signals in the society. Under the
majority rule $\sigma^{\text{maj}}$, for example, agents will then endorse the negative
signals from their samples, which further perpetuates these signals'
popularity and makes them more likely to be seen by future agents.

We will see that $\phi_{\sigma^{\text{maj}}}$, the inflow accuracy
function associated with the majority rule, plays an important role
in determining the equilibrium steady states of \emph{any} limit 
equilibrium. As a first step in this direction, we observe that the steady states of
the majority rule $\sigma^{\text{maj}}$ satisfy the following useful properties:
\begin{lem}
\label{lem:maj_ss_types}If $x$ is a steady state of $\sigma^{\text{maj}}$,
then it is strictly informative if and only if $x>1/2,$ and strictly
misleading if and only if $x<1/2.$ Also, $x=1/2$ is not a fixed
point of $\phi_{\sigma^{\text{maj}}}$.
\end{lem}
By definition, steady states are  classified as informative or misleading based on their sampling accuracy. The lemma says for the majority rule, we can equivalently classify steady states based on whether viral accuracy is larger than $1/2$. 

The number of  steady states for a fixed strategy depends on $\lambda.$ The
three plots in Figure \ref{fig:inflowplots} show the inflow
accuracy function under majority rule with $q=0.55,$ $K=7,C=3$, and three different virality weights:
$\lambda=0.3,$ $\lambda=0.6$, and $\lambda=0.9$. When $\lambda=0.3$
and $\lambda=0.6$, there is only an informative steady state, and
this steady state is more accurate when $\lambda=0.6.$ But when $\lambda=0.9,$
there is both an informative steady state and a misleading steady
state. (In general, Lemma \ref{lem:maj_shape} in the Appendix proves that the inflow accuracy function of majority rule is always concave or S-shaped as in these plots.)

\begin{figure}
    \centering
\includegraphics[scale=0.58]{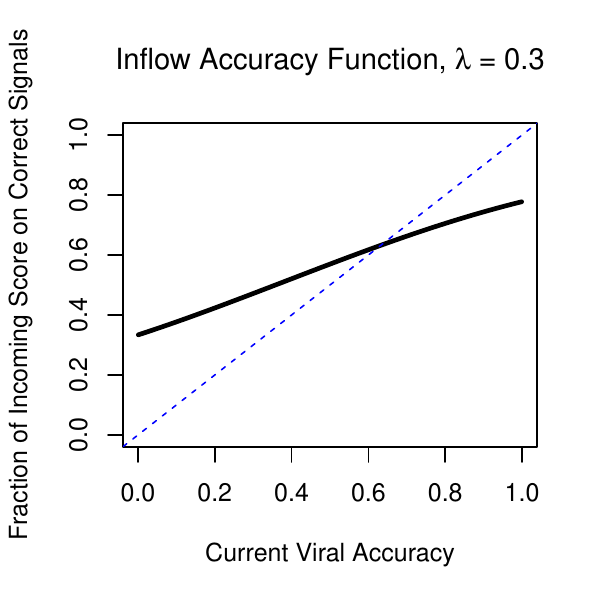}\includegraphics[scale=0.58]{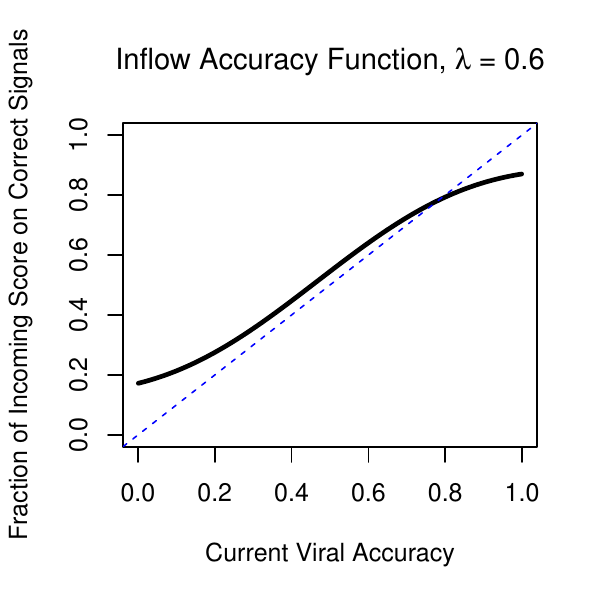}\includegraphics[scale=0.58]{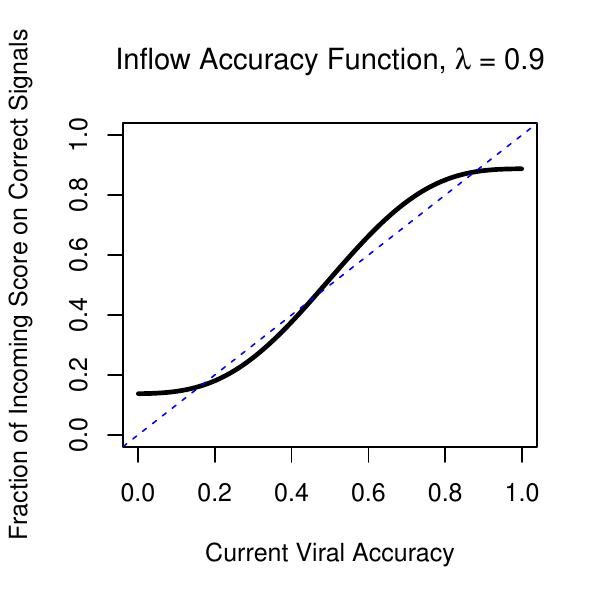}
    \caption{The inflow accuracy function for the majority rule with $q=0.55,$ $K=7,C=3$, and $\lambda \in \{0.3,0.6,0.9\}$. With $\lambda=0.3$ and $\lambda=0.6,$ there is a single informative steady state. With $\lambda=0.9,$ a misleading steady state appears.}
    \label{fig:inflowplots}
\end{figure}

\subsection{Equilibrium Steady States}\label{sec:eq_ss}

So far, we have discussed steady states associated with arbitrary
strategies. We are mainly interested in \emph{equilibrium steady states},
i.e., the distribution $\pi(\cdot\mid\sigma^{*})$ when $\sigma^{*}$
is a limit equilibrium strategy.

We now define the critical virality weight, which is the smallest $\lambda$ for which there is a misleading steady state under the majority rule. We will see the set of equilibrium steady states changes sharply around this critical
value of $\lambda$.
\begin{defn}
The \emph{critical virality weight} $\lambda^{*}$ is $$\lambda^{*}:=\inf\{\lambda\in[0,1]:\phi_{\sigma^{\text{maj}}}(x^*)=x^*\text{ for some } x^* \in [0,1/2]\},$$
provided this set is non-empty. Otherwise, we let $\lambda^{*}=\infty.$
\end{defn}

Depending on the values of the parameters $q,K,C,$ it turns out that
either $\sigma^{\text{maj}}$ only has strictly informative steady
states for any virality weight (so $\lambda^{*}=\infty$), or there is some
smallest $0<\lambda^{*}\le1$ where a fixed point in $[0,1/2]$ first
appears for $\phi_{\sigma^{\text{maj}}}.$ For instance, for $q=0.55, K=7, C=3$, Figure \ref{fig:touchpoint} shows that $\lambda^* \approx 0.76$. 

The next theorem characterizes  which values of $\lambda$ induce a misleading equilibrium steady state. The answer is independent of the selection of the limit equilibrium.

\begin{thm}
\label{thm:equilibrium_ss_by_lambda}For $0<\lambda\le\lambda^{*}$,
the unique limit equilibrium is $\sigma^{\text{maj}}$. At every $\lambda < \lambda^*$, $\sigma^{\text{maj}}$
only has one equilibrium steady state, and it is strictly higher than
$q$ (thus, strictly informative). For $\lambda\ge\lambda^{*}$, 
every limit equilibrium induces at least one strictly misleading
steady state.

\end{thm}
This result shows how  virality weight affects the
types of equilibrium steady states: there are only informative equilibrium
steady states when $\lambda<\lambda^{*},$ while there will always
be misleading equilibrium steady states when $\lambda\ge\lambda^{*}$.\footnote{It is easy to check that misleading steady states can indeed arise. In fact, one can show that the threshold $\lambda^*$ is finite whenever $K$ and $C$ are large enough.}
It also shows the majority rule is the only possible limit equilibrium
for non-zero virality weights below the critical virality weight $\lambda^{*}$.\footnote{When $\lambda=0,$ the only possible equilibrium steady state is
$q$, so every sampled signal is exactly as informative as
one's private signal. Equilibrium may not be unique since there is some degree of freedom in tie-breaking: for example, 
when there are $k+1$ positive signals and $k$ negative signals in the sample and one's private signal is negative, the agent is indifferent between endorsing any signal.} For virality weights above $\lambda^{*}$, the majority rule may not be a limit equilibrium,
and there may be multiple limit equilibria. Nevertheless, the result
tells us that every limit equilibrium has a positive probability
of generating a misleading steady state where false signals dominate
samples. Agents are aware of the possibility of a misleading steady state and would like to account for it, but are unsure whether society converged to a misleading steady state or an informative one.

The theorem greatly simplifies checking whether there is a misleading steady state at a limit equilibrium under given parameter values. Without the theorem, checking for misleading steady states would require solving for equilibrium strategies, which is a complicated calculation depending on $\pi(\cdot|\sigma)$ and therefore the entire stochastic process. The theorem says we can instead check for misleading steady states under the majority rule, which are simply roots of a polynomial.

Finally, Theorem \ref{thm:ss_half_stable} and Theorem \ref{thm:equilibrium_ss_by_lambda} together imply a discontinuity in equilibrium learning outcomes at $\lambda^*$. For $\lambda$ just below $\lambda^*$, we converge almost surely to a steady state where a majority of agents believe the true state is more likely. But when $\lambda=\lambda^*$, there is a positive probability of converging to a misleading steady state. The expected accuracy under the limit equilibrium strategy $\sigma^{\text{maj}}$ also discontinuously drops at $\lambda^*$:
\begin{cor}\label{cor:discontinuity}The expected steady-state viral accuracy under the unique limit equilibrium jumps downward at $\lambda^*$:
$\lim_{\lambda \rightarrow (\lambda^*)^- }\mathbb{E}_{\pi(x^* \mid \sigma^{\text{maj}},\lambda)} [x^*] >  \mathbb{E}_{\pi(x^* \mid \sigma^{\text{maj}},\lambda^*)} [x^*].$
\end{cor}
In our model, expected viral accuracy measures expected social welfare, since agents derive utility from endorsing correct signals and viral accuracy measures (up to an affine transformation) the fraction of correct endorsements. The corollary implies that in our social-media application,  outcomes and user welfare can  be very sensitive to design choices or to attempts at influencing platforms, such as misinformation campaigns. %reference future section on design choices

We now turn to the benefits of higher virality weight $\lambda$. Our next result says that a larger virality weight can lead to a more accurate informative steady state:
\begin{prop}
\label{prop:informative} For $\lambda\in(0,\lambda^{*})$, the unique
steady state $x^*$ at the unique limit equilibrium
under $\lambda$ is strictly increasing in $\lambda.$ 
\end{prop}
In the region of virality weights that do not generate misleading equilibrium steady states, increasing $\lambda$ allows
more information aggregation. This is because a positive signal in $i$'s sample not only tells $i$ about the realization of a single signal, 
but also lets $i$ draw inferences about the hidden information available to $i$'s predecessors who may have chosen
to endorse that positive signal. As $\lambda$ increases, the sampled signals
are closer to indicating a consensus among many agents.

Theorem \ref{thm:equilibrium_ss_by_lambda} and Proposition \ref{prop:informative} together formalize the trade-off in the virality weight $\lambda$ described in the introduction. Increasing $\lambda$ initially increases the steady-state viral accuracy and sampling accuracy. But starting at a critical threshold $\lambda^*$, it discontinuously creates the social form of confirmation bias discussed in the introduction, which we have now formalized in terms of a misleading steady state.

Finally, we note that our analysis has focused on viral accuracy rather than agents' beliefs. One can obtain a lower bound on how many agents hold directionally correct beliefs at a misleading steady state: an agent will only endorse an incorrect signal if they believe the true state has probability at most 50\% or if there are not enough correct signals to endorse. The latter possibility somewhat complicates this analysis, however. We therefore omit the details, but do provide numerical evidence of wrong beliefs in the next subsection.

\subsection{Numerical Illustrations of Equilibrium for $\lambda>\lambda^*$}

Theorem \ref{thm:equilibrium_ss_by_lambda} determines the unique equilibrium for $\lambda \leq \lambda^*$ but does not give a complete characterization when $\lambda>\lambda^*$. We describe a set of numerical simulations that calculate equilibrium in this region and discuss long-run viral accuracy and beliefs. We find sizable jumps in the probability of the misleading steady state, the expected welfare, and expected belief in the true state when virality weight crosses $\lambda^*$. As $\lambda$ increases above $\lambda^*$,  agents rely more on their private signals in response to the possibility of misleading steady states, but nevertheless converge quite often to misleading steady states  where they hold fairly strong wrong beliefs. Another insight is that behavioral responses to increasing the virality weight can be larger than  the mechanical effects. So the equilibrium implications of changing the virality weight may be the opposite of those in a model where agents exogenously use majority rule as a fixed strategy.

We numerically estimate equilibria in an example with signal precision $q=0.55$,  capacity $C=3$, and sample sizes $K\in\{6,8,10\}$. We suppose there are $K$ seed signals each matching the state with probability $q$ and each with an initial popularity score of 1. Our simulation includes all virality weights higher than the respective $\lambda^*$ in a grid of width 0.02, $\lambda \in \{1, 0.98, 0.96, 0.94, ...\}$. For each $K$ and $\lambda$, we first check for a pure-strategy equilibrium. When there is no pure-strategy equilibrium, we check for a mixed-strategy equilibrium of the following form: agents use the majority rule except when they see a sample with $k$ signals on one side and $k-2$ signals on the other. In that case, they will follow their private signal with some probability $0<p<1$ and follow the sample majority with probability $1-p$ (see details in Appendix \ref{subsec:simulation_details}). For all parameters, this procedure finds a single limit equilibrium, which we describe below. The simulation results can be found in Figure \ref{fig:combined} and Appendix Table \ref{tab:eqm-all}.

\begin{figure}
\centering
\includegraphics[width=.96\linewidth]{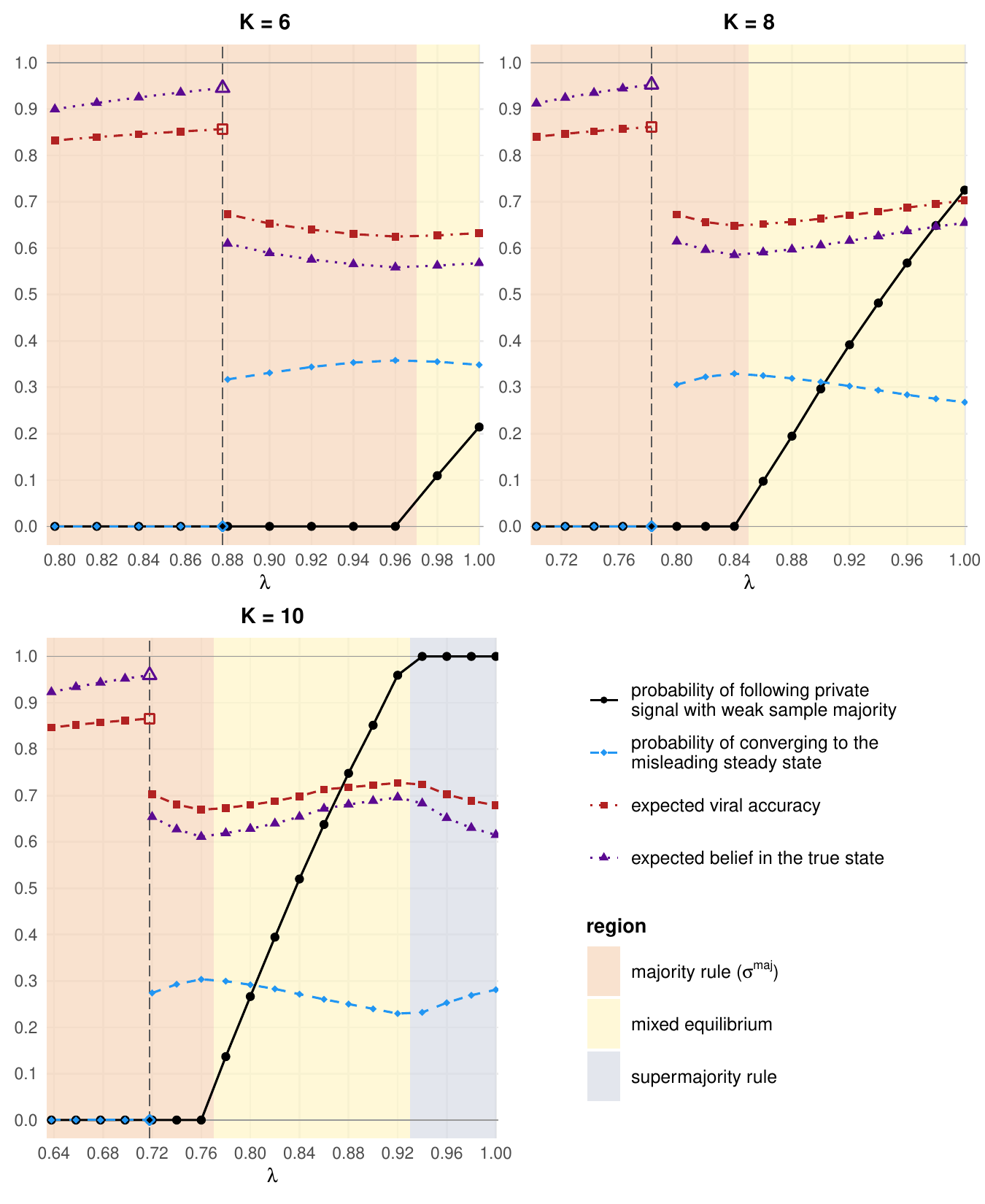}
\caption{Equilibrium quantities as functions of $\lambda$ around the critical virality weight $\lambda^*_K$, for sample sizes $K=6,8,10$. Parameters $q=0.55$, $C=3$ are fixed throughout. The dashed vertical line in each panel marks $\lambda_K^*$. Each panel shows four curves: the probability $p$  of following private signal upon observing a weak sample majority ($k+2$ signals on one side and $k$ signals on the other), the probability of converging to the misleading steady state, the expected steady-state viral accuracy $E[x^{*}]$, and the expected belief in the true state. Background shading marks the equilibrium regime: peach = $\sigma^{\mathrm{maj}}$ (majority rule), yellow = mixed equilibrium, light blue = supermajority rule.}
\label{fig:combined}
\end{figure}

We find three equilibrium regimes for $\lambda>\lambda^*$. For $\lambda$ slightly above $\lambda^*$, the majority rule remains a limit equilibrium. As $\lambda$ continues increasing, there is a second regime where the limit equilibrium is a mixed strategy of the form described above for some interior mixing probability $0<p<1$, with $p$ increasing roughly linearly in $\lambda$. Finally, for $K=10$,  there is a third regime of $\lambda$ near 1 where agents use a pure strategy we call the \emph{supermajority rule}:  if there are at least four more signals on the majority side than the minority side in the sample, then follow the majority; otherwise, follow the private signal. (The degenerate cases of $p=0$ and $p=1$ correspond to majority rule and supermajority rule.)

The intuition behind these equilibria is that the possibility of society being stuck in a misleading steady state makes sampled
signals less informative about the state of nature. When $\lambda$ is not much higher than $\lambda^{*}$, sampled signals are strictly
more informative than private signals, making the majority rule optimal. But for larger $\lambda$, 
a sample
of $k$ positive signals and $k-2$ negative ones 
is either exactly as informative as one positive private signal (in the second regime) or strictly less informative
(in the third regime for $K=10$).

%The probability of converging to the informative steady state and the expected total social welfare (i.e., the expected viral accuracy of the steady state) decrease in $\lambda$ in the pure-strategy regimes but increase in $\lambda$ in the mix-strategy regime. 

We also estimate equilibrium beliefs given each possible observation. Figure \ref{fig:belief_dist} shows the distributions of beliefs in the true state given sampled signals for $K=10, \lambda=1$, conditional on the informative and misleading steady states. The two conditional distributions are almost non-overlapping and the steady state almost fully determines whether an agent will have more than 50\% belief in the true state.\footnote{Figure \ref{fig:belief_dist} shows beliefs based only on sampled signals and not private signals for visual clarity, but we include the posterior beliefs based on all available information in Appendix Figure \ref{fig:belief_dist_with_signal}.}

\begin{figure}
\centering
\includegraphics[width=\linewidth, trim={0 0 0 2cm}, clip]{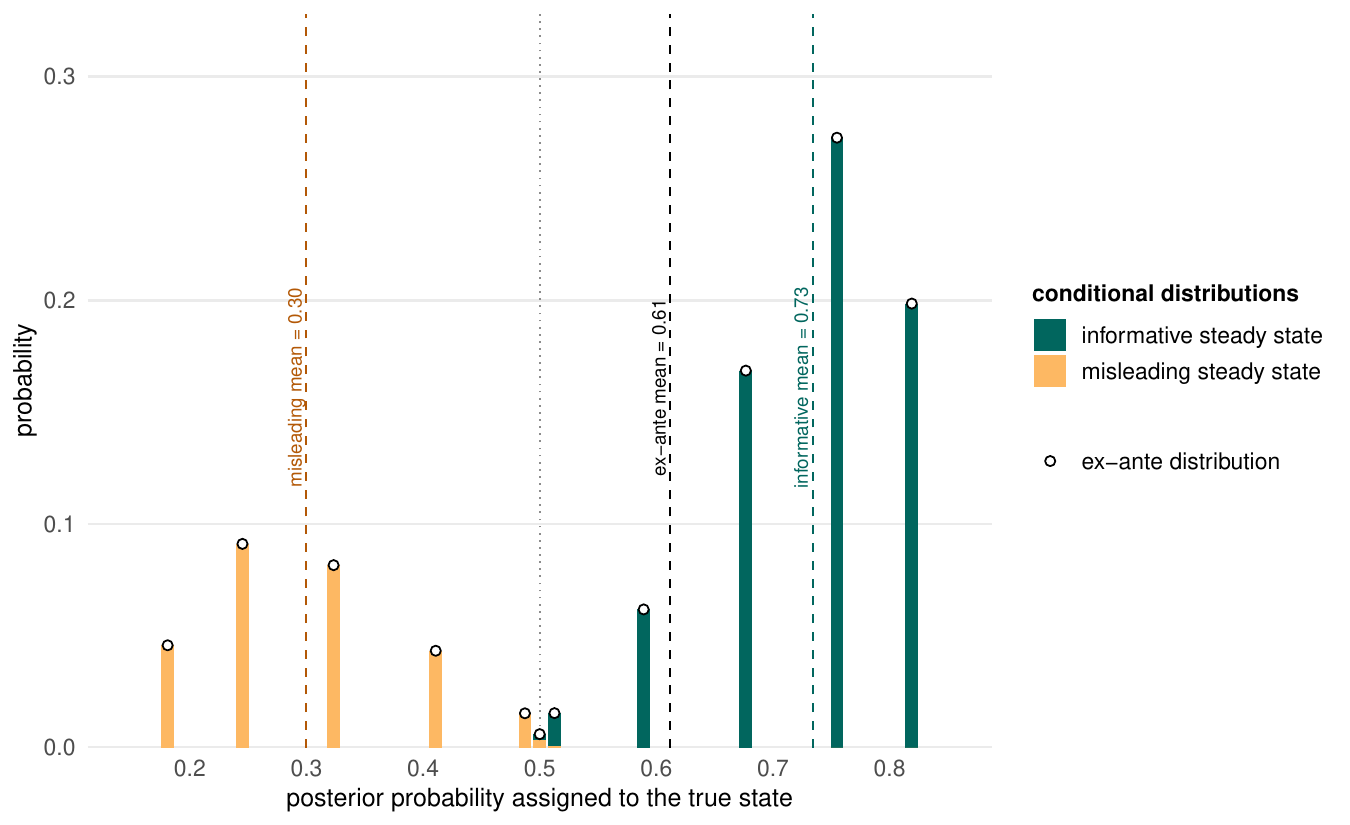}
\caption{Total bar height shows the ex-ante distribution of beliefs in the true state based only on sampled signals, for $K=10$ and $\lambda=1$. Each color is the conditional distribution of beliefs given one steady state. The two conditional distributions are stacked but they are almost non-overlapping.}
\label{fig:belief_dist}
\end{figure}

We make three observations about the simulation results: 
\begin{enumerate}
    \item Society converges to the misleading steady state at a substantial rate (more than 22\% of the time for every $K$ and every $\lambda>\lambda^*$  in the simulation). Moreover, there is a large jump in the probability of converging to the misleading steady state  as $\lambda$ crosses $\lambda^*$. This generates a large downward jump in expected welfare and expected belief in the true state. 
    
    \item When the steady state is misleading, agents form incorrect beliefs: agents have less than $38\%$ expected posterior belief in the true state of nature after observing the sample for all parameters in the simulation. Moreover, for larger $K$, almost all agents have directionally wrong beliefs about which state is more likely to be the true state of nature  (see Figure~\ref{fig:belief_dist} and Appendix Figure~\ref{fig:belief_dist_with_signal}).
    
    \item The mechanical effects of increasing $\lambda$ on welfare and beliefs can be overturned by behavioral responses. For a fixed strategy, higher $\lambda$ increases the probability of converging to the misleading steady state, thus lowering welfare and decreasing expected belief in the true state. But in the mixed-strategy regions, agents respond to higher $\lambda$ by relying more on their private signals. This behavioral effect is larger than the mechanical effect: in each of the three mixed-strategy regions, increasing $\lambda$ increases  expected welfare and expected belief in the true state.
    
\end{enumerate}

\section{Consequences for Platforms}

Before turning to discussing the proofs of our main results, we provide two consequences for the application to social media. First, we solve for an empirically observable object, the distribution of popularity scores across signals. Second, we provide a simple policy that guarantees higher accuracy by changing the virality weight $\lambda$ over time.

\subsection{Popularity Distributions}\label{sec:power_law}

The steady-state structure of signals in the society lets us derive the distribution of popularity scores. In our social-media application, this gives predictions about an empirically observable object: the distribution of shares (or likes) across posts on a platform. This section solves for the popularity distribution of signals, and we find that this distribution has a power-law tail.

We describe the distribution of popularity under a strategy $\sigma$, and it is without loss to assume that $\omega=1$. This distribution will naturally depend on the steady state the society reaches, and we characterize popularity scores conditional on viral accuracy converging to a steady state $x^*$. We will see that the distribution only depends on the strategy $\sigma$ through the steady state $x^*$.

The characterization requires a few definitions. Define 
\begin{equation}
m_{+}:=(1+C)x^{*}-q\mbox{ and } m_{-}:=C+q-(1+C)x^{*}\label{eq:mplus-mminus}
\end{equation}
to be the asymptotic expected number of positive and negative signals
endorsed in each period.
 Let $q_{+}=q$ and $q_{-}=1-q$ be the probabilities of positive and negative private signals, respectively. The distribution of popularity will depend on $x^*$ through the quantities \begin{equation} \delta:=\frac{(1-\lambda)(1+C)}{\lambda} \text{ and }
\gamma_{\theta}:=2+\frac{q_{\theta}(1+\delta)}{m_{\theta}},\label{eq:gamma-def}
\end{equation}
where $\theta \in \{+,-\}.$
\begin{prop}\label{prop:distributions}
Suppose $0<\lambda\le1.$ Suppose agents choose a state-symmetric strategy $\sigma$. Condition
on $\omega=1$ and on viral accuracy converging to a steady state $x^{*}$ of $\sigma$. Then, as $t\to\infty$, the empirical distribution of scores within
each signal type $\theta \in \{+,-\}$ converges (in probability) to the distribution
\begin{equation}
g_{\theta}(r)=(\gamma_{\theta}-1)\,\frac{\Gamma(\gamma_{\theta}+\delta)}{\Gamma(1+\delta)}\cdot\frac{\Gamma(r+\delta)}{\Gamma(r+\gamma_{\theta}+\delta)},\qquad r=1,2,\dots\label{eq:gtheta}
\end{equation}
where $\Gamma$ is the gamma function. The limit distribution has
a power-law tail: there is a constant $c_{\theta}$ such
that $g_{\theta}(r)\sim c_{\theta}\,r^{-\gamma_{\theta}}$ for $r$ large.
\end{prop}

At steady state, agents' observations and therefore their actions are drawn from a fixed distribution. We show that this implies the popularity distribution converges over time to a stationary distribution. The proof pins down this stationary distribution by analyzing a variant of a preferential attachment process.

As long as $\lambda>0$ the proposition predicts a power-law tail for the popularity distribution of signals of each type, with shape parameters that depend on the steady state $x^*$ (and therefore on agent behavior). A number of empirical papers argue the distributions of numbers of shares on social-media platforms fit power-law distributions. See, for example, \cite*{kwak2010twitter} for an early analysis and \cite{garg2025political} for more recent work using maximum-likelihood based tests.

The thickness of the tail distributions is determined by the parameter $\gamma_{\theta}$, which is just a rational function of $x^*$ and model parameters. This lets us compare the distributions of popularity scores for positive and negative signals.

\begin{cor}\label{cor:tails}
If $x^{*}>q$, then $\gamma_{+}<\gamma_{-}$. If $x^{*}<q$, then
$\gamma_{+}>\gamma_{-}$. 
\end{cor}

When the true state is $\omega=1$, positive signals have a heavier tail if $x^* > q$ and negative signals have a thicker tail if $x^*<q$. In particular, this means that the distribution of correct signals has a heavier
tail than the distribution of incorrect signals in the
unique equilibrium steady state when $\lambda<\lambda^{*}$, since
we know from Theorem \ref{thm:equilibrium_ss_by_lambda} that this steady state must be strictly higher than $q$. It also
implies that in every misleading steady state, the popularity distribution
is more heavy-tailed among the incorrect signals than the correct
signals.

\subsection{Changing Virality Weight Over Time}\label{sec:time_varying}

%In our primary application to social media, are there platform policies that can  improve accuracy of learning or eliminate misleading steady states? The question is potentially relevant to regulators, as well as to platforms facing constraints on accuracy due to regulation or public pressure. 

%We first show that if the platform can change the virality weight $\lambda$ over time, then a simple policy can guarantee a high viral accuracy in large societies. We then discuss when interventions to change behavior can eliminate misleading steady states.

Our model  found a basic trade-off in virality weight: higher  weights can improve the accuracy of informative steady states but can also lead to misleading steady states. A natural question is whether sampling rules outside of the class we considered can circumvent this trade-off and improve accuracy. We now show that a simple modification can do so: letting the virality weight change over time. In our main application to social-media platforms, this corresponds to generating news feeds with different algorithms when an issue first emerges and after the discussion has developed further.

To formalize this, let $\overline{x}$ be the viral accuracy at the informative steady state under the majority rule when $\lambda=1$. We will show that $\overline{x}$ is the best possible steady-state viral accuracy that can be reached with positive probability with any (not necessarily equilibrium) strategy and any \emph{fixed} virality weight $\lambda$. Furthermore, provided $\lambda^* < 1$, Theorem \ref{thm:equilibrium_ss_by_lambda} tells us that any limit equilibrium and any $\lambda$ that attain the steady state  $\overline{x}$ with positive probability also generate a misleading steady state with positive probability.  The following result says if the platform can start with $\lambda = 0$ and switch to $\lambda = 1$ after some time, then there is a limit equilibrium where viral accuracy in later periods is arbitrarily close to  $\overline{x}$ with probability arbitrarily close to 1. Compared to any limit equilibrium for any fixed $\lambda$, this generates strictly higher expected viral accuracy when society is large enough.

\begin{prop}
\label{prop:changing_lambda}Suppose $\lambda=0$ for the first $t_{0}(n)$
periods and then $\lambda=1$ for all subsequent periods. We can choose
a sequence of $t_{0}(n)$ such that $\sigma^{\text{maj}}$ is an equilibrium
for $n$ sufficiently large, and given strategy $\sigma^{\text{maj}}$
we have $x(n)\to\overline{x}$ in probability (as the number of agents $n \rightarrow \infty$). The viral accuracy $\overline{x}$ is the highest viral accuracy at any steady state given any fixed $\lambda$ and any state-symmetric strategy.
\end{prop}

If the platform were to set $\lambda=1$ in all periods, then there is some strictly positive probability that society converges to a misleading steady state. The key idea is that the platform can make the probability of this bad outcome arbitrarily small by showing random news feeds ($\lambda = 0$) to a large enough number of agents in the early periods (with the length of this $\lambda=0$ regime fixed ex-ante). These early agents have a very high probability of generating a viral accuracy higher than $1/2$ because the low virality weight lets independent information accumulate.  The platform then exploits this favorable initial condition and switches to  $\lambda = 1$ to generate stronger (and very likely correct) beliefs.

The proposition shows that by modifying news-feed algorithms to be dynamic in a simple way can improve accuracy. This policy could provide guidance for regulators concerned about content accuracy on platforms. Regulation could, for example, limit how much platforms can show viral content early in the discussion of an issue but ease these restrictions after some time. The policy is also a potential approach for platforms facing regulatory constraints requiring some level of accuracy.

Real-world platforms of course need not be limited to sampling rules from the particular class we consider in this paper (time-varying or otherwise). The proposition demonstrates that showing less viral content early in a learning process and more viral content later can improve accuracy, and we expect this dynamic would extend to other sampling rules. We note that platforms could obtain very high accuracy by learning the true state and then showing users exclusively stories matching the state. This may be a helpful policy in some situations, but may be controversial or difficult to implement in others. Our Proposition \ref{prop:changing_lambda} suggests that high accuracy may also be attainable with carefully designed \emph{content-neutral} algorithms, where sampling probabilities only depend on the signals' popularity scores and not on which state of nature they support.

\section{Overview of Techniques}\label{sec:techniques}

In this section, we provide an overview of the proofs of our main results and the underlying techniques. This includes introducing relevant results from stochastic approximation, which may be unfamiliar to some readers.

In time period $t$, there are $n_0+t$ signals in the society, including the $n_0$ seed signals and the $t$ signals contributed by agents. As the characterization in Section~\ref{sec:power_law} indicates, the popularity distributions are quite complicated objects. A  convenient property of our model is that the distribution of observations for a period $t+1$ agent only depends on two numbers. The first is the viral accuracy $x(t)$, which determines the distribution of signals when they are sampled according to their popularity scores. The second is the fraction $z(t)$ of  signals in the current pool  that match the true state, which determines the distribution of signals when they are sampled uniformly at random. As a result, the period $t+1$ values $x(t+1)$ and $z(t+1)$  only depend  on $x(t)$ and $z(t)$.

Our analysis uses stochastic approximation techniques to analyze the two-dimensional stochastic process $(x(t),z(t))$. The viral accuracy $x(t)$ is the main object of study. The fraction $z(t)$ of state-matching signals  in the pool converges to the signal precision $q$ by the law of large numbers, and we can also bound the rate of this convergence with standard methods. We therefore omit most of the details on the analysis of $z(t)$ here and refer the interested reader to the proofs.

The stochastic approximation techniques used here are similar to a literature on generalized P\'{o}lya urns, but our setting differs in two respects. First, the relevant stochastic process in our model is two-dimensional since we keep track of both the fraction of signals that match the true state and the viral accuracy. Second, signals are endorsed in correlated groups of $C$ signals rather than one at a time.

\subsection{Convergence}

We begin by sketching the proof of Proposition~\ref{prop:conv}, which fixed a state-symmetric strategy $\sigma$ and states that viral accuracy converges to a steady-state value $x^*$. The basic idea is to decompose the changes $(x(t+1)-x(t),z(t+1)-z(t))$ in the state of the system into a deterministic term and a martingale noise term.

To formalize this, let $\xi(t+1)$ be the two-dimensional random vector  whose first coordinate is the fraction of the $C+1$ popularity points added in period $t+1$ that accrue to correct signals (including the new signal's initial popularity point)  and whose second coordinate is a binary indicator for whether the new signal $s_{t+1}$ matches the state. This random vector captures  changes to the system state. If $S_t$ is total popularity score after $t$ agents have acted and $N_t=n_0+t$ is the pool size, then
\[
x(t+1)=x(t)+\frac{C+1}{S_t+C+1}(\xi_1(t+1)-x(t)) \text{ and }
z(t+1)=z(t)+\frac{1}{N_t+1}(\xi_2(t+1)-z(t)).
\]
We express the difference  $\xi(t+1)-(x(t),z(t))$ as the sum of a deterministic function
$$h(x(t),z(t)) = \mathbb{E}[\xi(t+1) \mid x(t), z(t)] - (x(t),z(t))$$
and a martingale difference term $$M(t+1) = \xi(t+1) - \mathbb{E}[\xi(t+1) \mid x(t),z(t)].$$

Given such a decomposition, a stochastic approximation result gives sufficient conditions for the discrete-time stochastic process $(x(t),z(t))$ to eventually be well approximated by the continuous-time (deterministic) differential equation $\dot{\mathbf{r}}(t) = h(\mathbf{r}(t))$ in $\mathbb{R}^2$. The relevant conditions are: $h$ is Lipschitz continuous, the system is suitably bounded, and that the changes to the state between periods become smaller at an appropriate rate. The final condition holds because $(x(t),z(t))$ are averages across all signals, and a single agent's impact on these averages vanishes as $t$ grows. The  intuition behind the stochastic approximation result is that eventually the martingale terms  average out and the deterministic part $h$ dominates in deciding the system's trajectory. 

Once we know that $(x(t),z(t))$ is eventually well-approximated by our continuous-time differential equation, the two possibilities are convergence and cycling. The final step is to rule out cycles in the solution $\mathbf{r}(t)$ to the differential equation. The second coordinate must converge to $q$ by the law of large numbers. The remaining dynamics are only in one dimension and therefore cannot generate cycles. We conclude $x(t)$ converges almost surely.

\subsection{Steady States}

After establishing that the state of the system $(x(t),z(t))$ must converge, we ask which points are the limits. Any limit reached with positive probability must be an equilibrium point of the continuous-time differential equation $\dot{\mathbf{r}}(t) = h(\mathbf{r}(t))=0.$
Since $z(t) \rightarrow q$ and the first coordinate of $h(x,q)$ is $\phi_{\sigma}(x)-x$, these equilibrium points correspond to fixed points of the inflow accuracy function. In economics, stochastic approximation techniques are often used in systems with a single fixed point (for each strategy), but a complication in our setting is that the inflow accuracy function can have multiple fixed points.

Theorem~\ref{thm:ss_half_stable} shows that a fixed point is reached with positive probability if it is stable on at least one side. The proof treats the cases of one-sided and two-sided stability separately, and both adapt results on generalized  P\'{o}lya  urns.\footnote{Another result on generalized P\'{o}lya  urns shows that a fixed point that is not stable on either side cannot be a steady state. Since the proof is surprisingly involved and we do not require this result, we omit its likely generalization.}

The case of a fixed point $x^*$ that is stable on both sides is easier and relies on a simple idea. Whether the stochastic process reaches $x^*$ with positive probability depends only on its local structure near $x^*$. Two-sided stability implies that we can replace the inflow accuracy function $\phi_{\sigma}(x)$, which can have multiple fixed points, with an alternate function $\widetilde{\phi}_{\sigma}(x)$ that matches $\phi_{\sigma}(x)$ in a neighborhood of $x^*$ but only has a single fixed point. The alternate stochastic process must converge to its sole fixed point $x^*$ with probability one, so the original stochastic process must also converge to $x^*$ with positive probability.

The touchpoint (one-sided stability) case is more subtle. One might expect that the stochastic process of viral accuracy should not converge with positive probability to fixed points of $\phi_{\sigma}(x)$ which are unstable from one side, because  random noise in the  process $x(t)$ can bring it to the unstable side of the fixed point and cause it to drift away from the fixed point subsequently. But careful analysis shows that there is a positive probability event where $x(t)$ converges to the touchpoint $x^*$ while always staying on the stable side, with the noise terms never large enough to move the process over  to the unstable side. This is because these noise terms vanish more quickly than the deterministic part $h(\cdot)$ of the stochastic process pushes $x(t)$ toward $x^*$. The proof extends the techniques of \cite{pemantle1991touchpoints}, which shows a similar result for generalized P{\'o}lya urns.

\subsection{Equilibrium}
\label{sec:equilibrium_intuition}

Our proof sketches so far have described behavior under a fixed strategy, but our ultimate goal is to understand equilibrium behavior. We first describe how a steady-state analysis simplifies agents' inference problems and then sketch the proof of Theorem~\ref{thm:equilibrium_ss_by_lambda}.

The non-trivial part of behavior is how agents form beliefs about the state of nature. This is a complicated process in general: an agent must calculate the conditional probability of their observations in each state over all possible positions $t \in \{1,\hdots,n\}$ and all possible realizations of the stochastic process $(x(t),z(t))$. Fortunately, the situation is considerably improved by the convergence of $x(t)$ to a steady state $x^*$ and $z(t)$ to $q$. Intuitively, when we take $n$ to be large, we can approximate an agent's beliefs by assuming that $x(t)$ is equal to a steady-state value drawn from the distribution $\pi(\cdot \mid \sigma)$ and that $z(t)=q$. Under this approximation, the agent's inference problem becomes more manageable. There remains, however, a major obstacle (which is standard in applications of stochastic approximation tools): there is no closed-form solution for the distribution over steady states $\pi(x^* \mid \sigma)$.

Fully describing equilibria is therefore likely intractable, but we can nevertheless  say a fair amount qualitatively. The key to this is Theorem~\ref{thm:equilibrium_ss_by_lambda}, which says that whether there is a misleading steady state under a limit equilibrium is equivalent to whether there is a misleading steady state under the majority-rule strategy (and the same parameters). Proving this reduction consists primarily of establishing two properties of the majority rule.

Property (1) is that if any strategy (from a broad class which must contain all best responses) sustains a misleading steady state, then the majority rule does too.\footnote{This makes use of the capacity constraint model of endorsements. If the number of signals endorsed depended on the realization of the sampled signals, it seems plausible there could be a misleading steady state under a strategy that sometimes endorses fewer signals than the majority rule but no misleading steady state under the majority rule.} Essentially, the majority rule puts as much weight as possible on social information and as little weight as possible on private information. This is ideal for sustaining misleading steady states: agents will endorse wrong signals when their social information is wrong.

Property (2) is that when there are no misleading steady states, the majority rule $\sigma^{\text{maj}}$ is the best response if the number of agents $n$ in the society is sufficiently large. As described above, in the large $n$ limit we can approximate agents' inferences with a steady-state inference problem. (Formalizing this requires some additional stochastic approximation arguments to bound the speed of convergence to the steady state.) Conditional on being at a steady state with viral accuracy $x^*$, a sampled signal is a binary symmetric signal of precision $\lambda x^* + (1-\lambda)q$. So an agent observes $K$ binary symmetric signals of the same (potentially unknown) precision. We show that in the absence of misleading steady states, all of these precisions must be greater than $q$. The sampled signals are more informative than an agent's private signal, so $\sigma^{\text{maj}}$ is optimal.

Given these properties,  establishing our reduction is straightforward. When $\lambda < \lambda^*,$ so that the majority rule does not induce a misleading steady state,  property (1) implies no strategy that can be a best response induces a misleading steady state. When $\lambda \geq \lambda^*$, suppose there is a limit equilibrium $\sigma^*$ that does not induce a misleading steady state. Then property (2) means $\sigma^*=\sigma^{\text{maj}}$. But since $\sigma^{\text{maj}}$ does have a misleading steady state in this region, this gives a contradiction.

Underlying this argument is the fact, from Theorem \ref{thm:ss_half_stable}, that fixed points stable on at least one side must be steady states. Without this result, we could only conclude that there is a fixed point in the  misleading region under any limit equilibrium when $\lambda \geq \lambda^*$. But the fixed point under the limit equilibrium might be a touchpoint (even for $\lambda>\lambda^*$), so we would not know whether it is actually reached with positive probability and cannot conclude that the equilibrium must have a misleading  steady state.

\section{Concluding Discussion}
\label{sec:conclusion}

We have developed a model of learning from sampled signals where rational agents  selectively endorse signals that they believe  more likely to be true, increasing their popularity. To track the evolution of signal popularities in the society, we combine stochastic approximation tools with analysis of equilibrium behavior. We find that a sampling rule that weighs popularity more heavily can help aggregate information, but can also generate misleading steady states where incorrect signals circulate widely. We conclude with some remarks on the broader applicability of the model.

%The compositions of news feeds  depend on what stories others have shared and  how much weight the platform's sampling algorithm places on these previous sharing decisions. Putting more weight on popularity when sampling news feeds can help aggregate more information. But at a critical threshold of virality weight, a misleading steady state where users primarily see and share incorrect stories discontinuously emerges.

\textbf{Motivations for Endorsing Signals.} We have assumed that agents want to endorse signals that match the true state. As discussed in Section \ref{sec:equilibrium_intuition}, the key property driving our analysis (and in particular the proof of Theorem~\ref{thm:equilibrium_ss_by_lambda}) is that majority rule is a strict best response whenever there are no misleading steady states.
So our results hold for any utility functions satisfying this property: one simple example is if agents want to endorse signals that are more popular. This also implies our characterization results are robust to small deviations from our benchmark utility function or from symmetric priors, since the strict best response property of majority rule must continue to hold under small enough perturbations of the utility function or prior.\footnote{With an asymmetric prior, a caveat is that our results still characterize limit equilibrium restricting to symmetric strategies. These exist for a positive measure set of virality weights when the asymmetry is not too large, but may not exist for all virality weights.} Nevertheless, sufficiently different preferences could lead to different dynamics. In particular, if agents primarily care about influencing the long-run accuracy of sampled signals in the society or an eventual societal decision, then perhaps it is possible for them to avoid misleading steady states.

\textbf{Beyond Social Media.} We view social-media platforms as the main application of our model, but briefly mention two other potential applications. First, our model could also be interpreted as describing communication norms in offline information sharing. The parameter $\lambda$ would then measure how frequently people communicate their personal experiences or private information relative to passing along others' experiences or information. Second, related models have been used to describe product-adoption dynamics when consumers use simple heuristics \citep{smallwood1979product}. Our techniques suggest a path toward introducing equilibrium behavior into such models.

\begingroup
\singlespacing
\setlength{\bibsep}{0pt}
\setlength{\itemsep}{0pt}
\setlength{\parsep}{0pt}
\setlength{\parskip}{0pt}
\bibliographystyle{ecta}
\bibliography{SharingSignals}
\endgroup
\appendix
\section{Selected Proofs}

\begin{proof}[Proof of Proposition \ref{prop:existence}]
Fixing $n$ gives a symmetric finite game, and agents have a state-symmetric best response whenever all other agents use a state-symmetric strategy $\sigma$. So by Kakutani's fixed point theorem, there exists a symmetric BNE.

Now fix $q,K,C, \lambda, n_0$ and $\psi$. For each $n$, there exists a symmetric BNE $\sigma^{(n)}$. Because the space of strategies $\sigma$ is compact, we can choose a convergent subsequence. The limit of this subsequence is a limit equilibrium.
\end{proof}

\begin{proof}[Proof of Proposition \ref{prop:conv}]
The proof applies a convergence result from stochastic approximation
from Chapter 2 of \citet{borkar2009stochastic}. Suppose agents use
strategy $\sigma$. Without loss of generality, we can condition on
$\omega=1$.

Let 
$Y=\{\mathbf{y}=(x,z) \in [0,1]^2\}.$
Let $S_0:=\sum_{m=1}^{n_0}\rho_0(s_0^m)$ be the initial total seed score and let $N_0^+$ be the number of seed signals with realization $1$ (conditioning on $\omega=1$). After $t$ agents have acted, total popularity score is $S_t=S_0+(C+1)t$ and the pool contains $N_t=n_0+t$ signals. For each $t$, define $\mathbf{y}(t)\in Y$ by
\[
x(t)=\frac{\sum_{s\in\mathcal P_t:s=1}\rho_t(s)}{S_t}
\qquad\text{and}\qquad
z(t)=\frac{N_0^++\sum_{i=1}^t\mathbf 1\{s_i=1\}}{N_t}.
\]
The first entry of $\mathbf{y}(t)=(x(t),z(t))$ measures the total popularity share of signals with realization $1$. The second entry measures the fraction of signals in the current pool which have realization $1$.

Let $\mathbf{\xi}(t+1)$ be the random variable with first entry equal to the fraction of the $C+1$ popularity points added in period $t+1$ that accrue to positive signals (including the new signal's initial popularity point)  and second entry equal to a binary indicator for whether $s_{t+1}=1$. The exact recursions are
\[
x(t+1)=x(t)+\frac{C+1}{S_t+C+1}\bigl(\xi_1(t+1)-x(t)\bigr) \text{ and }
z(t+1)=z(t)+\frac{1}{N_t+1}\bigl(\xi_2(t+1)-z(t)\bigr).
\]
Since $(C+1)/(S_t+C+1)=1/(t+1)+O(t^{-2})$ and $1/(N_t+1)=1/(t+1)+O(t^{-2})$, the step size of the stochastic process is $\frac{1}{t+1}$ plus a summable perturbation.

Following the notation of \citet{borkar2009stochastic}, we write
\[
h(\mathbf{y}(t))=\mathbb{E}\left[\mathbf{\xi}(t+1)\,\middle|\,\mathbf{y}(t)\right]-\mathbf{y}(t)\text{ and }M(t+1)={\mathbf{\xi}(t+1)}-\mathbb{E}\left[{\mathbf{\xi}(t+1)}\,\middle|\,\mathbf{y}(t)\right].
\]
We can then decompose the change in the stochastic process $\mathbf{y}(t)$ into $h(\mathbf{y}(t))$, which depends deterministically on $\mathbf{y}(t)$, a martingale difference term $M(t+1),$ and a summable initial-condition perturbation.  We would like to apply Theorem 2.1 of Chapter 2 of \citet{borkar2009stochastic},
which requires the following assumptions: 
\begin{itemize}
\item[\textbf{(A1)}] $h$ is Lipschitz continuous. 
\item[\textbf{(A2)}] $\sum_{t}\frac{1}{t+1}=\infty$ while $\sum_{t}\frac{1}{(t+1)^{2}}<\infty$. 
\item[\textbf{(A3)}] $\mathbb{E}\left[M(t+1)\mid\mathbf{y}(t)\right]=0$ and $\{M(t)\}$
are square-integrable with $\mathbb{E}\left[\|M(t+1)\|^{2}\mid\mathbf{y}(t)\right]\leq \kappa (1+\|\mathbf{y}(t)\|^{2})$ a.s. for all $t$ and some $\kappa>0$. 
\item[\textbf{(A4)}] $\|\mathbf{y}(t)\|$ remains bounded a.s. 
\end{itemize}
Properties \textbf{(A2)} and \textbf{(A4)} are immediate. For \textbf{(A3)},
the martingale property holds by the construction of $M(t)$ and the
remaining properties hold because $M(t)$ is bounded (independent
of $t$).

Property \textbf{(A1)} remains. Since $-\mathbf{y}(t)$ is Lipschitz
continuous in $\mathbf{y}(t)$, we must check that $\mathbb{E}\left[\mathbf{\xi}(t+1)\mid\mathbf{y}(t)\right]$
is Lipschitz continuous in $\mathbf{y}(t)$.

Write $\sigma_{1}(s,k)$ for the expected number of ``1'' signals that strategy $\sigma$
endorses, when the agent's private signal
is $s$ and $k$ signals in the sample are ``1''.  Write $P_{k}(x,z,\lambda)$ for $\mathbb{P}[\text{Binom}(K,\lambda x+(1-\lambda)z)=k]$,
where $\text{Binom}(n,p)$ is the binomial distribution with $n$
trials and success probability $p$.

For every $t$, the conditional expectation of the random variable $\xi_{1}(t+1)$ equals
\[
\frac{1}{C+1}\left(q+\sum_{0\leq k\leq K}P_{k}(x(t),z(t),\lambda)(q\sigma_{1}(1,k)+(1-q)\sigma_{1}(-1,k))\right).
\]
Indeed, popularity-based sampling selects a positive signal with probability exactly $x(t)$, while uniform sampling selects a positive signal with probability exactly $z(t)$. The finite seed pool therefore changes the recursion only through the initial condition and the summable step-size perturbations described above, which do not affect the asymptotic behavior of the process.

This conditional expectation is a polynomial of degree at most $K$ in $x(t)$
and $z(t)$, and therefore is Lipschitz
continuous on $Y$. The conditional expectation of $\xi_{2}(t+1)$ is constant, and therefore Lipschitz continuous in $Y$ as well.

For $\mathbf{r} = (r_1, r_2)$, we can define a continuous-time differential equation by letting 
\begin{equation}
\dot{\mathbf{r}}(t)=h(\mathbf{r}(t)),t\geq0.\label{eq:diffeq}
\end{equation}
An invariant set $A$ of (\ref{eq:diffeq}) is a set such that $\mathbf{r}(0)\in A$
implies $\mathbf{r}(t)\in A$ for all $t\geq0$. An invariant set
is internally chain transitive if for any $\mathbf{r},\mathbf{r}'\in A$,
$\epsilon>0$ and $T>0$, there exists $\mathbf{r}^{0}=\mathbf{r},\mathbf{r}^{1},\hdots,\mathbf{r}^{n}=\mathbf{r}' \in A$
such that the trajectory of $\mathbf{r}(t)$ starting from $\mathbf{r}(0)=\mathbf{r}^{i}$
meets with an $\epsilon$-neighborhood of $\mathbf{r}^{i+1}$ at some
time $t\geq T$.

By Theorem 2.1 of Chapter 2 of \citet{borkar2009stochastic}, the stochastic
process $\mathbf{y}(t)$ converges to an internally chain transitive
invariant set of equation (\ref{eq:diffeq}). Because $z(t)=(N_0^++\operatorname{Binom}(t,q))/(n_0+t)\rightarrow q$ almost surely, any internally chain transitive invariant set that $\mathbf{y}(t)$ converges to must be contained in $[0,1]\times\{q\}$. 
We claim that at any $\mathbf{r}$ contained in an internally chain
transitive invariant set $A$, we must have $\frac{dr_1(t)}{dt}=0$
when $\mathbf{r}(t)=\mathbf{r}$. Suppose an internally chain transitive invariant set $A$ of (\ref{eq:diffeq})
contains a point $\mathbf{r}$ at which $
\frac{dr_1(t)}{dt}>0.
$
Letting $\mathbf{r}(0)=\mathbf{r}$, we can choose some $t'>0$ such
that $r_1(t')>r_1(0)$ and $
\frac{dr_1(t)}{dt}>0$
at $t=t'$. Now let $\mathbf{r}'=\mathbf{r}(t')$. We have $\mathbf{r}'\in A$
by invariance.

If we consider the trajectory $\mathbf{r}(t)$ beginning with $\mathbf{r}(0)=\mathbf{r}'$,
we cannot have $r_1(t)$ fall below $r_1(0)$ since $\dot{r}_1(0)>0$ and since $A\subseteq[0,1]\times\{q\}$ implies $r_2(t)\equiv q$ along trajectories starting in $A$, 
the sign of $\dot{r}_1(t)$ only depends on $t$ through $r_1(t)$. For $\epsilon>0$ sufficiently small this implies that the trajectory
$\mathbf{r}(t)$ beginning with $\mathbf{r}(0)=\mathbf{r}'$ never
enters an $\epsilon$-neighborhood of $\mathbf{r}$. This contradicts
the assumption that $A$ is internally chain transitive.

If $A$ contains a point $\mathbf{r}$ at which 
$
\frac{dr_1(t)}{dt}<0,$
we obtain a contradiction similarly. This shows  that 
$
\frac{dr_1(t)}{dt}=0
$
at all $\mathbf{r}$ contained in an internally chain transitive invariant
set.

Values of $r_1(t)$ for which $\frac{dr_1(t)}{dt}=0$
correspond to the roots of a  polynomial of degree at most $K$ (which is not identically zero since $\phi_\sigma(0)-0\ge q/(1+C)>0$), and therefore
there are at most finitely many such values. Calling the set of such
values $X^{*}(1)$, since these equilibria are isolated, any internally chain transitive set contained in them must be a singleton; hence $x(t)$ converges almost
surely to some $x^{*}\in X^{*}(1)$.
\end{proof}

\begin{proof}[Proof of Theorem \ref{thm:ss_half_stable}]

We say a fixed point $x^*$ of $\phi_{\sigma}(x)$ is a \emph{touchpoint} if there exists $\epsilon>0$ such that $\phi_{\sigma}(x)<x$ for all $x \neq x^*$ in $(x^*-\epsilon,x^*+\epsilon)$ or $\phi_{\sigma}(x)>x$ for all $x \neq x^*$ in $(x^*-\epsilon,x^*+\epsilon)$.

Case (i): $x^*$ is a touchpoint.

The seed signals only change the initial condition. Conditional on any seed realization with finite positive total score, every  $(x(t),z(t))$ that has positive probability under the original process still has positive probability after some feasible realization of seed signals. Hence, before applying the local arguments below, we may condition on a finite history that brings $x(t)$ into the relevant one-sided neighborhood of $x^*$.

The proof extends the arguments from Theorem 1 of \cite{pemantle1991touchpoints}. Suppose that $\phi_{\sigma}(x)>x$ for all $x\neq x^*$ in $(x^*-\epsilon,x^*+\epsilon)$. The other case is the same.
%We can assume without loss of generality that $x^*$ is the smallest fixed point of $\phi_{\sigma}(x)$.%WLOG from Pemantle paper but perhaps could use additional justification
%It is sufficient to show that there exists $t_0$ such that with positive probability, $x(t)<x^*$ for all $t>t_0$.

Fix $v \in (0,\frac12)$ and $v_1 \in (v,\frac12)$. Choose $\gamma>1$ such that $\gamma v_1 < \frac12$. Define $g(r) = re^{(1-r)/(2v_1\gamma)}$. Then $g(1)=1$ and $g'(1) = 1-1/(2v_1\gamma) < 0$, so we can choose $r_0 \in (0,1)$ with $g(r_0)>1$. Also define $T(n) = e^{n(1-r_0)/(\gamma v_1)}.$ Then $g(r_0)^n = r_0^n T(n)^{1/2} > 1.$

Choose $N$ such that $\gamma r_0^N < \epsilon$. Since $T(1)^{1/2}r_0 = g(r_0)>1$, we can find $\chi>0$ such that $T(1)^{1/2-\chi}r_0>1$ and therefore $T(n)^{1/2-\chi}r_0^{n} \rightarrow \infty$. Let $$\tau_N = \inf\{j>T(N): x(j-1) < x^*-r_0^N < x(j) \} \text{ and }\tau_{n+1} = \inf\{j \geq \tau_n: x(j) > x^*-r_0^{n+1}\}$$
for each $n \geq N$ (using the convention that $\tau_N=\infty$ if the inequalities are not satisfied for any $j$). So $\tau_n$ is the first time the stochastic process crosses $x^*-r_0^n$.

We will show the probability that $\tau_n > T(n)$ for all $n \geq N$ is positive. Since we can assume (by the argument in case (ii) below) that $x(t) \rightarrow x^*$ from below whenever $\tau_n > T(n)$ for all $n \geq N$, this will complete the case.

Let $z(t)=(N_0^++\sum_{i=1}^t\mathbf 1\{s_i=1\})/(n_0+t)$ be the fraction of signals in the current pool with realization $1$. 
We first bound the probability that $z(t)$ is far from $q$. Define a function $$\phi_{\sigma,z}(x):=\frac{q+\sum_{k=0}^{K}\mathbb{P}[\text{Binom}(K,\lambda x+(1-\lambda)z)=k]\cdot[q\cdot\mathbb{E}[\sigma(1,k)]+(1-q)\cdot\mathbb{E}[\sigma(-1,k)]]}{1+C}$$
to be the inflow accuracy when a fraction $z$ of signals in the current pool have value $1$.

We begin by defining an event $\mathscr{C}$ under which the number of private signals with positive realization is close to $q$ for $t$ sufficiently large. Let $\mathscr{C}_1$ be the event that for all $n \geq N$ and for all $t \geq T(n)$,
$\phi_{\sigma,z(t)}(x) -x \geq - 1/T(n)^{1/2-\chi}$
on  $(x^*-\epsilon,x^*+\epsilon)$. Because $\phi_{\sigma,z}(x) -x$ is polynomial (in $z$ and $x$) and is non-negative on this interval when $z=q$, this holds for $|z(t)-q|<B/T(n)^{1/2-\chi}$ for some $B>0$. %for example by an envelope theorem argument

Suppose event $\mathscr{C}_1$ holds and $\tau_n > T(n)$. Then we have
\begin{align*}
    \sum_{t=\tau_n}^{j} h_1(\mathbf{y}(t))/(t+1)
    & = \sum_{t=\tau_n}^{j} (\phi_{\sigma,z(t)}(x(t)) - x(t))/(t+1)
    \\ & \geq - \sum_{m = n}^{\infty} \frac{1}{T(m)^{1/2-\chi}} \sum_{T(m) \leq t < T(m+1)} \frac{1}{t+1} \text{ by the definition of }\mathscr{C}_1
    \\ & \geq -\sum_{m=n}^{\infty} \frac{\log(\lceil T(m+1)\rceil)-\log(\lceil T(m)\rceil)}{T(m)^{1/2-\chi}} 
    \\ & \geq -\sum_{m=n}^{\infty} \left(\frac{1-r_0}{\gamma v_1}+1\right) \cdot e^{-m(1/2-\chi)(1-r_0)/(\gamma v_1)}
    \\ & = - \left(\frac{1-r_0}{\gamma v_1}+1\right)\cdot \frac{e^{-n(1/2-\chi)(1-r_0)/(\gamma v_1)}}{1-e^{-(1/2-\chi)(1-r_0)/(\gamma v_1)}}. \stepcounter{equation}\tag{\theequation}\label{eq:ineqsumh}
\end{align*}
We define $\mu =\left(\frac{1-r_0}{\gamma v_1}+1\right)\cdot \frac{1}{1-e^{-(1/2-\chi)(1-r_0)/(\gamma v_1)}}$, so that the right-hand side is $-\mu T(n)^{-(1/2-\chi)}.$

Let $\mathscr{C}_2$ be the event that for all $n \geq N$ and for all $t \geq T(n)$,
\begin{equation}\label{eq:eventc2}\phi_{\sigma,z(t)}(x) -x  \leq v\gamma r_0^n\end{equation}
for all $x \in [x^*-\gamma r_0^n,x^*].$  Because $\phi_{\sigma,z}(x) -x$ is polynomial (in $z$ and $x$) and $\frac{d(\phi_{\sigma,q}(x) -x)}{dx}(x^*)=0,$
we can choose $B'$ such that for all $n \geq N$ we have $\phi_{\sigma,z}(x) -x \leq v\gamma r_0^n$
for $x \in [x^*-\gamma r_0^n,x^*]$ whenever $|z-q|< B' r_0^n$ (since we can bound the entries of the Hessian of $\phi_{\sigma,z}(x) -x$ above by a constant on the rectangle $[x^*-\gamma r_0^N,x^*] \times [q-r_0^N,q+r_0^N]$). Because $T(n)^{1/2-\chi} r_0^n > 1$, this holds for $|z(t)-q|<B'/T(n)^{1/2-\chi}$ for some $B'>0$.

Define the event $\mathscr{C} = \mathscr{C}_1 \cap \mathscr{C}_2$ to be the intersection of these two events. The event $\mathscr{C}$ holds when  $|z(t)-q|<\min(B,B')/T(n)^{1/2-\chi}$ for all $n \geq N$ and all $t \geq T(n)$. Since $z(t)=(N_0^++\operatorname{Binom}(t,q))/(n_0+t)$, the contribution from the seed signals is $O(1/t)$. Thus, after increasing constants and taking $N$ large, the Chernoff bound implies that the probability of $|z(t)-q| > \min(B,B')/T(n)^{1/2-\chi}$ is at most $2e^{-\min(B,B')^2t^{2\chi}/(2q^2)}$.   So the probability that the event $\mathscr{C}$ does not hold for some $n \geq N$ and all $t \geq T(n)$ is at most
$$2\sum_{n=N}^{\infty} \sum_{t=T(n)}^{\infty} 2e^{-\min(B,B')^2t^{2\chi}/(2q^2)}.$$ For $N$ sufficiently large, this sum is approximately $$\sum_{n=N}^{\infty} \frac{1}{\chi}\left(\frac{\min(B,B')^2}{2q^2}\right)^{-\frac{1}{2\chi}} \Gamma\left(\frac{1}{2\chi}, T(n)^{2\chi}\min(B,B')^2/(2q^2)\right) $$
where $\Gamma(s,x)$ is the incomplete Gamma function. Since $\Gamma(s,x)/(x^{s-1}e^{-x}) \rightarrow 1$ as $x \rightarrow \infty$, this sum converges to zero as $N \rightarrow \infty$. Increasing $N$ if necessary, we can conclude that the event $\mathscr{C}$ has positive probability. For the remainder of case (i), we condition on this event $\mathscr{C}$.

Now let $\mathscr{B}$ be the event $\{\inf_{j>\tau_n} x(j) \geq x^*-\gamma r_0^n\}$. We will bound the probability of this event conditional  on $\tau_n > T(n)$. Let $Z_{m,n} = \sum_{t=m}^{n-1} M(t+1)$ be the sum of the martingale parts of the stochastic process. (Here $M(t+1)$ denotes the scaled martingale increment $(\xi_1(t+1)-\mathbb{E}[\xi_1(t+1)\mid\mathbf{y}(t)])/(t+1)$, so that $x(t+1)=x(t)+h_1(\mathbf{y}(t))/(t+1)+M(t+1)$). Because the scaled martingale increments satisfy $|M(t+1)| \le B_0/(t+1)$ for some constant $B_0$ depending only on $C$ and the seed pool, we have \begin{equation}
    \label{eq:l2bound}
\mathbb{E}[Z_{m,\infty}^2]=\sum_{t=m}^{\infty} \mathbb{E}[M(t+1)^2] \leq \sum_{t=m}^{\infty}\left(\frac{B_0}{t+1}\right)^2\leq 
\frac{B_0^2}{m}.\end{equation} 

We have:
%reference for bound on variance of sum of martingales http://galton.uchicago.edu/~lalley/Courses/385/Martingales.pdf
\begin{align*}
    \mathbb{P}\left[\mathscr{B}^c\,\middle|\, \tau_n > T(n)\right] & = \mathbb{P}\left[\inf_{j>\tau_n} x(j) < x^*-\gamma r_0^n \,\middle|\, \tau_n > T(n)\right]
    \\ & \leq \mathbb{P}\left[\inf_{j>\tau_n} Z_{\tau_n,j} < -(\gamma-1) r_0^n+\mu T(n)^{-(1/2-\chi)}\, \middle|\, \tau_n > T(n)\right] \text{ by equation }(\ref{eq:ineqsumh})
    \\& \leq \mathbb{E}\left[Z_{\tau_n,\infty}^2 \, \middle|\, \tau_n > T(n)\right]/((\gamma-1) r_0^n-\mu T(n)^{-(1/2-\chi)})^2 \text{ by Doob's $L^2$ max inequality}
    \\& \leq (B_0)^2e^{-n(1-r_0)/(v_1\gamma)}((\gamma-1) r_0^n-\mu T(n)^{-(1/2-\chi)})^{-2} \text{ by  (\ref{eq:l2bound}) and   definition of } T(n).
\end{align*}
(We could use Doob's $L^2$ max inequality since $\{Z_{\tau_n,j}\}_{j\geq\tau_n}$ is a square-integrable martingale with $Z_{\tau_n,j}\to Z_{\tau_n,\infty}$ a.s.)
Recall that $T(n)^{1/2-\chi}r_0^{n} \rightarrow \infty$, so for $n$ sufficiently large $$(\gamma-1) r_0^n-\mu T(n)^{-(1/2-\chi)} \geq \frac{\gamma-1}{2} r_0^n.$$
We conclude that
$\mathbb{P}\left[\mathscr{B}^c\,\middle|\, \tau_n > T(n)\right] \leq (B_0)^2\left(\frac{\gamma-1}{2}\right)^{-2}g(r_0)^{-2n}.$
This bounds the conditional probability of the event $\mathscr{B}$ not holding.

When the event $\mathscr{B}$ does hold and $\tau_n>T(n)$,
\begin{align*}
    \sum_{\substack{T(n) < t <T(n+1)\\ x(t) < x^*}} h_1(\mathbf{y}(t))/(t+1) &= \sum_{\substack{T(n) < t <T(n+1)\\ x(t) < x^*}} (\phi_{\sigma,z(t)}(x(t))-x(t))/(t+1)
    \\ & \leq  \sum_{\substack{T(n) < t <T(n+1)\\ x(t) < x^*}}
    v\gamma r_0^n/(t+1) \text{ by equation (\ref{eq:eventc2})}
    \\ & \leq (\log \lceil T(n+1) \rceil-\log \lceil T(n) \rceil)(v\gamma r_0^n)  \text{ by the harmonic series partial sums}
    \\ & \leq (v\gamma r_0^n)((1-r_0)/(\gamma v_1) + 1/T(n))
    \\ & = (v/v_1)(r_0^n-r_0^{n+1})+v\gamma r_0^n/T(n).
 \end{align*}

Now suppose $\mathscr{B}$ holds and $\tau_n>T(n)$ but $\tau_{n+1} \leq T(n+1)$. Then \begin{align*}
Z_{\tau_n,\tau_{n+1}}&=x(\tau_{n+1})-x(\tau_n)-\sum_{t=\tau_n}^{\tau_{n+1}-1} h_1(\mathbf{y}(t))/(t+1) 
\\ & \geq x(\tau_{n+1})-x(\tau_n)-\sum_{\substack{T(n) < t <T(n+1)\\ x(t) < x^*}} h_1(\mathbf{y}(t))/(t+1)
\\ & \geq r_0^n-r_0^{n+1} - \xi_n -(v/v_1)(r_0^n-r_0^{n+1})-v\gamma r_0^n/T(n) \text{ by the inequality above and definition of }\tau_n
\\ & = r_0^n(1-r_0)(1-v/v_1)-\xi_n - v\gamma r_0^n/T(n),
\end{align*}
where $\xi_n$ is an error term since $x(\tau_n)$ may be larger than $x^*-r_0^n$ and $\widetilde{\xi}_n =\xi_n + v\gamma r_0^n/T(n) $. Since the error term $\xi_n$ is at most $1/T(n)$ and so is lower order than $r_0^n$, we have as $n\rightarrow \infty$
\begin{equation}\label{eq:xismall}\frac{r_0^n(1-r_0)(1-v/v_1)-\widetilde{\xi}_n }{r_0^n(1-r_0)(1-v/v_1)}\rightarrow 1. \end{equation}

Combining our bounds, we have:
\begin{align*}
 & \mathbb{P}[\tau_{n+1}\leq T(n+1)\,|\,\tau_{n}>T(n)]\\
\le & \mathbb{P}\left[\mathscr{B}^{c}\ |\,\tau_{n}>T(n)\right]+\mathbb{P}\left[\mathscr{B},\sup_{j\ge\tau_{n}}Z_{\tau_{n},j}\geq r_{0}^{n}(1-r_{0})(1-v/v_{1})-\widetilde{\xi}_{n}\ |\,\tau_{n}>T(n)\right]\\
\le & (B_0)^{2}\left(\frac{\gamma-1}{2}\right)^{-2}g(r_{0})^{-2n}+\frac{\mathbb{E}[Z_{\tau_{n},\infty}^{2}\,|\,\tau_{n}>T(n)]}{(r_{0}^{n}(1-r_{0})(1-v/v_{1})-\widetilde{\xi}_{n})^{2}}\text{ by Doob's \ensuremath{L^{2}} max inequality}\\
\le & (B_0)^{2}\left(\frac{\gamma-1}{2}\right)^{-2}g(r_{0})^{-2n}+\frac{(B_0)^{2}T(n)^{-1}}{(r_{0}^{n}(1-r_{0})(1-v/v_{1})-\widetilde{\xi}_{n})^{2}}\text{ \text{ by inequality (\ref{eq:l2bound})}}\\
\le & (B_0)^{2}\left(\frac{\gamma-1}{2}\right)^{-2}g(r_{0})^{-2n}+(B_0)^{2}((1-r_{0})(1-v/v_{1}))^{-2}g(r_{0})^{-2n}\cdot\left(\frac{r_{0}^{n}(1-r_{0})(1-v/v_{1})}{r_{0}^{n}(1-r_{0})(1-v/v_{1})-\widetilde{\xi}_{n}}\right)^{2}.
\end{align*}

We claim that the sum of these probabilities converges. The sum of the first terms converges because $g(r_0)>1$. For the second term, recall that the fraction $\frac{r_0^n(1-r_0)(1-v/v_1)}{r_0^n(1-r_0)(1-v/v_1)-\widetilde{\xi}_n }$ converges to $1$. So the sum of the second terms also converges because $g(r_0)>1$. We have
$$\mathbb{P}[\tau_n > T(n) \text{ for all }n \geq N] = \mathbb{P}[\tau_N>T(N)] \prod_{n=N}^{\infty} (1-\mathbb{P}[\tau_{n+1} \leq T(n+1)\,|\, \tau_n>T(n)]).$$
On the right-hand side, each factor in the product is positive and $\sum_{n=N}^{\infty} \mathbb{P}[\tau_{n+1} \leq T(n+1)\,|\, \tau_n>T(n)]$ is finite. By a standard result on infinite products, this implies the product is positive. So the probability that $\tau_n>T(n)$ for all $n \geq N$ is positive, which implies that the probability $\pi(x^*|\sigma)$ of converging to $x^*$ is positive.

Case (ii): There exists $\epsilon>0$ such that $\phi_{\sigma}(x) > x$ for all $x \in (x^*-\epsilon,x^*)$ and $\phi_{\sigma}(x) < x$ for all $x \in (x^*,x^*+\epsilon)$.

Our argument is based on the related result for generalized P{\'o}lya urns  from \cite*{hill1980strong}. We begin with a lemma, which says that suitably changing a stochastic process away from a neighborhood of a fixed point does not affect whether we converge to that fixed point with positive probability:

\begin{lem}\label{lem:positiveprobconv}
Suppose 
$\widetilde{\mathbf{y}}(t+1)=\widetilde{\mathbf{y}}(t)+\frac{1}{t+1}\left( \widetilde{\mathbf{\xi}}(t+1)-\widetilde{\mathbf{y}}(t) \right)$, where the conditionally i.i.d. random variables $\widetilde{\mathbf{\xi}}(t+1)$ have the same conditional distribution as ${\mathbf{\xi}}(t+1)$ in a neighborhood $U$ of $(x^*,q)$, have the same support as ${\mathbf{\xi}}(t+1)$ for all $(x,z) \in (0,1)^2$, and have expectations $\mathbb{E}[\widetilde{\mathbf{\xi}}(t+1)]$ that are Lipschitz continuous in $(x,z)$. Then $x(t)$ converges to $x^*$ with positive probability if and only if $\widetilde{x}(t) = \widetilde{\mathbf{y}}_1(t)$ does.
\end{lem}

\begin{proof}The stochastic process $x(t)$ converges to $x^*$ with positive probability if and only if there exists some $T$ and some $(x(T),z(T))$ reached with positive probability under $\mathbf{y}(t)$ such that starting with initial condition $(x(T),z(T))$, with positive probability $x(t)\rightarrow x^*$ and $(x(t),z(t)) \in U$ for $t \geq T$.

Because the random variables $\widetilde{\mathbf{\xi}}(t)$ have the same support as ${\mathbf{\xi}}(t)$ whenever $x$ and $z$ are interior, the point $(x(T),z(T))$ is reached with positive probability under $\widetilde{\mathbf{y}}(t)$ if and only if it is reached with positive probability under $\mathbf{y}(t)$. Because $\widetilde{\mathbf{\xi}}(t)$ and ${\mathbf{\xi}}(t)$ agree on $U$, starting with initial condition $(x(T),z(T))$, with positive probability $\widetilde{x}(t)\rightarrow x^*$ and $(\widetilde{x}(t),\widetilde{z}(t)) \in U$ for $t \geq T$ if and only if the same holds for $(x(t),z(t))$. These conditions hold for some $(x(T),z(T))$ if and only if $\widetilde{x}(t)$ converges to $x^*$ with positive probability.
\end{proof}
Now choose $\widetilde{\xi}(t)$ satisfying the conditions of the lemma, agreeing with ${\xi}(t)$ in the second coordinate, and such that the unique fixed point of the corresponding function $\widetilde{\phi}_{\sigma}(x)$ is $x^*$. To do so, choose an open neighborhood $U$ of $(x^*, q)$ such that $x^*$ is the unique fixed point of $\phi_{\sigma}(x)$ with $(x,q) \in \overline{U}$. Let $\widetilde{\xi}(t)=\xi(t)$ on the closure $\overline{U}$ of $U$. For each $z$, let $\widetilde{\xi}(t)$ be constant in $x$ outside of the neighborhood $U$.

Then $\widetilde{\xi}(t)$ and $\xi(t)$ have the same support for all interior $x$ and $z$. Lipschitz continuity follows from Lipschitz continuity of the expectations of $\xi(t)$ in $x$ and $z$, which we checked in the proof of Proposition~\ref{prop:conv}.

Since $x^*$ is the unique fixed point of $\widetilde{\phi}_{\sigma}(x)$, by the same argument as in Proposition~\ref{prop:conv}, we have $\widetilde{x}(t) \rightarrow x^*$ almost surely. Note that this step uses Lipschitz continuity of $\mathbb{E}[\widetilde{\mathbf{\xi}}(t+1)]$. So by Lemma~\ref{lem:positiveprobconv}, $x(t) \rightarrow x^*$ with positive probability.

For the converse, if $x(t)\to x^*$ with positive probability, then $h_1(x^*,q)=\phi_\sigma(x^*)-x^*=0$ since $z(t)\to q$ a.s.\ and the drift $h_1$ is continuous. Hence $x^*$ is a fixed point of $\phi_\sigma$.
\end{proof}

\begin{proof}[Proof of Lemma \ref{lem:maj_ss_types}]
If $x>1/2,$ then sampling accuracy is $\lambda x+(1-\lambda)q>1/2$
since $q>1/2$ also. Conversely, suppose $\omega=1$, $x\le 1/2$ and sampling accuracy is at least $1/2$. Under $\sigma^{\mathrm{maj}}$, the expected number of positive signals endorsed is at least $C/2$. To see this, pair every sample realization with $k<K/2$ positive signals with its mirror realization with $K-k$ positive signals, and note the mirror realization has weakly larger probability when sampling accuracy is at least $1/2$. Conditional on the event that there are either $k$ or $K-k$ positive signals in the sample, the average number of positive endorsements is at least $C/2$. If $K$ is even, then the average number of positive endorsements conditional on $K/2$ positive signals in the sample is also at least $C/2$, since tie is broken in the direction of the private signal and $q>1/2$.  Thus $\phi_{\sigma^{\mathrm{maj}}}(x)>1/2\ge x$, so such an $x$ cannot be a fixed point. Therefore, if $x<1/2$ and $\phi_{\sigma^{\text{maj}}}(x)=x,$ the sampling accuracy must be strictly less than $1/2$, and $x=1/2$ is not a fixed point.
\end{proof}

\begin{proof}[Proof of Theorem \ref{thm:equilibrium_ss_by_lambda}]

We begin with four preliminary lemmas. Throughout the proof, write $p(x):=\lambda x+(1-\lambda)q$ for the sampling accuracy associated with viral accuracy $x$.

\begin{lem}
\label{lem:3_implies_1}
Suppose $\sigma$ is state symmetric, $\mathbb{E}[\sigma(1,k)]\ge\mathbb{E}[\sigma(-1,k)]$
for every $0\le k\le K$, and that $\sigma(1,K/2)(U_{K/2})=1$, $\sigma(-1,K/2)(L_{K/2})=1$ if $K$ is even.  If sampling accuracy at $x$ is weakly smaller
than 1/2, then $\phi_{\sigma^{\text{maj}}}(x)\le\phi_{\sigma}(x)$. The inequality is strict if the sampling accuracy is strictly smaller than 1/2 and $\sigma \ne\sigma^{\text{maj}}$.
\end{lem}
\begin{proof}
Fix $k<K/2$ and write $P_j=P_j(x,\lambda)$,
$A=\mathbb E[\sigma(1,k)]$, and $B=\mathbb E[\sigma(-1,k)]$.
By state symmetry,
$\mathbb E[\sigma(1,K-k)]=C-B$, $\mathbb E[\sigma(-1,K-k)]=C-A$.
The contribution of the pair $k,K-k$ to the numerator of $\phi_\sigma$ in the expression from Definition~\ref{def:inflow_acc} is
\[
\begin{aligned}
T_\sigma(k)
&=P_k(qA+(1-q)B)
+P_{K-k}(q(C-B)+(1-q)(C-A))\\
&=P_{K-k}C
+(A-B)(qP_k-(1-q)P_{K-k})
+B(P_k-P_{K-k}).
\end{aligned}
\]
The corresponding contribution under majority rule is
\[
T_{\mathrm{maj}}(k)=P_kL_k+P_{K-k}U_{K-k}
=P_{K-k}C+L_k(P_k-P_{K-k}),
\]
where the last equality uses $U_{K-k}=C-L_k$. Therefore
\[
T_\sigma(k)-T_{\mathrm{maj}}(k)
=
(A-B)(qP_k-(1-q)P_{K-k})
+(B-L_k)(P_k-P_{K-k}).
\]
Since $p(x)\le 1/2$ and $k<K/2$, we have $P_k\ge P_{K-k}$ and $qP_k-(1-q)P_{K-k}>0$. Hence this difference is weakly positive for every $k<K/2$. The term for $k=K/2$, when $K$ is even, is the same for $\sigma$ and $\sigma^{\text{maj}}$ by assumption. This proves $\phi_{\sigma^{\text{maj}}}(x)\le\phi_{\sigma}(x)$. If $p(x)<1/2$, then $P_k>P_{K-k}$ for every $k<K/2$, so $T_\sigma(k)-T_{\mathrm{maj}}(k)=0$ requires $A=B=L_k$ for every $k<K/2$; by state symmetry and the  assumption on $\sigma(\cdot,K/2)$, this is exactly $\sigma=\sigma^{\text{maj}}$.
\end{proof}

The proofs of the next two lemmas are in the Online Appendix.

\begin{lem}
\label{lem:1_implies_2}Suppose $\sigma$ is state symmetric and $\sigma(1,k)(U_k)=1$ for every $k\ge K/2.$
Then, $\phi_{\sigma}$ does not have any fixed point $x$ with $\lambda x+(1-\lambda)q > 1/2$  and  $x \le q$.
\end{lem}

\begin{lem}\label{lem:maj_shape}
The inflow accuracy function for the majority rule is convex-then-concave as a function of sampling accuracy $p=\lambda x+(1-\lambda)q$: its second derivative has at most one sign change on $(0,1)$, and any such sign change is from positive to negative. In particular, $\phi_{\sigma^{\mathrm{maj}}}$ is concave on the region where sampling accuracy is at least $1/2$.
\end{lem}

\begin{lem}
\label{lem:drift_right} 

For each $\epsilon',\epsilon''>0$, $p\in(0,1),$ strategy $\sigma^{*}$ and $0\le\overline{\lambda}\le1$ with $\phi_{\sigma^{*}}^{\overline{\lambda}}(x)-x\ge2\epsilon'$ for every $x$ with $\overline{\lambda}x+(1-\overline{\lambda})q\le p+2\epsilon'$ (where $\phi_{\sigma^{*}}^{\overline{\lambda}}$ is the inflow accuracy function with virality weight $\overline{\lambda})$, there is some $N$ and some $\delta>0$ so that for every $\sigma$ with $\|\sigma-\sigma^{*}\|_2<\delta$ and $\lambda$ with $|\lambda-\overline{\lambda}|<\delta,$ we have $\mathbb{P}_{\sigma,\lambda}[\lambda x(t)+(1-\lambda)q\ge p+\epsilon'/2]>1-\epsilon''$ for every $t\ge N$. 
\end{lem}

\begin{proof}

When $\overline{\lambda}=0$, the hypothesis requires $q>p+2\epsilon'$ (otherwise the condition $\overline{\lambda}\,x+(1-\overline{\lambda})q\le p+2\epsilon'$ reduces to $q\le p+2\epsilon'$, and since $\phi_{\sigma^{*}}^{0}(x)$ is constant in $x$ and bounded by 1, the assumption $\phi_{\sigma^{*}}^{0}(x)-x\ge 2\epsilon'$ for all $x\in[0,1]$ cannot hold). Given $q>p+2\epsilon'$, for any $\delta>0$ small enough and $\lambda\in[0,\delta)$ we have $\lambda x(t)+(1-\lambda)q\ge q(1-\delta)>p+\epsilon'/2$ for all $x(t)\in[0,1]$, regardless of $\sigma$ and $t$. So the conclusion holds with $N=1$. For the remainder of the proof, assume $\overline{\lambda}>0$.

Because $\phi^{\lambda}_{\sigma}(x)$ is polynomial in $\lambda$, $\sigma$, and $x$, there exists $\delta>0$ such that $\phi_{\sigma}^{\lambda}(x)-x\ge\epsilon'$ for every $x$ with $\lambda x+(1-\lambda)q\le p+\epsilon'$ when $\|\sigma^*-\sigma\|_2<\delta$ and $|\overline{\lambda}-\lambda|<\delta$. Shrinking $\delta$ if necessary, we can also assume that $\lambda$ is bounded away from zero when $|\overline{\lambda}-\lambda|<\delta$ (which is possible since $\overline{\lambda}>0$).

For the remainder of the proof, fix $\sigma$ and $\lambda$ in these neighborhoods. We will observe at the end of the proof that the bounds we will prove are uniform in the choice of $\sigma$ and $\lambda$.

Let $p'>p+\epsilon'$ be the largest number in $(0,1)$ such that \begin{equation}\label{eq:definep'}\phi_{\sigma}(x)-x \geq \epsilon'/2\end{equation} for all $x$ satisfying $\lambda x + (1-\lambda)q \leq p'$. Let $N_1<N_2$ be positive integers with $N_2 \geq b N_1$ for some integer $b>1$. We will first show that for $N_1$ and $N_2$ large enough, the probability that $\lambda x(t) + (1-\lambda)q < p'$ for all $t \in [N_1,N_2]$ is small. We will then show that if $\lambda x(t_1) + (1-\lambda)q > p'$ for some $N_1 \leq t_1 < N_2$, then the probability that $\lambda x(N_2) + (1-\lambda)q < p+\epsilon'/2$ is small.

Since $z(t)=(N_0^++\operatorname{Binom}(t,q))/(n_0+t)$, the contribution of the seed signals to $z(t)$ is $O(1/t)$. By the Chernoff bound applied to the agents' signals and compactness of the set of strategies $\sigma$ under consideration, we can choose a constant $B>0$ independent of $\sigma$ such that \begin{equation}\label{eq:chernoff}
    \max_{x\in [0,1]} |\phi_{\sigma,z(t)}(x)-\phi_{\sigma}(x)| < \epsilon'/4
\end{equation} with probability at least $1-2e^{-Bt}$ for $t$ sufficiently large.

Recall that we can decompose $\mathbf{y}(t)$ as a stochastic-approximation recursion with drift $h(\mathbf{y}(t))$, martingale term $M(t+1)$, and perturbations from the seed signals. The scaled martingale increments satisfy $|M(t)|<B_1/t$ for some constant $B_1$ and all large $t$. So by Theorem C.7 from Appendix C of \cite{borkar2009stochastic}, for any $\alpha>0$ and any $t_1$ and $t_2$,
\begin{equation}\label{eq:martingaleconcentration}
    \mathbb{P}\left(\sup_{t_1<t<t_2} \left|\sum_{i=t_1}^{t}M(i)\right|>\alpha \right) \leq 4e^{-\frac{\alpha^2}{\sum_{i=t_1}^{t_2} B_1^2/i^2}}.
\end{equation}

Consider the event $E$ that $\lambda x(t) + (1-\lambda)q < p'$ for all $N_1 \leq t \leq N_2$. Suppose inequality (\ref{eq:chernoff}) holds for all $N_1 \leq t < N_2$. Then we have \begin{align*}x(N_2)-x(N_1) &  = \sum_{t=N_1}^{N_2-1} \frac{\phi_{\sigma,z(t)}(x(t)) - x(t)}{t+1}+ \sum_{t=N_1}^{N_2-1}M(t+1) 
\\ & = \sum_{t=N_1}^{N_2-1} \frac{\phi_{\sigma,z(t)}(x(t)) - \phi_{\sigma}(x(t))}{t+1}+\sum_{t=N_1}^{N_2-1} \frac{\phi_{\sigma}(x(t)) - x(t)}{t+1}+ \sum_{t=N_1}^{N_2-1}M(t+1) 
\\ & \geq \sum_{t=N_1}^{N_2-1}\epsilon'/4 \cdot \frac{1}{t+1}+\sum_{t=N_1}^{N_2-1}M(t+1) \text{ by inequalities }(\ref{eq:definep'})\text{ and  } (\ref{eq:chernoff})
\\ & \geq (\epsilon'/4)(\log(N_2+1)-\log(N_1+1))+ \sum_{t=N_1}^{N_2-1}M(t+1).\end{align*}
When event $E$ holds, the right-hand side must be at most $p'/\lambda$. Taking $b$ and therefore $N_2/N_1$ sufficiently large, we can assume that $(\epsilon'/4)(\log(N_2+1)-\log(N_1+1)) > 2p'/\lambda$ (since $\log((bN_1+1)/(N_1+1))\to\log b$ as $N_1\to\infty$).  By equation (\ref{eq:martingaleconcentration}), the absolute value of the sum of martingales is greater than $p'/\lambda$ with probability at most $$4e^{-\frac{(p'/\lambda)^2}{\sum_{i=N_1}^{N_2} B_1^2/(i+1)^2}}\leq 4e^{-\frac{(p'/\lambda)^2N_1N_2}{ 2B_1^2(N_2-N_1)}} < 4e^{-\frac{(p'/\lambda)^2N_1}{ 2B_1^2}}.$$
Along with the Chernoff bound, this gives an upper bound on the probability of event $E$.

If event $E$ does not hold, there exists some $N_1 \leq t \leq N_2$ such that $\lambda x(t) + (1-\lambda)q \geq p'$. Choose $t_1$ so that $t_1-1$ is the largest such $t$.  Since $\lambda x(t_1-1)+(1-\lambda)q\ge p'$ and $|x(t_1)-x(t_1-1)|\le B_2/(t_1+1)$ for some constant $B_2$ (bounded one-step increment with a finite seed pool), we have $\lambda x(t_1)+(1-\lambda)q\ge p'-\lambda B_2/(t_1+1)$.

Suppose $\lambda x(N_2) + (1-\lambda)q \leq p+\epsilon'/2$. For $N_1$ sufficiently large, this implies $t_1 \leq N_2$. So we must have
$$x(N_2)-x(t_1) \leq \frac{(p+\epsilon'/2)-(p'-\lambda B_2/(t_1+1))}{\lambda} = \frac{(p+\epsilon'/2)-p'}{\lambda}+\frac{B_2}{t_1+1} < -\frac{\epsilon'}{4\lambda}$$
for $N_1$ (and hence $t_1\ge N_1$) sufficiently large, since $(p+\epsilon'/2-p')/\lambda < -\epsilon'/(2\lambda)$ and $B_2/(t_1+1)<\epsilon'/(4\lambda)$. On the other hand, when inequality (\ref{eq:chernoff}) holds for all $N_1 \leq t < N_2$,
\begin{align*} x(N_2)-x(t_1) &  = \sum_{t=t_1}^{N_2-1} \frac{\phi_{\sigma,z(t)}(x(t)) - x(t)}{t+1}+ \sum_{t=t_1}^{N_2-1}M(t+1) 
\\ & = \sum_{t=t_1}^{N_2-1} \frac{\phi_{\sigma,z(t)}(x(t)) - \phi_{\sigma}(x(t))}{t+1}+\sum_{t=t_1}^{N_2-1} \frac{\phi_{\sigma}(x(t)) - x(t)}{t+1}+ \sum_{t=t_1}^{N_2-1}M(t+1) 
\\ & \geq
\sum_{t=t_1}^{N_2-1}\epsilon'/4 \cdot  \frac{1}{t+1}+\sum_{t=t_1}^{N_2-1}M(t+1) \text{ by inequalities }(\ref{eq:definep'})\text{ and  } (\ref{eq:chernoff})
\\ & \geq (\epsilon'/4)(\log(N_2+1)-\log(t_1+1)) + \sum_{t=t_1}^{N_2-1}M(t+1).\end{align*}
Applying equation (\ref{eq:martingaleconcentration}) with $\alpha = \epsilon'/(4\lambda)$, the absolute value of the sum of martingales is greater than $\epsilon'/(4\lambda)$ with probability at most
$$4e^{-\frac{(\epsilon'/(4\lambda))^2}{\sum_{i=t_1}^{N_2} B_1^2/(i+1)^2}}
\leq 4e^{-\frac{(\epsilon')^2N_2t_1}{32\lambda^2B_1^2(N_2-t_1)}}
\leq 4e^{-\frac{(\epsilon')^2t_1}{32\lambda^2B_1^2}}.$$
When this does not hold and the Chernoff bounds apply, $x(N_2)-x(t_1)$ is greater than $-\epsilon'/(4\lambda)$ and therefore $\lambda x(N_2) + (1-\lambda) q > p + \epsilon'/2$ for $N_1$ sufficiently large (using $\lambda x(t_1)+(1-\lambda)q\ge p'-o(1)$ and $p'>p+\epsilon'$). This gives an upper bound on the probability that $\lambda x(N_2) + (1-\lambda) q \leq p + \epsilon'/2$.

We conclude that 
$$\mathbb{P}_{\sigma}[\lambda x(N_2)+(1-\lambda)q< p+\epsilon'/2] \leq 4e^{-\frac{(p'/\lambda)^2N_1}{ 2B_1^2}}+\sum_{t=N_1+1}^{N_2-1} 4e^{-\frac{(\epsilon')^2t}{32\lambda^2 B_1^2 }} + 2\sum_{t=N_1}^{N_2-1} e^{-Bt}$$
for $N_1$ sufficiently large. Because the second and third terms are geometric series, we can choose $N_1$ sufficiently large so that this probability is less than $\epsilon''$ for all $N_2 \geq b N_1$.  Because $\lambda$ is bounded away from zero, we can make this choice uniformly in $\lambda$ and $\sigma$ (subject to the constraints $|\lambda-\overline{\lambda}|<\delta$ and $\|\sigma-\sigma^{*}\|_2<\delta$). So for $N_1$ sufficiently large, we have $\mathbb{P}_{\sigma}[\lambda x(t)+(1-\lambda)q\ge p+\epsilon'/2]>1-\epsilon''$
for $t \geq N = bN_1$.
\end{proof}

We can now prove Theorem \ref{thm:equilibrium_ss_by_lambda}. We first record two observations that will be used in Steps 2 and
3 of the proof. Let $r(t):=\lambda x(t)+(1-\lambda)z(t)$ be the probability
that a sampled signal matches the state after period $t.$ The first
observation is that there exists a constant $\bar{p}<1$ such that
for every $\eta$ and every sufficiently large $t,$ $\mathbb{P}[r(t)\le\bar{p}\mid\omega=1]\ge1-\eta$
uniformly over all strategies and virality weights. This is because
the arrivals of new signals in each period affect both $x(t)$ and
$z(t)$, and the law of large numbers implies that almost surely the
fraction of correct signals converges to $q<1.$ 

Second, let $d_{K}=2(\left\lfloor K/2\right\rfloor +1)-K\in\{1,2\}.$
Fix $a$ with $1/2<a<\bar{p}$ and $L$ with $L<(a/(1-a))^{d_{K}}$.
Then there exists $\eta>0$ so that for every random variable $R$
valued in $[0,1]$, if we have $\mathbb{P}[a\le R\le\bar{p}]\ge1-\eta$
then we must also have $\frac{\mathbb{E}[R^{k}(1-R)^{K-k}]}{\mathbb{E}[R^{K-k}(1-R)^{k}]}>L$
for every integer $k$ with $K/2<k\le K.$ This is because on the
event $\{a\le R\le\bar{p}\}$, we get $R^{k}(1-R)^{K-k}=(\frac{R}{1-R})^{2k-K}R^{K-k}(1-R)^{k}\ge(a/(1-a))^{d_{K}}\cdot R^{K-k}(1-R)^{k}$.
Also, let $m:=\min_{K/2<k\le K}\min_{r\in[a,\bar{p}]}r^{K-k}(1-r)^{k},$
so $m>0$ since its domain of minimization is compact and the minimand
is strictly positive on it. On the complement of the event $\{a\le R\le\bar{p}\},$
$R^{k}(1-R)^{K-k}$ is at least 0 and $R^{K-k}(1-R)^{k}$ is at most
1. So, $\frac{\mathbb{E}[R^{k}(1-R)^{K-k}]}{\mathbb{E}[R^{K-k}(1-R)^{k}]}\ge\frac{(a/(1-a))^{d_{K}}(1-\eta)m}{(1-\eta)m+\eta}$,
which converges to $(a/(1-a))^{d_{K}}>L$ as $\eta\to0.$ 

\textbf{Part 1}: Fix $0<\lambda\le\lambda^{*}$ and suppose $\sigma^{*}$
is a limit equilibrium. 

\textbf{Step 1}: Either $\sigma^* = \sigma^\text{maj}$, or all fixed points of $\phi_{\sigma^{*}}$ are strictly
informative.

We verify that $\sigma^{*}$ satisfies the hypotheses of Lemma \ref{lem:3_implies_1}.
Note $\sigma^{*}$ is the limit of a sequence of symmetric BNEs
$(\sigma^{(i)})$, where every $\sigma^{(i)}$ is state symmetric.
Also, in the $i$-th finite society under the equilibrium $\sigma^{(i)},$
belief about $\{\omega=1\}$ must be weakly higher after observing
$k$ positive signals and $s=1$ than $k$ positive signals and $s=-1$
for every $0\le k\le K.$ So by optimality of $\sigma^{(i)}$, we
have $\mathbb{E}[\sigma^{(i)}(1,k)]\ge\mathbb{E}[\sigma^{(i)}(-1,k)]$
for every $i$ and every $0\le k\le K.$ The limit $\sigma^{*}$ must
also satisfy state symmetry and $\mathbb{E}[\sigma^{*}(1,k)]\ge\mathbb{E}[\sigma^{*}(-1,k)]$
for every $0\le k\le K.$ Also, when $K$ is even, by the state symmetry of the equilibrium
$\sigma^{(i)}$ we know that a sample with $K/2$ positive signals
in society $i$ generates an equilibrium posterior belief that both
states are equally likely. Thus, optimality of $\sigma^{(i)}$ implies
$\sigma^{(i)}(1,K/2)(U_{K/2})=1$ and $\sigma^{(i)}(-1,K/2)(L_{K/2})=1$. The limit
$\sigma^{*}$ must then also satisfy $\sigma^{*}(1,K/2)(U_{K/2})=1$, $\sigma^{*}(-1,K/2)(L_{K/2})=1$.

If $\phi_{\sigma^{*}}$ has a strictly misleading fixed point and $\sigma^* \ne \sigma^\text{maj}$, that
is some $x\in[0,1]$ with $\lambda x+(1-\lambda)q<1/2$ and such that
$\phi_{\sigma^{*}}(x)=x$, then by Lemma \ref{lem:3_implies_1} we
get $\phi_{\sigma^{\text{maj}}}(x)<x.$ But we also have $\phi_{\sigma^{\text{maj}}}(0)>0,$
which means $\phi_{\sigma^{\text{maj}}}$ has a strictly misleading
fixed point in $(0,x)$ by the intermediate-value theorem, and further
$\phi_{\sigma^{\text{maj}}}$ will continue to have a nearby fixed
point for nearby values of $\lambda.$ Since $x\le1/2,$ this implies
for some $\lambda'<\lambda^{*},$ $\phi_{\sigma^{\text{maj}}}$ has
a fixed point in $[0,1/2],$ which contradicts the definition of $\lambda^{*}$.

If there is some $x$ with $\lambda x+(1-\lambda)q=1/2$ and such
that $\phi_{\sigma^{*}}(x)=x$, then by Lemma \ref{lem:3_implies_1}
we get $\phi_{\sigma^{\text{maj}}}(x)\le x.$ But since the sampling
accuracy at $x$ is exactly 1/2, every sample is as likely as its
mirror image, so the majority rule is expected to endorse at least $C/2$
correct signals out of $C$, hence $\phi_{\sigma^{\text{maj}}}(x)>1/2$
after accounting for the arrival of new signals that tend to match
the true state. This is a contradiction. Thus, every fixed point of
$\phi_{\sigma^{*}}$ must be strictly informative unless $\sigma^* = \sigma^\text{maj}$.

\textbf{Step 2}: If $\lambda < \lambda^*$, $\phi_{\sigma^{*}}$ only has fixed points in $(q,1].$  If $\lambda = \lambda^*$, either $\sigma^* = \sigma^\text{maj}$ or $\phi_{\sigma^{*}}$ only has fixed points in $(q,1].$
% if not sigma maj, this applies. If it is sigma maj and lambda < lambda*, by definition it can't  have fixed points in [0, 1/2], and conclusion here further rules out (1/2, q]

We first show  all fixed points of $\phi_{\sigma^{*}}$ are strictly
informative, except when $\lambda = \lambda^*$ and $\sigma^* = \sigma^\text{maj}$. If $\lambda < \lambda^*$ and $\sigma^* = \sigma^\text{maj}$, by definition of $\lambda^*$ all fixed points of $\phi_{\sigma^{*}}$ are strictly informative. And if $\sigma^* \ne \sigma^\text{maj}$, then by \textbf{Step 1},  all fixed points of $\phi_{\sigma^{*}}$ are strictly
informative. If $\lambda = \lambda^*$ and  $\sigma^* \ne \sigma^\text{maj}$, again  \textbf{Step 1} implies  all fixed points of $\phi_{\sigma^{*}}$ are strictly
informative.

We  verify that,  except when $\lambda = \lambda^*$ and $\sigma^* = \sigma^\text{maj}$, $\sigma^{*}$ is such that $\sigma^{*}(1,k)(U_k)=1$
for every $k\ge K/2,$ and thus satisfies the hypotheses of Lemma
\ref{lem:1_implies_2}. Since all fixed points of $\phi_{\sigma^{*}}$
are strictly informative, there exists some $\epsilon'>0$ so that
$\phi_{\sigma^{*}}(x)-x>2\epsilon'$ for every $x$ where $\lambda x+(1-\lambda)q\le0.5+2\epsilon'$.

Shrinking $\epsilon'$ if necessary, let $a_{0}:=0.5+\epsilon'/4<\bar{p}$.
Because $(a_{0}/(1-a_{0}))^{d_{K}}>1,$ the second observation above
gives some $\eta>0$ so that whenever $\mathbb{P}[a_{0}\le R\le\bar{p}]\ge1-\eta,$
we get $\frac{\mathbb{E}[R^{k}(1-R)^{K-k}]}{\mathbb{E}[R^{K-k}(1-R)^{k}]}>1$
for each $K/2<k\le K.$ Apply Lemma \ref{lem:drift_right} with $p=1/2$ and with probability
tolerance $\eta/4$ to obtain $N$ and $\delta.$ By the law of large
numbers and accounting for the finite seed pool, choose $N'$ large
enough so that for all $t\ge N'$, $\mathbb{P}[|z(t)-q|>\epsilon'/4]<\eta/4$
and $\mathbb{P}[r(t)>\bar{p}]<\eta/4$, the latter being possible
by the first observation above. Since $\sigma^{(i)}\to\sigma^{*}$,
for all large enough $i$ we get $\parallel\sigma^{(i)}-\sigma^{*}\parallel<\delta$
and $\max\{N,N'\}/n_{i}<\eta/4$. 

Let $R$ be the sampling accuracy at a uniformly random position in
 society $i$. At every position after $\max\{N,N'\},$ except on an
event with probability less than $3\eta/4,$ Lemma \ref{lem:drift_right} and construction
of $N'$ implies that sampling accuracy is at least $\lambda x(t)+(1-\lambda)q-\epsilon'/4\ge a_{0}$
and at most $\bar{p}$. The positions earlier than $\max\{N,N'\}$
have a total probability less than $\eta/4$. So, $\mathbb{P}[a_{0}\le R\le\bar{p}\mid\omega=1]\ge1-\eta.$ 

This implies after observing $k$ positive signals in the sample,
$\omega=1$ is strictly more likely if $k>K/2$. Hence, by optimality,
$\sigma^{(i)}(1,k)(U_k)=1$ for every $k>K/2.$ Also, for any belief
about sampling accuracy, a sample with $k=K/2$ is uninformative,
so if $K/2$ is an integer then $\sigma^{(i)}(1,K/2)(U_{K/2})=1$ by optimality.
Thus we see for all large enough $i,$ $\sigma^{(i)}(1,k)(U_k)=1$ for
every $k\ge K/2,$ hence the same must hold for the limit $\sigma^{*}$.

Combining Lemma \ref{lem:1_implies_2} (which rules out steady states at or lower than $q$ with a sampling accuracy strictly higher than 1/2) with the argument at the beginning of \textbf{Step 2} (which rules out steady states with sampling accuracy 1/2 or lower), we have completed this step. 

\textbf{Step 3}: $\sigma^{*}=\sigma^{\text{maj}}$.

By \textbf{Step 2}, we just need to establish this when $\phi_{\sigma^{*}}$ only has fixed points in $(q,1]$.
By state symmetry
it suffices to show that $\sigma^{*}(-1,k)(U_k)=1$ for every $k>K/2.$
Since $\phi_{\sigma^{*}}$ only has fixed points in $(q,1]$, there
exists some $\epsilon'>0$ so that $\phi_{\sigma^{*}}(x)-x\ge2\epsilon'$
for every $x$ where $\lambda x+(1-\lambda)q\le q+2\epsilon'$. 
Shrinking $\epsilon'$ if necessary, let $a_{q}:=q+\epsilon'/4<\bar{p}$.
Because $d_{K}\ge1$ and $a_{q}>q$, $(a_{q}/(1-a_{q}))^{d_{K}}>q/(1-q).$
Applying the second observation above with $L=q/(1-q)$ gives some
$\eta>0$ so that whenever $\mathbb{P}[a_{q}\le R\le\bar{p}]\ge1-\eta,$
we get $\frac{\mathbb{E}[R^{k}(1-R)^{K-k}]}{\mathbb{E}[R^{K-k}(1-R)^{k}]}>q/(1-q)$
for each $K/2<k\le K.$ Apply Lemma \ref{lem:drift_right} with $p=q$ and with probability
tolerance $\eta/4$. Repeating the construction in Step 2 (with $a_{q}$
in place of $a_{0}$) shows that, for all large enough society index
$i$, the random sampling accuracy of a uniformly randomly positioned
agent satisfies $\mathbb{P}[a_{q}\le R\le\bar{p}]\ge1-\eta$. This
implies after observing $k$ positive signals in the sample, posterior
belief in $\omega=1$ is strictly more than $q$ if $k>K/2$. 
By optimality, $\sigma^{(i)}(-1,k)(U_k)=1$.
So we also have in the limit $\sigma^{*}(-1,k)(U_k)=1$ for every $k>K/2.$

\textbf{Part 2}: For $\lambda < \lambda^*$, $\sigma^{\text{maj}}$ has a unique steady state.

By Part~1 and the definition of $\lambda^*$, all fixed points of $\phi_{\sigma^{\mathrm{maj}}}$ are strictly informative when $\lambda<\lambda^*$. Let $g(x):=\phi_{\sigma^{\mathrm{maj}}}(x)-x$. Since $g(0)>0$ and $g(1)<0$, there is at least one fixed point. Let $x^*$ be the smallest one. Then $x^*>1/2$, and $g$ is concave on $[x^*,1]$ by Lemma~\ref{lem:maj_shape}. Since $g$ is positive to the left of $x^*$ and $g(x^*)=0$, the right derivative of $g$ at $x^*$ is weakly negative. Concavity then implies $g(x)<0$ for every $x>x^*$ (unless $g$ were identically zero on a nontrivial interval, which is impossible because $g(1)<0$ and $g$ is a nonzero polynomial). Hence there is no second fixed point, so the steady state is unique. Combining this with Lemma~\ref{lem:1_implies_2} applied to $\sigma^{\mathrm{maj}}$ gives $x^*>q$.

\textbf{Part 3}: Suppose $\lambda^* < \infty$. When $\lambda=\lambda^{*}$, the unique equilibrium $\sigma^{\mathrm{maj}}$ (from Part~1) has a fixed point in $[0,1/2)$: the set $\{\lambda\in[0,1]:\exists\, x\in[0,1/2],\,\phi_{\sigma^{\mathrm{maj}}}(x,\lambda)=x\}$ is closed, so the infimum $\lambda^{*}$ is attained, and Lemma~\ref{lem:maj_ss_types} upgrades the fixed point from $[0,1/2]$ to $[0,1/2)$. The leftmost-root argument below then produces a strictly misleading steady state.

Now, suppose $\lambda>\lambda^{*}$ and suppose $\sigma^{*}$
is a limit equilibrium.

\textbf{Step 1}: $\phi_{\sigma^{*}}$ must have a weakly misleading
fixed point.

If not, then there exists some $\epsilon>0$ so that $\phi_{\sigma^{*}}(x)-x>\epsilon$
for every $x$ where $\lambda x+(1-\lambda)q\le0.5+\epsilon$. By
repeating the arguments in Part 1, Steps 2 and 3, we conclude $\sigma^{*}=\sigma^{\text{maj}}$.

But we show $\sigma^{\text{maj}}$ has a strictly misleading fixed
point for every $\lambda>\lambda^{*}$. By the definition of $\lambda^{*}$,
we can choose some $\lambda'$ with $\lambda^{*}\leq\lambda'<\lambda$
such that there exists a strictly misleading fixed point $x'$ under
$\sigma^{\text{maj}}$ at $\lambda'$ (we get ``strictly'' because
by Lemma \ref{lem:maj_ss_types}, 1/2 is not a fixed point of $\sigma^{\text{maj}}$
and all fixed points in $[0,1/2)$ are strictly misleading). We rewrite
the inflow accuracy function $\phi_{\sigma}(x)$ as $\phi_{\sigma}(x,\lambda)$
to make explicit its dependence on $\lambda$.

Observe $\phi_{\sigma}(x,\lambda)$ only depends on $x$ and $\lambda$
through the value of $\lambda x+(1-\lambda)q$. We can define $x$
by 
$
\lambda x+(1-\lambda)q=\lambda'x'+(1-\lambda')q.
$
Since $\lambda'<\lambda$ and $x'<q$, this equality implies that
$x>x'$. For $x'$ to be a strictly misleading fixed point under the
majority rule we must have $\lambda'x'+(1-\lambda')q<\frac{1}{2}$,
and therefore $x<\frac{1}{2}$ as well.

So 
$
\phi_{\sigma^{\text{maj}}}(x,\lambda)=\phi_{\sigma^{\text{maj}}}(x',\lambda')=x',
$
where the second equality holds because $x'$ is a fixed point under
$\sigma^{\text{maj}}$ and $\lambda'$. So we conclude that $\phi_{\sigma^{\text{maj}}}(x,\lambda)<x.$
Since $\phi_{\sigma^{\text{maj}}}(0,\lambda)>0,$ by the intermediate
value theorem there is some fixed point of $\phi_{\sigma^{\text{maj}}}$
between $0$ and $x$. Since $x<\frac{1}{2},$ this is a strictly
misleading fixed point, contradiction.

Note that since $\phi_{\sigma^{*}}(0)>0,$ the first weakly misleading
fixed point of $\phi_{\sigma^{*}}$ is stable at least from the left,
so it is also a weakly misleading steady state.

\textbf{Step 2}: $\phi_{\sigma^{*}}$ cannot have a fixed point with
a sampling accuracy of exactly 1/2.

Each $\sigma^{(i)}$, by optimality, has the property that $\mathbb{E}[\sigma^{(i)}(1,k)]\ge\mathbb{E}[\sigma^{(i)}(-1,k)]$
for every $0\le k\le K.$ So we must have $\mathbb{E}[\sigma^{*}(1,k)]\ge\mathbb{E}[\sigma^{*}(-1,k)]$
for each $0\le k\le K.$ Suppose $\lambda x+(1-\lambda)q=1/2$ and
$\phi_{\sigma^{*}}(x)=x.$ For each $0\le k<K/2,$ we get 
{\footnotesize
\begin{align*}
& P_{k}(x,\lambda)\!\left[q\,\E[\sigma^{*}(1,k)]+(1-q)\E[\sigma^{*}(-1,k)]\right]
 +P_{K-k}(x,\lambda)\!\left[q\,\E[\sigma^{*}(1,K-k)]+(1-q)\E[\sigma^{*}(-1,K-k)]\right] \\
=&\, P_{k}(x,\lambda)\Big[q\,\E[\sigma^{*}(1,k)]+(1-q)\E[\sigma^{*}(-1,k)]
 + q\,\E[\sigma^{*}(1,K-k)] + (1-q)\E[\sigma^{*}(-1,K-k)]\Big] 
 \quad \text{since }P_k(x,\lambda)=P_{K-k}(x,\lambda)\\
=&\, P_k(x,\lambda)\Big[C+(2q-1)\big(\E[\sigma^{*}(1,k)]-\E[\sigma^{*}(-1,k)]\big)\Big]\\
\ge&\, P_k(x,\lambda)C 
\quad \text{since }\E[\sigma^{*}(1,k)]\ge\E[\sigma^{*}(-1,k)],\;2q-1>0\\
\ge&\, P_k(x,\lambda)\tfrac{C}{2}+P_{K-k}(x,\lambda)\tfrac{C}{2}.
\end{align*}
}
If $K$ is even, the tie term is also at least $P_{K/2}(x,\lambda)C/2$: state symmetry gives $\E[\sigma^*(-1,K/2)]=C-\E[\sigma^*(1,K/2)]$, while optimality gives $\E[\sigma^*(1,K/2)]\ge \E[\sigma^*(-1,K/2)]$. So
\[
\phi_{\sigma^{*}}(x):=\frac{q+\sum_{k=0}^{K}P_{k}(x,\lambda)\cdot[q\cdot\mathbb{E}[\sigma^{*}(1,k)]+(1-q)\cdot\mathbb{E}[\sigma^{*}(-1,k)]]}{1+C}\ge\frac{q+C/2}{1+C}>1/2
\]
 since $q>1/2.$ But this means $\lambda\phi_{\sigma^{*}}(x)+(1-\lambda)q>1/2,$
contradiction.
\end{proof}

\begin{proof}[Proof of Proposition \ref{prop:informative}]
Let $\lambda<\lambda'<\lambda^{*}$ and suppose that $x^{*}$ is a
steady state under $\lambda$. We want to show that there exists a
steady state $(x')^{*}>x^{*}$ under $\lambda'$.

As in the proof of Part 3 of Theorem \ref{thm:equilibrium_ss_by_lambda},
let $\phi_{\sigma}(x,\lambda)$ be the inflow accuracy function with
its dependence on $\lambda$. Under the  majority rule, for any $0<C<K$ and for any private signal realization, the expected number of positive signals endorsed  is weakly increasing and nonconstant in the number of positive sampled signals. Hence the binomial expectation defining $\phi_{\sigma^{\text{maj}}}$ is strictly increasing in sampling accuracy, so $\phi_{\sigma^{\text{maj}}}(x,\lambda)$
is strictly increasing in $\lambda$ when $x>q$. By Theorem \ref{thm:equilibrium_ss_by_lambda},
we have $x^{*}>q$ and therefore 
$
x^{*}=\phi_{\sigma^{\text{maj}}}(x^{*},\lambda)<\phi_{\sigma^{\text{maj}}}(x^{*},\lambda').
$

Since $\phi_{\sigma^{\text{maj}}}(1,\lambda')<1,$ by the intermediate
value theorem there exists $(x')^{*}\in(x^{*},1)$ such that $
\phi_{\sigma^{\text{maj}}}((x')^{*},\lambda')=(x')^{*}.$ This is a steady state under $\lambda'$ that is greater than $x^{*}.$\end{proof}

\newpage
\setcounter{page}{1}
\begin{center}
{ \Large\textbf{Online Appendix}}{\Large\par}
\end{center}
\section{Omitted Proofs}

\begin{proof}[Proof of Lemma~\ref{lem:1_implies_2}]
Suppose by way of contradiction that such a fixed point $x$ exists. Let $y=\lambda x+(1-\lambda)q$ be the sampling accuracy, and note $x\le y\le q$, with $y>1/2$. Write $a_k=\mathbb{E}[\sigma(1,k)]$ and $b_k=\mathbb{E}[\sigma(-1,k)]$. The hypothesis and state symmetry imply $a_k\ge b_k$ for every $k$: if $k\ge K/2$, then $a_k=U_k\ge b_k$, while if $k<K/2$, then $b_k=C-a_{K-k}=C-U_{K-k}=L_k\le a_k$. Therefore the expected number $S$ of positive signals endorsed when $\omega=1$ satisfies
\begin{equation}
\label{eq:S_bound}
S=q\sum_k P_k a_k+(1-q)\sum_k P_k b_k\ge y\sum_k P_k a_k+(1-y)\sum_k P_k b_k. 
\end{equation} 
For fixed $y$, the right-hand side of Equation (\ref{eq:S_bound}) is minimized, subject to the hypotheses, by setting $a_k=U_k$ for all $k<K/2$ for which $K-2k>1$; when $K-2k=1$ the value of $a_k$ is irrelevant. Indeed, the partial derivative of the right-hand side with respect to $a_k$ is $yP_k-(1-y)P_{K-k}=yP_k(1-(y/(1-y))^{K-2k-1})$. Thus the right-hand side is minimized  by the strategy that endorses $U_k$ positive signals after private signal $1$ and $L_k$ positive signals after private signal $-1$, for any number $k$ of positive signals in the sample.

Let $B\sim\operatorname{Binom}(K,y)$. The difference between this minimum and $yC$ is
\[
y\,\mathbb E[U_B]+(1-y)\mathbb E[L_B]-yC
=
y\,\mathbb E[U_B-C]+(1-y)\mathbb E[L_B].
\]
Since
\[
C-U_B=\sum_{m=0}^{C-1}\mathbf 1\{B\le m\},
\qquad
L_B=\sum_{m=0}^{C-1}\mathbf 1\{B\ge K-m\},
\]
this difference equals
\[
\sum_{m=0}^{C-1}
\Big((1-y)\Pr[B\ge K-m]-y\Pr[B\le m]\Big).
\] 

Expanding the two tails, this is equal to
$$\sum_{m=0}^{C-1}
\left[
\sum_{b=K-m}^{K}(1-y)\Pr[B=b]
-
\sum_{b=0}^{m}y\Pr[B=b]
\right].$$
Let $\widetilde B\sim\operatorname{Binom}(K+1,y)$ and write
$Q_j=\Pr[\widetilde B=j]$. Using
\[
(1-y)\Pr[B=b]=\frac{K+1-b}{K+1}Q_b,\qquad
y\Pr[B=b]=\frac{b+1}{K+1}Q_{b+1},
\]
the expression becomes
\[
\frac{1}{K+1}
\sum_{m=0}^{C-1}
\left[
\sum_{b=K-m}^{K}(K+1-b)Q_b
-
\sum_{b=0}^{m}(b+1)Q_{b+1}
\right].
\]
In the first inner sum, set $c=K+1-b$; in the second, set $c=b+1$. Then
\[
\sum_{b=K-m}^{K}(K+1-b)Q_b
=
\sum_{c=1}^{m+1}cQ_{K+1-c},
\]
and
\[
\sum_{b=0}^{m}(b+1)Q_{b+1}
=
\sum_{c=1}^{m+1}cQ_c.
\]
Therefore the expression is
\[
\frac{1}{K+1}
\sum_{m=0}^{C-1}
\sum_{c=1}^{m+1}
c\left(Q_{K+1-c}-Q_c\right).
\]
Switching the order of summation, each fixed $c\in\{1,\ldots,C\}$ appears for
$m=c-1,\ldots,C-1$, i.e. for $C-c+1$ values of $m$. Hence the expression equals
\[
\frac{1}{K+1}\sum_{c=1}^{C} c(C-c+1)
\Big(\Pr[\operatorname{Binom}(K+1,y)=K+1-c]
-
\Pr[\operatorname{Binom}(K+1,y)=c]\Big).
\]

Since $y>1/2$, the binomial mass is larger at $K+1-c$ than at $c$ whenever $c<(K+1)/2$. For the summand with the index $c>(K+1)/2$, pair it with the summand with index $K+1-c<c$; the  binomial differences in the two summands are opposite, while the coefficient difference is
\[
(K+1-c)(C-K-1+c+1)-c(C-c+1)=(K-C)(2c-K-1)>0,
\]
so sum of the two paired summands is  positive. Unpaired summands are positive, and at least one such summand is strictly positive because $0<C<K$. Thus the minimized right-hand side of Equation (\ref{eq:S_bound}) is strictly larger than $yC$. Hence $S>yC$, and so
\[
\phi_{\sigma}(x)>\frac{q+yC}{1+C}\ge \frac{y+yC}{1+C}=y\ge x,
\]
contradicting that $x$ is a fixed point.
\end{proof}

\begin{proof}[Proof of Lemma~\ref{lem:maj_shape}]
Let $p$ denote sampling accuracy and write
\[
A(p):=\sum_{k=0}^K \binom Kk p^k(1-p)^{K-k}\beta_k,
\qquad 
 \beta_k:=q\,\mathbb{E}[\sigma^{\mathrm{maj}}(1,k)]+(1-q)\,\mathbb{E}[\sigma^{\mathrm{maj}}(-1,k)].
\]
Then $\phi_{\sigma^{\mathrm{maj}}}(x)=(q+A(\lambda x+(1-\lambda)q))/(1+C)$, so its curvature in $x$ is the same as the curvature of $A(p)$ in $p$, up to the positive factor $\lambda^2/(1+C)$ when $\lambda>0$. If $\lambda=0$, the function is constant and the claim is immediate.

Recall $L_k=\max\{0,C+k-K\}$ and $U_k=\min\{k,C\}$. Under the majority rule,
\[
 \beta_k=
 \begin{cases}
 L_k, & k<K/2,\\
 qU_k+(1-q)L_k, & k=K/2 \text{ if } K \text{ is even},\\
 U_k, & k>K/2.
 \end{cases}
\]
Using an identity for the second derivative of the Bernstein polynomial,
\[
A''(p)=K(K-1)\sum_{j=0}^{K-2}\binom{K-2}{j}p^j(1-p)^{K-2-j}\Delta^2 \beta_j,
\]
where $\Delta^2\beta_j=\beta_{j+2}-2\beta_{j+1}+\beta_j$. Let $d_j:=\beta_{j+1}-\beta_j$ for $j=0,\ldots,K-1$. A direct calculation from the definition of majority rule gives the following first-difference sequences. If $K=2h+1$ is odd and $C\le h$, then
\[
(d_j)_{j=0}^{K-1}=(0,\ldots,0,C,0,\ldots,0).
\]
If $K=2h+1$ is odd and $C>h$, writing $r:=K-C$, then
\[
(d_j)_{j=0}^{K-1}=(0,\ldots,0,1,\ldots,1,r+1,1,\ldots,1,0,\ldots,0).
\]
If $K=2h$ is even and $C\le h$, then
\[
(d_j)_{j=0}^{K-1}=(0,\ldots,0,qC,(1-q)C,0,\ldots,0).
\]
Finally, if $K=2h$ is even and $C>h$, writing $r:=K-C$, then
\[
(d_j)_{j=0}^{K-1}=(0,\ldots,0,1,\ldots,1,1+qr,1+(1-q)r,1,\ldots,1,0,\ldots,0).
\]
Because $q>1/2$, each of these sequences is weakly increasing and then weakly decreasing. Hence the signs of $\Delta^2\beta_j=d_{j+1}-d_j$, after zeroes are omitted, consist of some positive signs followed by some negative signs.

Set $t=p/(1-p)$. Since $(1-p)^{K-2}>0$ on $(0,1)$,
\[
\frac{A''(p)}{K(K-1)}=(1-p)^{K-2}\sum_{j=0}^{K-2}\binom{K-2}{j}\Delta^2\beta_j\,t^j.
\]
The coefficients on the non-zero terms of the polynomial in $t$ have at most one sign change. By Descartes' rule of signs, it has at most one positive root. Thus $A''$ has at most one zero in $(0,1)$. Moreover, whenever both signs occur, the first nonzero coefficient is positive and the last nonzero coefficient is negative, so the only possible curvature switch is from convex to concave.

It remains to show that the possible switch occurs weakly before $p=1/2$. At $p=1/2$,
\[
\frac{2^{K-2}A''(1/2)}{K(K-1)}=\sum_{j=0}^{K-2}\binom{K-2}{j}\Delta^2\beta_j.
\]
When $K$ is odd, the positive and negative terms in the first-difference calculation above occur in symmetric pairs with equal binomial weights, so this sum is zero. When $K=2h$ is even, using the convention $\binom{K-2}{m}=0$ for $m\notin\{0,\ldots,K-2\}$, the same calculation gives
\[
\sum_{j=0}^{K-2}\binom{K-2}{j}\Delta^2\beta_j
=\eta(2q-1)\left[\binom{K-2}{h-2}-\binom{K-2}{h-1}\right]\le0,
\]
where $\eta=C$ if $C\le h$ and $\eta=K-C$ if $C>h$. The inequality follows because $\binom{K-2}{h-1}\ge \binom{K-2}{h-2}$. Therefore $A''(1/2)\le0$. Since $A''$ has at most one sign change and any sign change is from positive to negative, $A''(p)\le0$ for every $p\ge1/2$. Thus $\phi_{\sigma^{\mathrm{maj}}}$ is concave whenever sampling accuracy is at least $1/2$.
\end{proof}

\begin{proof}[Proof of Proposition \ref{prop:distributions}]
Throughout, write $\mathcal{F}_{t}$ for the sigma-field generated
by all signals, all sampling randomness, and all endorsement randomness
(including the mixing randomness of $\sigma$) up to time $t$. 
Fix any state-symmetric strategy $\sigma(s,k).$  
Let $\mathcal P_t$ be the signal pool after $t$ agents have acted, including seed signals and agents'  signals, and let $\rho_j(t)$ be the score of signal $j\in\mathcal P_t$ at the end of period $t.$ Let total score at the end of period $t$ be $S_t:=\sum_{j\in\mathcal P_t}\rho_j(t).$ If $S_0$ is the realized initial total seed score, then
\begin{equation}
S_t=S_0+(1+C)t.\label{eq:total-score}
\end{equation}
The pool size is $N_t=n_0+t$.

\textbf{Step 1: Sampling is affine preferential attachment (with a
vanishing $t$-dependence).}

Fix $t\ge0$ and a signal $j\in\mathcal P_t$ with score $\rho_j(t)$.
A sample slot at time $t+1$ samples signal $j$ with probability $\lambda\frac{\rho_j(t)}{S_t}+(1-\lambda)\frac{1}{N_t}$. Define the \emph{time-$t$ attractiveness offset}
\begin{equation}
\delta_t:=\frac{1-\lambda}{\lambda}\cdot\frac{S_t}{N_t}.\label{eq:delta-t}
\end{equation}
Using \eqref{eq:total-score}, this is deterministic conditional on the seed realization and satisfies
\[
\delta_t=\frac{1-\lambda}{\lambda}\frac{S_0+(1+C)t}{n_0+t}=\delta+O\!\left(\frac{1}{t}\right),\qquad\delta=\frac{(1-\lambda)(1+C)}{\lambda}.
\]
The sampling probability can be rewritten exactly as
\begin{equation}
\Pr(\text{sample }j)=\frac{\rho_j(t)+\delta_t}{\sum_{\ell\in\mathcal P_t}(\rho_\ell(t)+\delta_t)}.\label{eq:affine-PA}
\end{equation}

The finite seed pool contributes only finitely many initial nodes. It may affect finite-time probabilities and the scores of those seed signals, but it has zero limiting empirical mass among signals. The argument below therefore tracks score counts for endogenously arriving signals; seed terms enter the normalized recursions only as $O(1/t)$ perturbations. Thus each sample slot is a preferential
attachment draw with attractiveness determined by the affine function  $\rho\mapsto\rho+\delta_{t}$, where $\delta_{t}\to\delta$.

\textbf{Step 2: Probabilities of endorsing signals.}

Define the (time-$t$) affine weights $w_{j}(t):=\rho_{j}(t)+\delta_{t}$. Define the total attractiveness among type-$\theta$ signals $W_{\theta}(t):=\sum_{i:\text{sign}(s_{i})=\theta}w_{i}(t)$ for $\theta\in\{+,-\}.$
At time $t+1$, let $k_{t+1}$ be the number of positive signals in
the size-$K$ sample, let $B_{t+1}:=\mathbf{1}\{s_{t+1}=1\}$, and let
$A_{t+1}\sim\sigma(s_{t+1},k_{t+1})$ be the (random) number of positive
signals endorsed.

The next lemma shows that the within-type identity of a sampled signal
is determined solely by the affine weights, regardless of the strategy.
\begin{lem}
\label{lem:within-type-slot} Fix $t\ge K$ and condition on $\mathcal{F}_{t}$.
For each $1\le J\le K$ and $j$ with $s_{j}=1,$ 
\[
\Pr(\text{signal }j\text{ in slot }J\mid\text{positive signal in slot }J,\mathcal{F}_{t})=\frac{w_{j}(t)}{W_{+}(t)}.
\]
Moreover, let $N_{j}$ be the number of times the positive signal $s_{j}$
appears among the $K$ sampled slots and let $k$ be the number of
sample slots with positive signals. Then for any $k\ge1$, 
\[
\mathbb{E}\!\left[\frac{N_{j}}{k}\,\middle|\,k,\mathcal{F}_{t}\right]=\frac{w_{j}(t)}{W_{+}(t)}.
\]
The same statements hold with $+$ replaced by $-$. 
\end{lem}
\begin{proof}
Conditional on $\mathcal{F}_{t}$, each slot is an independent draw
from \eqref{eq:affine-PA}. For positive signal $j$
\begin{align*}
\Pr(\text{signal }j\text{ in slot }J\mid\text{positive signal in slot }J,\mathcal{F}_{t}) & =\frac{\Pr(\text{signal }j\text{ in slot }J\mid\mathcal{F}_{t})}{\Pr(\text{positive signal in slot }J\mid\mathcal{F}_{t})}\\
 & =\frac{w_{j}(t)/\sum_{\ell\in\mathcal P_t}w_{\ell}(t)}{W_{+}(t)/\sum_{\ell\in\mathcal P_t}w_{\ell}(t)}=\frac{w_{j}(t)}{W_{+}(t)}.
\end{align*}
For the second claim, conditional on $(k,\mathcal{F}_{t})$, the identities
of the $k$ positive sample signals are i.i.d.\ with the above
distribution among the positive signals, by exchangeability of the
$K$ slots. Hence the expected fraction of those $k$ slots that signal
$j$ occupies is precisely $w_{j}(t)/W_{+}(t)$. The argument for
negative signals is identical.
\end{proof}
\textbf{Step 3: Conditional drift of the score counts under a general
strategy.}

Fix a positive signal $j$ and let $\Delta_{j}(t+1)$ be its total
score increment at time $t+1$ (the number of times it is endorsed from
the sample, which may be 2 or larger if it appears multiple times
there). Let $N_{j}$ and $k=k_{t+1}$ be as in Lemma~\ref{lem:within-type-slot}.

Given the realized sample and the endorsement decision $A_{t+1}=a$ with
$1\le a\le k$, the agent selects $a$ of the $k$ positive sample slots
uniformly at random, so
\[
\mathbb{E}[\Delta_{j}(t+1)\mid\text{sample},\,A_{t+1}=a]=a\,\frac{N_{j}}{k}.
\]
When $a=0$ (or when $k=0$, which forces $a=0$ by feasibility),
$\Delta_{j}(t+1)=0$. Since the right-hand side is linear in $a$,
\begin{equation}
\mathbb{E}[\Delta_{j}(t+1)\mid\text{sample},\,s_{t+1}]=\bar{\sigma}(s_{t+1},k)\,\frac{N_{j}}{k},\label{eq:Delta-general-feed}
\end{equation}
with the convention $0\cdot(N_{j}/0):=0$ when $k=0$ (vacuous since
$\bar{\sigma}(s,0)=0$).

Conditional on $(k,\mathcal{F}_{t})$, the identities of the signals
in the sample are independent of the realization of the private
signal $s_{t+1}$ and the randomization of the strategy $\sigma$.
Applying Lemma~\ref{lem:within-type-slot} and the tower property,
conditioning on $(k,\mathcal{F}_{t})$ we can factorize 
\begin{equation}\label{eq:Delta-given-k}
\mathbb{E}[\Delta_{j}(t+1)\mid k,\mathcal{F}_{t}]  =\mathbb{E}[\bar{\sigma}(s_{t+1},k)\mid k]\cdot\frac{w_{j}(t)}{W_{+}(t)}\nonumber  =\bar{a}(k)\cdot\frac{w_{j}(t)}{W_{+}(t)},
\end{equation}
where  $\bar{a}(k):=q\,\bar{\sigma}(1,k)+(1-q)\,\bar{\sigma}(-1,k)\label{eq:abar-k}$ 
is the expected number of positive endorsements when there are $k$ positive
sampled signals, averaged over the private-signal realization. Taking the
expectation over $k$, 
\begin{equation}
\mathbb{E}[\Delta_{j}(t+1)\mid\mathcal{F}_{t}]=a_{t}^{+}\cdot\frac{w_{j}(t)}{W_{+}(t)},\qquad a_{t}^{+}:=\sum_{k=0}^{K}P_{k}(p_t)\,\bar{a}(k),\label{eq:Delta-factor}
\end{equation}
where $p_{t}:=W_{+}(t)/\sum_{\ell\in\mathcal P_t}w_{\ell}(t)$ and $P_k(p):=\binom{K}{k}p^{k}(1-p)^{K-k}$.

\medskip{}

Now let $Z_{+,r}(t)$ be the number of positive signals with score
exactly $r$ at time $t$. Each endorsement of a positive score-$r$ signal
moves it to a higher score class, and each endorsement of a score-$(r-1)$
signal may bring it to score $r$. Also, a new positive signal arrives
with probability $q$ and enters at score $1$. The conditional drift
of $Z_{+,r}$ is 
\begin{equation}
\mathbb{E}\!\left[Z_{+,r}(t+1)-Z_{+,r}(t)\mid\mathcal{F}_{t}\right]=q\,\mathbf{1}_{\{r=1\}}+\frac{a_{t}^{+}}{W_{+}(t)}\Big((r-1+\delta_{t})Z_{+,r-1}(t)-(r+\delta_{t})Z_{+,r}(t)\Big)+\epsilon_{t}^{(r)},\label{eq:master-plus-random}
\end{equation}
where $\epsilon_{t}^{(r)}$ accounts for multi-increment events (a
signal's score increasing by $\ge2$ in a single period due to appearing
in multiple endorsed sample slots). Three mechanisms contribute:

\emph{Inflow overcounting (from class $r-1$):} a score-$(r-1)$ signal
that gains $\ge2$ points is counted 
in the inflow term of class $r$, but it overshoots to score $r+1$
or higher and should not contribute to $Z_{+,r}(t+1)$.

\emph{Outflow overcounting:} a score-$r$ signal that gains $\ge2$
points leaves class $r$ (counted by the $(r+\delta_{t})Z_{+,r}$
term) but arrives in class $r+2$ or higher rather than class $r+1$.

\emph{Skipped-class inflow:} a score-$(r-2)$ or lower signal that
gains $\ge2$ points may enter class $r$ directly, bypassing class
$r-1$.

For fixed score $r$, a signal $j$ with $\rho_{j}(t)=r$ is sampled
into two or more of the $K$ sample slots with probability at most $\binom{K}{2}(w_{j}(t)/\sum_{\ell}w_{\ell}(t))^{2}=O(K^{2}(r+\delta)^{2}/t^{2})$.
Summing over all $Z_{+,r}(t)\le t$ signals of score $r$ gives total
multi-increment probability $O(K^{2}(r+\delta)^{2}/t)$. Each such
event changes $Z_{+,r}$ by at most $O(C)$, so the outflow overcounting
contributes at most $O(K^{2}C(r+\delta)^{2}/t)$ to $\epsilon_{t}^{(r)}$.
For inflow overcounting  and skipped-class inflow, the same bound applies to signals of scores
$1,\dots,r-1$: for each score $r'\le r-1$, multi-increment probability
is $O(K^{2}(r'+\delta)^{2}/t)$. Summing over the finitely many classes
$r'=1,\dots,r-1$ (with $r$ fixed) gives a total inflow overcounting and skipped-class inflow
contribution that is also $O(1/t)$, with constants depending on $r$,
$K$, $C$, and $\delta$. Hence $|\epsilon_{t}^{(r)}|=O(1/t)$ for
each fixed $r$. After dividing by $t+1$ in the normalized recursion
below, this becomes $O(1/t^{2})=o(1/t)$ and is absorbed into the
error term.

An identical argument gives the corresponding equation for the negative
class, with $q$ replaced by $1-q$, $a_{t}^{+}$ replaced by $a_{t}^{-}:=C-a_{t}^{+}$,
and $W_{+}(t)$ replaced by $W_{-}(t)$.

\textbf{Step 4: Coefficients are asymptotically constant.}

\emph{(i) Limit of the expected endorsement inflows via score-mass accounting.}
Let $X_{+}(t):=\sum_{i\in\mathcal P_t:\,s_{i}=1}\rho_{i}(t)$ be total positive
score. For $t\ge0$, each period adds $B_{t+1}$ (a new positive
signal contribution) plus $A_{t+1}$ (positive endorsements) to $X_{+}(t)$,
while total score $S_{t}$ increases by $1+C$. Since $x(t)=X_{+}(t)/S_{t}\to x^{*}$
and $S_{t}/t\to1+C$, the average positive inflow per period must
satisfy 
\begin{equation}
x^{*}=\frac{q+m_{+}}{1+C},\qquad\text{so}\qquad m_{+}=(1+C)x^{*}-q,\label{eq:mass-accounting}
\end{equation}
and then $m_{-}=C-m_{+}$. 

Equivalently, $m_{+}$ can be expressed directly via the strategy
as $m_{+}=\sum_{k=0}^{K}P_{k}^{*}\,\bar{a}(k)$,
where $P_{k}^{*}=\binom{K}{k}(p^{*})^{k}(1-p^{*})^{K-k}$ with $p^{*}:=\lambda x^{*}+(1-\lambda)q$.

\emph{(ii) Convergence of $a_{t}^{+}\to m_{+}$.} On the steady-state
event, $x(t)\to x^{*}$ and the empirical fraction of positive signals
$N_{+}(t)/t\to q$ by the law of large numbers. Combined with $\delta_{t}\to\delta$,
this implies $p_{t}\to p^{*}$, so $P_{k}(t)\to P_{k}^{*}$ and hence
$
a_{t}^{+}=\sum_{k}P_{k}(t)\bar{a}(k)\xrightarrow{}m_{+}=\sum_{k}P_{k}^{*}\bar{a}(k).
$

\emph{(iii) Limit of $W_{\theta}(t)/t$ and the effective attachment
rates.} Write $N_{+}(t):=N_0^++\sum_{i=1}^t\mathbf{1}\{s_{i}=1\}$ and
$N_{-}(t)=n_0+t-N_{+}(t)$. Since $W_{+}(t)=X_{+}(t)+\delta_{t}N_{+}(t)$
and $W_{-}(t)=(S_{t}-X_{+}(t))+\delta_{t}N_{-}(t)$, we have on the
steady-state event: $\frac{X_{+}(t)}{t}\to(1+C)x^{*}$, $\frac{N_{+}(t)}{t}\to q$, $\delta_{t}\to\delta$. Hence, 
\begin{equation}
\frac{W_{+}(t)}{t}\to(1+C)x^{*}+\delta q=m_{+}+q(1+\delta)=:\tau_{+},\qquad\frac{W_{-}(t)}{t}\to m_{-}+(1-q)(1+\delta)=:\tau_{-}.\label{eq:tau-limits}
\end{equation}
For all sufficiently large $t$, $W_{\theta}(t)>0$ since $\tau_{\theta}>0$
(which holds because $m_{\theta}>0$ and $q_{\theta}>0$).

Define the limiting effective attachment rates 
\begin{equation}
\alpha_{\theta}:=\frac{m_{\theta}}{\tau_{\theta}}\in(0,1),\qquad\theta\in\{+,-\}.\label{eq:alpha-def}
\end{equation}

\textbf{Step 5: Convergence of empirical frequencies via a (stochastic) recursion.}

Fix $\theta\in\{+,-\}$ and $r\ge1$. Let $Z_{\theta,r}(t)$ be the
number of type-$\theta$ signals with score exactly $r$ at time $t$,
and define the normalized empirical frequency $
U_{\theta,r}(t):=\frac{Z_{\theta,r}(t)}{t}$ for $t\ge K.$
Let $\Delta_{\theta,r}(t+1):=Z_{\theta,r}(t+1)-Z_{\theta,r}(t)$.
Then $\Delta_{\theta,r}(t+1)$ is bounded in absolute value by $C+1$
(at most one new signal enters score-$1$, and at most $C$ endorsements
move signals across score classes).

Define the martingale difference $\xi_{\theta,r}(t+1):=\Delta_{\theta,r}(t+1)-\mathbb{E}[\Delta_{\theta,r}(t+1)\mid\mathcal{F}_{t}]$, so that $\mathbb{E}[\xi_{\theta,r}(t+1)\mid\mathcal{F}_{t}]=0$ and
$|\xi_{\theta,r}(t+1)|\le2(C+1)$ a.s. From the identity 
\[
U_{\theta,r}(t+1)-U_{\theta,r}(t)=\frac{\Delta_{\theta,r}(t+1)-U_{\theta,r}(t)}{t+1}
\]
and the conditional drift \eqref{eq:master-plus-random} (or its $-$
analogue), we obtain the recursion 
\begin{equation}
U_{\theta,r}(t+1)=U_{\theta,r}(t)+\frac{1}{t+1}\Big(A_{\theta,r}(t)-B_{\theta,r}(t)\,U_{\theta,r}(t)+\xi_{\theta,r}(t+1)\Big)+\varepsilon_{\theta,r}(t),\label{eq:U-recursion}
\end{equation}
where 
$$
A_{\theta,r}(t)  :=q_{\theta}\mathbf{1}_{\{r=1\}}+\underbrace{\frac{t\,a_{t}^{\theta}}{W_{\theta}(t)}}_{=:\alpha_{\theta}(t)}\,(r-1+\delta_{t})\,U_{\theta,r-1}(t), \qquad
B_{\theta,r}(t)  :=1+\alpha_{\theta}(t)\,(r+\delta_{t}),
$$
and $\varepsilon_{\theta,r}(t)=\epsilon_{t}^{(r)}/(t+1)=o(1/t)$ absorbs
the multi-increment error from \eqref{eq:master-plus-random}. By
Step~4 and $\delta_{t}\to\delta$, we have $\alpha_{\theta}(t)\to\alpha_{\theta}$
and $B_{\theta,r}(t)\to1+\alpha_{\theta}(r+\delta)$.

We now record a one-dimensional convergence lemma for recursions of
the form \eqref{eq:U-recursion}.
\begin{lem}
\label{lem:almost-linear} Let $(u_{t})_{t\ge t_{0}}$ be a bounded
sequence, and let $(a_{t})$ and $(b_{t})$ be real sequences with
$b_{t}\to b>0$ and $a_{t}\to a$ as $t\to\infty$.

\noindent\textbf{(Deterministic)} If 
$u_{t+1}=u_{t}+\frac{1}{t+1}\big(a_{t}-b_{t}u_{t}\big)+\varepsilon_{t}$ with $\varepsilon_{t}=o\!\left(\frac{1}{t}\right)$,
then $u_{t}\to a/b$.

\smallskip{}
\noindent\textbf{(Stochastic)} Suppose $(\mathcal{F}_{t})$ is a filtration
and $(\xi_{t+1})$ is a martingale difference sequence with $\mathbb{E}[\xi_{t+1}\mid\mathcal{F}_{t}]=0$
and $|\xi_{t+1}|\le B$ a.s. If 
\[
u_{t+1}=u_{t}+\frac{1}{t+1}\big(a_{t}-b_{t}u_{t}+\xi_{t+1}\big)+\varepsilon_{t}\qquad\text{with}\qquad\varepsilon_{t}=o\!\left(\frac{1}{t}\right),
\]
and $b_{t}\to b>0$, $a_{t}\to a$ almost surely (or deterministically),
then $u_{t}\to a/b$ almost surely (hence also in probability). 
\end{lem}
\begin{proof}
Write $v_{t}:=u_{t}-a/b$. Then 
\begin{equation}
v_{t+1}=v_{t}-\frac{b_{t}}{t+1}v_{t}+\frac{\eta_{t}}{t+1}+\frac{\xi_{t+1}}{t+1}+\varepsilon_{t},\qquad\eta_{t}:=a_{t}-a+\frac{a}{b}(b-b_{t}).\label{eq:v-rec}
\end{equation}

\smallskip{}
\emph{Deterministic case.} If $\xi_{t+1}\equiv0$, then $\eta_{t}\to0$
and \eqref{eq:v-rec} is a stable ``Euler'' discretization of $\dot{v}=-bv$.
A direct product-sum expansion yields 
\[
v_{t}=\Big(\prod_{s=t_{0}}^{t-1}\big(1-\tfrac{b_{s}}{s+1}\big)\Big)v_{t_{0}}+\sum_{s=t_{0}}^{t-1}\Big(\prod_{u=s+1}^{t-1}\big(1-\tfrac{b_{u}}{u+1}\big)\Big)\Big(\frac{\eta_{s}}{s+1}+\varepsilon_{s}\Big).
\]
Since $b_{s}\to b>0$, the product $\prod_{s=t_{0}}^{t-1}(1-b_{s}/(s+1))$
decays as $O(t^{-b/2})$ for large $t$, killing the initial condition.
For the sum, given any $\epsilon>0$, choose $S$ such that $|\eta_{s}|<\epsilon$
and $|\varepsilon_{s}|<\epsilon/s$ for $s\ge S$. The portion $s<S$
is killed by the product decay, and the tail $s\ge S$ is bounded
by $\epsilon\sum_{s\ge S}\prod_{u=s+1}^{t-1}(1-b_{u}/(u+1))/(s+1)\lesssim\epsilon/b$.
Letting $\epsilon\to0$ gives $v_{t}\to0$, i.e.\ $u_{t}\to a/b$.

\smallskip{}
\noindent\emph{Stochastic case.} The recursion \eqref{eq:v-rec} is a Robbins--Monro
stochastic approximation of the form 
\[
v_{t+1}=v_{t}+\frac{1}{t+1}\bigl(h_{t}(v_{t})+\xi_{t+1}\bigr)+\varepsilon_{t},\qquad h_{t}(v):=-b_{t}v+\eta_{t},
\]
with limiting drift $h(v)=-bv$. The ODE $\dot{v}=h(v)=-bv$ has the
globally asymptotically stable equilibrium $v^{*}=0$ (since $b>0$).
We verify the hypotheses for almost-sure convergence of stochastic
approximation (\cite{borkar2009stochastic} Chapter 2):
\begin{itemize}
\item[\textbf{(A1)}] Lipschitz continuity follows from linearity.
\item[\textbf{(A2)}] The step sizes $\gamma_{t}=1/(t+1)$ satisfy $\sum\gamma_{t}=\infty$
and $\sum\gamma_{t}^{2}<\infty$.
\item[\textbf{(A3)}] $\mathbb{E}[\xi_{t+1}\mid\mathcal{F}_{t}]=0$
and $|\xi_{t+1}|\le B$.
\item[\textbf{(A4)}] $(v_{t})$ is bounded since $(u_{t})$ is bounded
by hypothesis.

\end{itemize}
\noindent 
The vanishing measurement error condition from Section 2.2 also holds because $\varepsilon_{t}=o(\gamma_{t})$ and $\sup_{|v|\le M}|h_{t}(v)-h(v)|=|\eta_{t}|+|b_{t}-b|\cdot M\to0$
a.s.\ for each $M$ since $\eta_{t}\to0$ and $b_{t}\to b$. By the convergence theorem for stochastic approximation
with asymptotically stable equilibria (the extension of Theorem~2.1 to vanishing measurement error in Section~2.2 of \cite{borkar2009stochastic}), $v_{t}\to0$
almost surely, i.e.\ $u_{t}\to a/b$ almost surely.
\end{proof}
\textbf{Step 6: Apply Lemma~\ref{lem:almost-linear} by induction
over $r$ to find the stationary distribution.}

Fix $\theta\in\{+,-\}$ and define $q_{\theta}$ and $\alpha_{\theta}$
as above. We show that $U_{\theta,r}(t)\to f_{\theta,r}$ in probability
for each fixed $r$, where $\{f_{\theta,r}\}_{r\ge1}$ solves the
stationary equation system.

\emph{Base case $r=1$.} In \eqref{eq:U-recursion}, $U_{\theta,0}(t)\equiv0$,
so $A_{\theta,1}(t)=q_{\theta}$ and $B_{\theta,1}(t)\to1+\alpha_{\theta}(1+\delta)$.
Lemma~\ref{lem:almost-linear} (stochastic case) yields 
$U_{\theta,1}(t)\xrightarrow{a.s.}f_{\theta,1}:=\frac{q_{\theta}}{1+\alpha_{\theta}(1+\delta)}.$

\emph{Induction step.} Assume $U_{\theta,r-1}(t)\xrightarrow{a.s.}f_{\theta,r-1}$.
Then using $\alpha_{\theta}(t)\to\alpha_{\theta}$ and $\delta_{t}\to\delta$,
we have 
\[
A_{\theta,r}(t)=q_{\theta}\mathbf{1}_{\{r=1\}}+\alpha_{\theta}(t)(r-1+\delta_{t})U_{\theta,r-1}(t)\xrightarrow{a.s.}a_{\theta,r}:=\alpha_{\theta}(r-1+\delta)f_{\theta,r-1},
\]
since $r\ge 2$ in the induction step, $\mathbf{1}_{\{r=1\}}=0$. Also, 
$
B_{\theta,r}(t)=1+\alpha_{\theta}(t)(r+\delta_{t})\xrightarrow{}b_{\theta,r}:=1+\alpha_{\theta}(r+\delta).
$
Applying Lemma~\ref{lem:almost-linear} (stochastic case) to \eqref{eq:U-recursion}
yields 
\[
U_{\theta,r}(t)\xrightarrow{a.s.}f_{\theta,r}:=\frac{a_{\theta,r}}{b_{\theta,r}}=\frac{\alpha_{\theta}(r-1+\delta)f_{\theta,r-1}}{1+\alpha_{\theta}(r+\delta)}.
\]
Rearranging gives the stationary  equation system
\begin{equation}
f_{\theta,r}=q_{\theta}\mathbf{1}_{\{r=1\}}+\alpha_{\theta}\Big((r-1+\delta)f_{\theta,r-1}-(r+\delta)f_{\theta,r}\Big),\qquad r\ge1,\label{eq:stationary}
\end{equation}
with the convention $f_{\theta,0}=0$.

Since $U_{\theta,r}(t)\to f_{\theta,r}$ for each fixed $r$ and $\sum_{r\ge1}U_{\theta,r}(t)=N_{\theta}(t)/t\to q_{\theta}$,
Fatou's lemma gives $\sum_{r\ge1}f_{\theta,r}\le q_{\theta}$. To
establish equality (no escape of mass to infinity), note that for
any $R\ge1$, 
\[
\sum_{r>R}U_{\theta,r}(t)\le\frac{1}{R}\sum_{r\ge1}r\,U_{\theta,r}(t)=\frac{\text{total type-}\theta\text{ score}}{Rt}.
\]
On the steady-state event, the total type-$\theta$ score divided
by $t$ converges to a finite constant (e.g., $(1+C)x^{*}$ for type~$+$),
so the right-hand side is $O(1/R)$ uniformly in large~$t$. Taking
$R\to\infty$ shows $\limsup_{t\to\infty}\sum_{r>R}U_{\theta,r}(t)\to0$,
which together with pointwise convergence gives $\sum_{r\ge1}f_{\theta,r}=q_{\theta}$.

The within-type empirical distribution therefore satisfies $g_{\theta}(r)=\frac{f_{\theta,r}}{q_{\theta}}$.

\textbf{Step 7: Solve the stationary equation system.}

Let $\tau_{\theta}=m_{\theta}+q_{\theta}(1+\delta)$ and $\alpha_{\theta}=m_{\theta}/\tau_{\theta}$.
Define $\gamma_{\theta}:=1+\tau_{\theta}/m_{\theta}=2+\frac{q_{\theta}(1+\delta)}{m_{\theta}}$,
which is \eqref{eq:gamma-def}. From \eqref{eq:stationary}, for $r\ge2$ we obtain the ratio recursion
\[
f_{\theta,r}=f_{\theta,r-1}\cdot\frac{r-1+\delta}{r+\delta+\tau_{\theta}/m_{\theta}}=f_{\theta,r-1}\cdot\frac{r-1+\delta}{r+\gamma_{\theta}+\delta-1}.
\]
Iterating and normalizing gives the Gamma-function form \eqref{eq:gtheta}:
\[
g_{\theta}(r)=(\gamma_{\theta}-1)\,\frac{\Gamma(\gamma_{\theta}+\delta)}{\Gamma(1+\delta)}\cdot\frac{\Gamma(r+\delta)}{\Gamma(r+\gamma_{\theta}+\delta)}.
\]

\textbf{Step 8: Tail behavior.} Using $\Gamma(r+\delta)/\Gamma(r+\gamma_{\theta}+\delta)\sim r^{-\gamma_{\theta}}$
as $r\to\infty$ implies the desired tail distribution. Since $\gamma_{\theta}>2$
whenever $m_{\theta}>0$, the limiting within-type distribution has
finite mean.
\end{proof}

\begin{proof}[Proof of Corollary \ref{cor:tails}]
We have $\gamma_{-}-\gamma_{+}=(1+\delta)\left(\frac{1-q}{m_{-}}-\frac{q}{m_{+}}\right)=(1+\delta)\frac{(1-q)m_{+}-qm_{-}}{m_{+}m_{-}}.$
Since $m_{-}=C-m_{+}$, the numerator becomes $m_{+}-qC.$ Also, $m_{+}-qC=((1+C)x^{*}-q)-qC=(1+C)(x^{*}-q).$
Therefore, $\gamma_{-}-\gamma_{+}=(1+\delta)\frac{(1+C)(x^{*}-q)}{m_{+}m_{-}}$
where $m_{+},m_{-}>0.$ So, the sign of $\gamma_{-}-\gamma_{+}$ is
the same as the sign of $x^{*}-q$. 
\end{proof}

\begin{proof}[Proof of Proposition \ref{prop:changing_lambda}]
We will use the following simple monotonicity property of majority rule. Let $a_k:=q\,\mathbb E[\sigma^{\mathrm{maj}}(1,k)]
+(1-q)\mathbb E[\sigma^{\mathrm{maj}}(-1,k)]$
be the expected number of positive signals endorsed by the majority rule when the
sample contains \(k\) positive signals and \(\omega=1\). The sequence \((a_k)_{k=0}^K\) is
nondecreasing and nonconstant whenever \(0<C<K\). Hence $p\mapsto \sum_{k=0}^K \Pr[\operatorname{Binom}(K,p)=k]a_k$
is strictly increasing in $p$. In particular, for fixed \(x>q\),
\(\phi_{\sigma^{\mathrm{maj}}}^{\lambda}(x)\) is strictly increasing in \(\lambda\).

We first show that we can choose $t_0(n)$ such that $\sigma^{\text{maj}}$ is an equilibrium for $n$ sufficiently large and $x(n) \rightarrow \overline{x}$ in probability. The main step is the following lemma.
\begin{lem}
\label{lem:changing_lambda}Suppose $\lambda=0$ for the first $t_{0}$
periods and then $\lambda=1$ for all subsequent periods. There exists
a number $\bar{t}$ and a function $\bar{n}(t)$ so that for any $t_{0}\ge\bar{t}$
and $n\ge\bar{n}(t_{0})$, $\sigma^{\text{maj}}$ is an equilibrium
in a society with $n$ agents. Given any $\epsilon>0$, there exists
a number $\hat{t}$ and a function $\hat{n}(t)$ so that for any $t_{0}\ge\hat{t}$
and $n\ge\hat{n}(t_{0})$, we have $|x(n)-\overline{x}|<\epsilon$
under strategy $\sigma^{\text{maj}}$ with probability at least $1-\epsilon$.
\end{lem}

\begin{proof}
We first show the second claim. Let $\epsilon>0$. When $\lambda=0$ in all periods, the function
$\phi_{\sigma^{\mathrm{maj}}}(x)$ is constant. Let its unique fixed point
be $\underline{x}$. 
Lemma~\ref{lem:1_implies_2} applied to $\sigma^{\mathrm{maj}}$ implies
$\underline{x}>q$.  We will bound $x(t_{0})-\underline{x}$.

Since $z(t)=(N_0^++\operatorname{Binom}(t,q))/(n_0+t)$, the seed signals' contribution to $z(t)$ is $O(1/t)$. By the Chernoff bound, for any $\delta>0$ we can
choose a constant $B>0$ such that $|z(t)-q|<\delta$
with probability at least $1-2e^{-Bt}$. So we can choose $\hat{t}'$
such that this holds for all $t\geq\hat{t}'$ with probability at
least $1-\epsilon/4$. Now taking $\delta$ sufficiently small and
$\hat{t}$ sufficiently large (compared to $\hat{t}'$), by  Proposition~\ref{prop:conv} (since $\phi^0_{\sigma^{\mathrm{maj}}}$ is constant in $x$ with value $\underline{x}$, the stochastic approximation converges $x(t)\to\underline{x}$ a.s.), we have $|x(t_{0})-\underline{x}|<\epsilon/2$ with
probability at least $1-\epsilon/2$ for any $t_{0}\ge\hat{t}$.

Now for each $t_{0}\ge\hat{t},$ consider the infinite-horizon stochastic
process $x(t)$ that starts with $t_{0}$ periods of $\lambda=0$
and subsequently continued with $\lambda=1$ and $\sigma=\sigma^{\text{maj}}$.
We know $x(t)$ converges almost surely as $t\rightarrow\infty$ from
Theorem 2.1 of Chapter 2 of \citet{borkar2009stochastic}, which applies
as in Proposition~\ref{prop:conv} because $\sum_{t=t_{0}}^{\infty}\frac{1}{t}=\infty.$
We next show the steady state reached is $\overline{x}$ with probability
at least $1-\epsilon$.

We can condition on the event $|x(t_{0})-\underline{x}|<\epsilon/2$,
which occurs with probability at least $1-\epsilon/2$. We claim that
given this event, with probability at least $1-\epsilon/2$ there do not exist any $t_{2}>t_{1}>t_{0}$ such that $x(t_{1})>\underline{x}-\epsilon/2$ and
$x(t_{2})<\underline{x}-\epsilon$.

We will add the superscript $\lambda$ in  $\phi_{\sigma^{\text{maj}}}^{\lambda}$ to clarify the viral weight associated with the inflow accuracy function.  By the monotonicity observation at the start of the proof of Proposition \ref{prop:changing_lambda}, we have $\phi_{\sigma^{\mathrm{maj}}}^{1}(\underline{x})
>
\phi_{\sigma^{\mathrm{maj}}}^{0}(\underline{x})
=
\underline{x},$
where the strict inequality uses $\underline{x}>q$ and $0<C<K$.
So shrinking $\epsilon$ if necessary, we can choose $\delta>0$ so
that 
\begin{equation}
\phi_{\sigma^{\text{maj}}}^{1}(x)>\underline{x}+\delta\label{eq:definedelta}
\end{equation}
for $x\in[\underline{x}-\epsilon,\underline{x}].$ If there exist $t_2>t_1>t_0$ such that $x(t_{1})>\underline{x}-\epsilon/2$
and $x(t_{2})<\underline{x}-\epsilon$,
then take $t_2$ to be the \emph{first} time after $t_1$ that $x(t)<\underline{x}-\epsilon$, so that $x(t)\ge\underline{x}-\epsilon$ for all $t\in[t_1,t_2)$. Increasing $t_1$ if necessary we can assume that $x(t)\le\underline{x}$ for all $t_1\le t\le t_2$. (Increase $t_0$ so that $|x(t+1)-x(t)|<\epsilon/2$ for all $t\ge t_0$, which holds since $|x(t+1)-x(t)|\le B_3/(t+1)\to 0$ for some constant $B_3$ after fixing the finite seed pool.)  Then $x(t)\in[\underline{x}-\epsilon,\underline{x}]$ for all $t\in[t_1,t_2-1]$, so inequality~(\ref{eq:definedelta}) applies at every step.

Applying the shifted binomial Chernoff bound to $z(t)$ again, we can choose a constant
$B>0$ such that 
\begin{equation}
\max_{x\in[0,1]}|\phi_{\sigma^{\text{maj}},z(t)}^{1}(x)-\phi_{\sigma^{\text{maj}}}^{1}(x)|<\delta\label{eq:chernoff2}
\end{equation}
with probability at least $1-2e^{-Bt}$ for $t$ sufficiently large.
Increasing $t_{0}$ if necessary, we can assume the inequality (\ref{eq:chernoff2})
holds for all $t\geq t_{0}$ with probability at least $1-\epsilon/4$.
We also condition on this event.

As in the proof of Lemma~\ref{lem:drift_right}, we can write $\mathbf{y}(t)$ as a stochastic-approximation recursion with drift $h(\mathbf{y}(t))$, martingale difference term $M(t+1)$, and summable perturbations from the initial seed pool. The scaled martingale increments satisfy $|M(t)|<B_4/t$ for some constant $B_4$ and all large $t$. So by Theorem C.7 from Appendix
C of \citet{borkar2009stochastic}, for any $\alpha>0$ and any $t_{1}$,
\begin{equation}
\mathbb{P}\left(\sup_{t_{2}\in(t_{1},\infty)}\left|\sum_{i=t_{1}}^{t_{2}}M(i)\right|>\alpha\right)\leq4e^{-\frac{\alpha^{2}}{\sum_{i=t_{1}}^{\infty}B_4^{2}/i^{2}}}.\label{eq:martingaleconcentration2}
\end{equation}

We have 
\begin{align*}
x(t_{2})-x(t_{1}) & =\sum_{t=t_{1}}^{t_{2}-1}\frac{\phi_{\sigma,z(t)}(x(t))-x(t)}{t+1}+\sum_{t=t_{1}}^{t_{2}-1}M(t+1)\\
 & =\sum_{t=t_{1}}^{t_{2}-1}\frac{\phi_{\sigma,z(t)}(x(t))-\phi_{\sigma}(x(t))}{t+1}+\sum_{t=t_{1}}^{t_{2}-1}\frac{\phi_{\sigma}(x(t))-x(t)}{t+1}+\sum_{t=t_{1}}^{t_{2}-1}M(t+1)\\
 & \geq-\sum_{t=t_{1}}^{t_{2}-1}\delta\cdot\frac{1}{t+1}+\sum_{t=t_{1}}^{t_{2}-1}\delta\cdot\frac{1}{t+1}+\sum_{t=t_{1}}^{t_{2}-1}M(t+1)\text{ by inequalities }(\ref{eq:definedelta})\text{ and }(\ref{eq:chernoff2})\\
 & =\sum_{t=t_{1}}^{t_{2}-1}M(t+1).
\end{align*}

Recall that $x(t_{2})-x(t_{1})\leq-\epsilon/2$. Combined with the display equation above, this requires $\sum_{t=t_1}^{t_2-1}M(t+1)\le x(t_2)-x(t_1)\le -\epsilon/2$, i.e., $\left|\sum_{t=t_1}^{t_2-1}M(t+1)\right|\ge\epsilon/2$. So given $t_{1}$,
inequality (\ref{eq:martingaleconcentration2}) with $\alpha=\epsilon/2$
bounds   the probability that $\sup_{t_2>t_1}\left|\sum_{i=t_1}^{t_2}M(i)\right|>\epsilon/2$ by at most $4e^{-\frac{\epsilon^2}{\sum_{i=t_1}^{\infty}4B_4^2/i^2}}$
Increasing $t_{0}$ if necessary, we can assume that the sum of these
probabilities over all $t_{1}\geq t_{0}$ is at most $\epsilon/4$, proving our claim.

Combining our bounds, we conclude that $x(t)\geq\underline{x}-\epsilon$
for all $t\geq t_{0}$ with probability at least $1-\epsilon$.  By Lemma~\ref{lem:maj_shape}, $\phi_{\sigma^{\mathrm{maj}}}^{1}$ is concave
on the region where sampling accuracy is at least $1/2$; when $\lambda=1$,
this is the region $x\ge 1/2$. Also, as in the proof of
Lemma~\ref{lem:maj_ss_types},
\[
\phi_{\sigma^{\mathrm{maj}}}^{1}(1/2)>1/2
\quad\text{and}\quad
\phi_{\sigma^{\mathrm{maj}}}^{1}(1)<1.
\]
Therefore $\phi_{\sigma^{\mathrm{maj}}}^{1}(x)-x$ has a unique zero on
$[1/2,1]$, which we denote by $\bar{x}$. 
Shrinking $\epsilon$ further if necessary so that $\underline{x}-\epsilon>1/2,$
$\bar{x}$ is the only steady state in $[\underline{x}-\epsilon,1]$.
We must therefore have $x(t)\to\bar{x}$ with probability at least
$1-\epsilon$. Hence $\limsup_{n\to\infty}\mathbb{P}(|x(n)-\bar{x}|<\epsilon)\ge1-\epsilon$.
So there is some $\hat{n}(t_{0})$ such that for all $n\ge\hat{n}(t_{0}),$
we have $|x(n)-\bar{x}|<\epsilon$ under strategy $\sigma^{\text{maj}}$
with probability at least $1-2\epsilon$. Replacing $\epsilon$ throughout
the argument with $\epsilon/2$ gives the stated conclusion. 

To complete the proof, we show the first claim that the majority rule
$\sigma^{\text{maj}}$ is an equilibrium when $t_{0}$ and $n$ are
sufficiently large. When $x(t)=\bar{x}$ and $\lambda=1$, sampled signals have accuracy
$\bar{x}>q$. Hence after any strict positive majority in the sample, the
posterior favors $\omega=1$ even if the private signal is negative; after
any strict negative majority, the posterior favors $\omega=-1$ even if
the private signal is positive; and at a tie the posterior follows the
private signal. Thus the unique optimal feasible number of positive
signals endorsed is exactly the choice under the majority rule.
By continuity, find $\epsilon>0$ so that if the event $\{|x(t)-\overline{x}|<\epsilon\}$
happens with probability at least $1-\epsilon,$ then the majority
rule still gives a strictly higher payoff than any other pure strategy.
Using the second part of the claim just proved, find $\hat{t}$ and
$\hat{n}(t)$ so that for any $t_{0}\ge\hat{t}$ and $n\ge\hat{n}(t_{0})$,
we have $|x(n)-\overline{x}|<\epsilon/2$ under strategy $\sigma^{\text{maj}}$
with probability at least $1-\epsilon/2$. Set $\bar{t}=\hat{t}$.
For each $t_{0}\ge\bar{t},$ let $\bar{n}(t_{0})$ be large enough
so that $\hat{n}(t_{0})/\bar{n}(t_{0})<\epsilon/2.$ When the total
number of agents is $n\ge\bar{n}(t_{0}),$ an agent in a uniformly
random position has at least $1-\epsilon/2$ chance of being in position
$\hat{n}(t_{0})$ or later, and if they are in such positions they
have at least $1-\epsilon/2$ chance of facing a current viral accuracy
$x(t)$ with $|x(t)-\overline{x}|<\epsilon/2$ when all others use
the strategy $\sigma^{\text{maj}}$. Thus $\sigma^{\text{maj}}$ is
the agent's best response.
\end{proof}
We can now complete the proof that we can choose $t_0(n)$ such that $\sigma^{\text{maj}}$ is an equilibrium for $n$ sufficiently large and $x(n) \rightarrow \overline{x}$ in probability. Take any decreasing
sequence $\epsilon^{(k)}\to0$. We will construct two increasing
sequences $t_{0}^{(k)}$ and $n^{(k)}$ inductively. Given
$t_{0}^{(1)},...,t_{0}^{(m)}$ and $n^{(1)},...,n^{(m)}$, we can apply Lemma~\ref{lem:changing_lambda} to
find numbers $t_{0}^{(m+1)}$ and $n^{(m+1)}$ so that for $t_{0}^{(m+1)}$
and for any $n\ge n^{(m+1)}$, $|x(n)-\overline{x}|<\epsilon^{(m+1)}$
under strategy $\sigma^{\text{maj}}$ with probability at least $1-\epsilon^{(m+1)}$
and $\sigma^{\text{maj}}$ is an equilibrium. It is without loss to
assume $t_{0}^{(m+1)}>\max\{t_{0}^{(1)},...,t_{0}^{(m)}\}$ and $n^{(m+1)}>\max\{n^{(1)},...,n^{(m)}\}$
(increasing them if necessary). Now for each $n,$ find the largest
$n^{(k)}$ so that $n\ge n^{(k)}$ and let $t_{0}(n)=t_{0}^{(k)}$
(if $n<n^{(1)},$ then set $t_{0}(n)=0$). This ensures (provided
$n\ge n^{(1)})$ that for this choice of $t_{0}(n)$, we have $\sigma^{\text{maj}}$
as an equilibrium and this equilibrium induces $\mathbb{P}[|x(n)-\overline{x}|<\epsilon^{(k)}]>1-\epsilon^{(k)}$.

We now prove the final statement in the proposition. By Lemma~\ref{lem:1_implies_2}, we have $\overline{x}>q$. Fix any virality weight $\lambda'$ and state-symmetric strategy $\sigma$ and suppose there is a steady state $x^*>\overline{x}$ with time-invariant virality weight $\lambda'$ and strategy $\sigma$. By Theorem~\ref{thm:ss_half_stable}, we have $\phi^{\lambda'}_{\sigma}(x^*)=x^*$.

We claim that
\[
\phi_{\sigma^{\mathrm{maj}}}^{\lambda'}(x^*)\ge
\phi_{\sigma}^{\lambda'}(x^*)=x^*.
\]
Let $r:=\lambda' x^*+(1-\lambda')q$
be the sampling accuracy at \(x^*\), and write
\(P_j=\Pr[\operatorname{Binom}(K,r)=j]\). Since \(x^*>q\), we have
\(r\ge q\), with strict inequality whenever \(\lambda'>0\).

Fix \(k>K/2\). Write $A:=\mathbb E[\sigma(1,k)], B:=\mathbb E[\sigma(-1,k)].$ By state symmetry, $\mathbb E[\sigma(1,K-k)]=C-B,
\mathbb E[\sigma(-1,K-k)]=C-A.$
The contribution of the pair \(k,K-k\) to the numerator of
\(\phi_{\sigma}^{\lambda'}(x^*)\) is
\[
\begin{aligned}
T_\sigma(k)
&=
P_k\big(qA+(1-q)B\big)
+
P_{K-k}\big(q(C-B)+(1-q)(C-A)\big)\\
&=
P_{K-k}C
+
A\big(qP_k-(1-q)P_{K-k}\big)
+
B\big((1-q)P_k-qP_{K-k}\big).
\end{aligned}
\]
Because $\frac{P_k}{P_{K-k}}
=
\left(\frac{r}{1-r}\right)^{2k-K}
\ge
\frac{q}{1-q},
$
and \(q>1/2\), both coefficients
$qP_k-(1-q)P_{K-k}$ and $(1-q)P_k-qP_{K-k}$
are weakly positive. Therefore \(T_\sigma(k)\) is maximized subject to
feasibility by setting $A=B=U_k$. This matches the majority rule in samples with \(k\) positive
signals, while state symmetry then implies the mirror choice
\(C-U_k=L_{K-k}\) in samples with \(K-k\) positive signals.

If \(K\) is even and \(k=K/2\), write $A:=\mathbb E[\sigma(1,K/2)].$
State symmetry gives $\mathbb E[\sigma(-1,K/2)]=C-A.$ The contribution of the $k=K/2$ term to the numerator of
\(\phi_{\sigma}^{\lambda'}(x^*)\) is
\[
P_{K/2}\big(qA+(1-q)(C-A)\big)
=
P_{K/2}\big((2q-1)A+(1-q)C\big),
\]
which is maximized by \(A=U_{K/2}\), again matching the majority-rule choice.
Summing over all $(k, K-k)$ pairs proves the claim.

Because \(x^*>q\), the monotonicity observation gives $\phi_{\sigma^{\mathrm{maj}}}^{1}(x^*)
\ge
\phi_{\sigma^{\mathrm{maj}}}^{\lambda'}(x^*)
\ge
x^*$. Since \(\phi_{\sigma^{\mathrm{maj}}}^{1}(1)<1\), the intermediate-value
theorem gives a fixed point $x^{**}\in[x^*,1)$
of \(\phi_{\sigma^{\mathrm{maj}}}^{1}\). But \(x^{**}\ge x^*>\bar{x}\),
contradicting the fact that \(\bar{x}\) is the unique fixed point of
\(\phi_{\sigma^{\mathrm{maj}}}^{1}\) in \([1/2,1]\).
\end{proof}

\section{\label{subsec:simulation_details}Details of the Equilibrium Simulations
for $\lambda>\lambda^{*}$}

In these simulations, we fix signal precision
$q=0.55$ and  capacity $C=3$. We consider sample sizes $K\in\{6,8,10\}$ and all virality weights higher than the critical virality weight in a grid of width 0.02, $\lambda \in \{1, 0.98, 0.96, ...  \}$. We let $n_0=K$ and suppose each seed is correct with probability $q$ and starts with a score of 1. Equivalently, we can think of adding $K$ initial agents who do not see samples and simply add their private signals to the pool of signals. We describe in detail below the methods for calculating equilibrium and estimating equilibrium beliefs. 

\subsection{Symmetric Pure-Strategy Limit Equilibria}
We first identify any symmetric pure-strategy limit equilibria for each $K, \lambda$ pair. Under any
symmetric pure strategy, the likelihood ratio of $\omega=1$ to $\omega=-1$
after observing $k=K/2$ positive sampled signals is 1, whereas
the likelihood ratio after observing $k=K/2+1,K/2+2,...,K$ positive
signals is the reciprocal of the likelihood ratio after observing
$K-k$ such signals. For each $k\in\{K/2+1,...,K\}$, the likelihood
ratio falls into one of the following three cases: (1) between $\frac{1-q}{q}$
and $\frac{q}{1-q},$ so it is optimal for the agent to follow their private
signal; (2) below $\frac{1-q}{q},$ so it is optimal for the agent to
endorse as many negative signals as possible; (3) above $\frac{q}{1-q},$
so it is optimal for the agent to endorse as many positive signals as
possible. Each of the $3^{K/2}$ assignments of these three cases
to various values of $k\in\{K/2+1,...,K\}$ implies a best-responding
strategy, and a pure-strategy limit equilibrium generates likelihood ratios for which the strategy is a best response. 

We use up to three rounds of simulations with increasing precision to identify symmetric pure-strategy limit equilibria. First, for each $(K, \lambda)$ and each of the $3^{K/2}$ candidate  equilibria,  we conduct 10,000 repetitions of a numerical simulation
with 100,000 agents, where all agents use the candidate strategy. A strategy passes the first round if it best responds to the likelihood ratios that it generates in the simulation. 

In the second round, we test a strategy for parameters $(K, \lambda)$ if it passed the first round for $(K, \lambda), (K, \lambda - 0.02)$, or $(K, \lambda+0.02)$. (This guards against missing equilibria that failed the first round due to simulation noise.) The second round is a high-precision re-run of the first round with 100,000 repetitions and 500,000 agents per simulation. If a strategy passes the second round for parameters $(K, \lambda)$ and all simulated likelihood ratios are more than 5 standard errors away from the decision boundaries, then we identify it as a pure-strategy limit equilibrium under those parameters. If a strategy generates  simulated likelihood ratios that are within 5 standard errors of the decision boundary for at least one of the sample realizations, we do a third round of re-run with $10^6$ repetitions and $5\times 10^6$ agents per simulation. In this scenario, the outcome of the third round of simulations determines whether the strategy is identified as a limit equilibrium. 

We find a unique pure-strategy limit equilibrium for $\lambda \in \{0.88, 0.9, 0.92, 0.94, 0.96\}$ when $K=6$, $\lambda \in \{0.8, 0.82, 0.84\}$ when $K=8$, and $\lambda \in \{0.72, 0.74, 0.76, 0.94, 0.96, 0.98, 1\}$ when $K=10$ (as shown in Appendix Table \ref{tab:eqm-all}). We find no pure-strategy limit equilibrium for the other parameter values. 

\subsection{Mixed-Strategy Limit Equilibria}

For $(K, \lambda)$ pairs where we do not find a symmetric pure-strategy limit equilibrium, we search for a symmetric mixed-strategy limit equilibrium in a one-parameter family: the strategy is almost majority rule, except when the agent sees a sample with $K/2+1$ signals on the majority side, they will follow their private signal
with some probability $p\in(0,1)$ and follow the sample majority with the complementary probability. For this strategy to be optimal, the likelihood ratio after observing $K/2+1$
positive signals must be  exactly equal to $\frac{q}{1-q}$, whereas the likelihood
ratio after observing $k>K/2+1$ positive signals must be  strictly above
$\frac{q}{1-q}$.

For each $(K, \lambda)$ pair, we consider all mixed strategies
with the mixing probabilities $p=0,0.05,...,0.95,1.0.$ For each such mixed
strategy, we conduct 30,000 repetitions of a numerical simulation with
30,000 agents who use this strategy. We use the simulated likelihood ratios under different $p$'s to linearly interpolate the mixing probability $p_\text{center}$ that would set the  likelihood ratio from observing $k=K/2+1$ positive signals to be exactly $\frac{q}{1-q}$.  

We then conduct another set of simulations of higher precision, zooming in on the mixing probabilities near $p_\text{center}$. For each $(K, \lambda)$ pair, the second set of simulations focus on the bracket $[p_\text{center}-0.15, p_\text{center}+0.15] \cap [0,1]$. We evenly place 21 mixing probabilities in the bracket. For
each mixed strategy, we conduct 100,000 repetitions of a numerical simulation
with 60,000 agents. These simulations allow us to estimate the
equilibrium in a society with $t$ agents for each $t\in\{200,201,...,60000\}$
by linearly interpolating the value of $p$ that would set the likelihood
ratio of an observation with $K/2+1$ positive signals to be exactly
$\frac{q}{1-q}$. Then, to estimate limit equilibria, we use constrained non-linear least squares to fit a rational function of
the form $t\mapsto\frac{at+b}{ct+1}$ to approximate the equilibrium
mixing probability $p_{t}$ in a society with $t$ agents, under the constraint $ a\ge 0$, $-1 \le b \le 2$, and  $c \ge 10^{-10}$. We divide
the estimated coefficients $a$ and $c$ in the rational function
to estimate $\lim_{t\to\infty}p_{t}$. This procedure estimates the limit equilibrium mixing probabilities reported in Appendix Table \ref{tab:eqm-all}. 

\subsection{Equilibrium Beliefs}
For each $(K, \lambda)$, we conduct a final set of simulations using the identified equilibrium (pure or mixed) to estimate equilibrium beliefs. For each parameter value, we conduct 200,000 repetitions of a numerical simulation with 150,000 agents who use the equilibrium strategy. We simulate the beliefs that agents would have after sampling $k$ positive signals in societies with $t=$ 100, 200, ... 150,000 agents. We use constrained non-linear least squares to fit a rational function of
the form $t\mapsto\frac{at+b}{ct+1}$ to approximate the belief from seeing $k$ positive sampled signals in a society of $t$ agents,  under the constraint $ a\ge 0$, $-1 \le b \le 2$, and  $c \ge 10^{-10}$. We divide
the estimated coefficients $a$ and $c$ in the rational function
to estimate the asymptotic beliefs as $t\to\infty$. This gives the estimated beliefs from seeing sample majorities of different sizes in Appendix Table \ref{tab:eqm-all}. 

\subsection{Detailed Simulation Results}
\begin{table}[H]
\centering
\scriptsize
\setlength{\tabcolsep}{2pt}
\textbf{$K = 6$} \quad ($\lambda^*_K \approx 0.8776$)\\[2pt]
\begin{tabular}{lcccccccccc}
\toprule
$\lambda$ & $p$ & $x^{*}_M$ & $x^{*}_I$ & $\pi_M$ & $E[x^{*}]$ & $b(\mathrm{maj}=2)$ & $b(\mathrm{maj}=4)$ & $b(\mathrm{maj}=6)$ & $E[b\mid x^{*}_M]$ & $E[b\mid x^{*}_I]$ \\
\midrule
$\lambda \uparrow \lambda^*_K$ & 0.0000 & --- & 0.8568 & 0.0000 & 0.8568 & 0.9536 & 0.9976 & 0.9999 & --- & 0.9451 \\
0.8800 & 0.0000 & 0.2771 & 0.8573 & 0.3170 & 0.6734 & 0.6070 & 0.7311 & 0.8303 & 0.3721 & 0.7154 \\
0.9000 & 0.0000 & 0.2317 & 0.8618 & 0.3312 & 0.6531 & 0.5872 & 0.7034 & 0.8020 & 0.3601 & 0.6978 \\
0.9200 & 0.0000 & 0.2093 & 0.8657 & 0.3438 & 0.6400 & 0.5705 & 0.6792 & 0.7754 & 0.3620 & 0.6818 \\
0.9400 & 0.0000 & 0.1938 & 0.8691 & 0.3534 & 0.6304 & 0.5575 & 0.6591 & 0.7518 & 0.3663 & 0.6682 \\
0.9600 & 0.0000 & 0.1821 & 0.8721 & 0.3580 & 0.6250 & 0.5480 & 0.6433 & 0.7318 & 0.3705 & 0.6574 \\
\rowcolor{mixedrow} 0.9800 & 0.1093 & 0.1880 & 0.8695 & 0.3551 & 0.6275 & 0.5473 & 0.6476 & 0.7401 & 0.3660 & 0.6645 \\
\rowcolor{mixedrow} 1.0000 & 0.2142 & 0.1934 & 0.8675 & 0.3484 & 0.6327 & 0.5473 & 0.6528 & 0.7492 & 0.3605 & 0.6727 \\
\bottomrule
\end{tabular}
\\[10pt]
\textbf{$K = 8$} \quad ($\lambda^*_K \approx 0.7825$)\\[2pt]
\begin{tabular}{lccccccccccc}
\toprule
$\lambda$ & $p$ & $x^{*}_M$ & $x^{*}_I$ & $\pi_M$ & $E[x^{*}]$ & $b(\mathrm{maj}=2)$ & $b(\mathrm{maj}=4)$ & $b(\mathrm{maj}=6)$ & $b(\mathrm{maj}=8)$ & $E[b\mid x^{*}_M]$ & $E[b\mid x^{*}_I]$ \\
\midrule
$\lambda \uparrow \lambda^*_K$ & 0.0000 & --- & 0.8616 & 0.0000 & 0.8616 & 0.9368 & 0.9955 & 0.9997 & 1.0000 & --- & 0.9524 \\
0.8000 & 0.0000 & 0.2356 & 0.8650 & 0.3056 & 0.6727 & 0.5807 & 0.6912 & 0.7898 & 0.8631 & 0.3495 & 0.7267 \\
0.8200 & 0.0000 & 0.2112 & 0.8684 & 0.3225 & 0.6565 & 0.5626 & 0.6652 & 0.7612 & 0.8364 & 0.3498 & 0.7089 \\
0.8400 & 0.0000 & 0.1949 & 0.8714 & 0.3293 & 0.6486 & 0.5481 & 0.6438 & 0.7362 & 0.8116 & 0.3539 & 0.6936 \\
\rowcolor{mixedrow} 0.8600 & 0.0973 & 0.2014 & 0.8692 & 0.3250 & 0.6521 & 0.5483 & 0.6484 & 0.7442 & 0.8211 & 0.3483 & 0.7029 \\
\rowcolor{mixedrow} 0.8800 & 0.1946 & 0.2082 & 0.8673 & 0.3192 & 0.6569 & 0.5460 & 0.6506 & 0.7508 & 0.8297 & 0.3448 & 0.7108 \\
\rowcolor{mixedrow} 0.9000 & 0.2964 & 0.2170 & 0.8656 & 0.3116 & 0.6635 & 0.5444 & 0.6551 & 0.7604 & 0.8412 & 0.3407 & 0.7216 \\
\rowcolor{mixedrow} 0.9200 & 0.3918 & 0.2251 & 0.8646 & 0.3026 & 0.6711 & 0.5442 & 0.6598 & 0.7693 & 0.8513 & 0.3366 & 0.7326 \\
\rowcolor{mixedrow} 0.9400 & 0.4817 & 0.2328 & 0.8642 & 0.2936 & 0.6788 & 0.5435 & 0.6638 & 0.7773 & 0.8601 & 0.3330 & 0.7435 \\
\rowcolor{mixedrow} 0.9600 & 0.5680 & 0.2409 & 0.8644 & 0.2837 & 0.6875 & 0.5440 & 0.6680 & 0.7844 & 0.8675 & 0.3301 & 0.7545 \\
\rowcolor{mixedrow} 0.9800 & 0.6488 & 0.2482 & 0.8652 & 0.2752 & 0.6954 & 0.5437 & 0.6697 & 0.7880 & 0.8715 & 0.3292 & 0.7632 \\
\rowcolor{mixedrow} 1.0000 & 0.7253 & 0.2558 & 0.8666 & 0.2677 & 0.7031 & 0.5432 & 0.6689 & 0.7879 & 0.8718 & 0.3307 & 0.7694 \\
\bottomrule
\end{tabular}
\\[10pt]
\textbf{$K = 10$} \quad ($\lambda^*_K \approx 0.7179$)\\[2pt]
\begin{tabular}{lcccccccccccc}
\toprule
$\lambda$ & $p$ & $x^{*}_M$ & $x^{*}_I$ & $\pi_M$ & $E[x^{*}]$ & $b(\mathrm{maj}=2)$ & $b(\mathrm{maj}=4)$ & $b(\mathrm{maj}=6)$ & $b(\mathrm{maj}=8)$ & $b(\mathrm{maj}=10)$ & $E[b\mid x^{*}_M]$ & $E[b\mid x^{*}_I]$ \\
\midrule
$\lambda \uparrow \lambda^*_K$ & 0.0000 & --- & 0.8657 & 0.0000 & 0.8657 & 0.9236 & 0.9932 & 0.9994 & 1.0000 & 1.0000 & --- & 0.9593 \\
0.7200 & 0.0000 & 0.2696 & 0.8661 & 0.2740 & 0.7026 & 0.5888 & 0.7001 & 0.7993 & 0.8721 & 0.9206 & 0.3532 & 0.7634 \\
0.7400 & 0.0000 & 0.2244 & 0.8694 & 0.2929 & 0.6805 & 0.5653 & 0.6673 & 0.7650 & 0.8415 & 0.8960 & 0.3400 & 0.7417 \\
0.7600 & 0.0000 & 0.2031 & 0.8723 & 0.3035 & 0.6691 & 0.5484 & 0.6427 & 0.7368 & 0.8139 & 0.8721 & 0.3411 & 0.7242 \\
\rowcolor{mixedrow} 0.7800 & 0.1368 & 0.2128 & 0.8699 & 0.2995 & 0.6731 & 0.5473 & 0.6464 & 0.7452 & 0.8245 & 0.8827 & 0.3369 & 0.7351 \\
\rowcolor{mixedrow} 0.8000 & 0.2664 & 0.2222 & 0.8681 & 0.2918 & 0.6797 & 0.5454 & 0.6488 & 0.7523 & 0.8337 & 0.8918 & 0.3339 & 0.7453 \\
\rowcolor{mixedrow} 0.8200 & 0.3945 & 0.2333 & 0.8668 & 0.2828 & 0.6877 & 0.5444 & 0.6527 & 0.7609 & 0.8440 & 0.9014 & 0.3315 & 0.7572 \\
\rowcolor{mixedrow} 0.8400 & 0.5199 & 0.2483 & 0.8660 & 0.2715 & 0.6983 & 0.5447 & 0.6582 & 0.7709 & 0.8550 & 0.9111 & 0.3315 & 0.7706 \\
\rowcolor{mixedrow} 0.8600 & 0.6375 & 0.2791 & 0.8659 & 0.2604 & 0.7131 & 0.5453 & 0.6628 & 0.7788 & 0.8635 & 0.9184 & 0.3457 & 0.7830 \\
\rowcolor{mixedrow} 0.8800 & 0.7479 & 0.2718 & 0.8663 & 0.2501 & 0.7176 & 0.5456 & 0.6656 & 0.7839 & 0.8690 & 0.9230 & 0.3322 & 0.7935 \\
\rowcolor{mixedrow} 0.9000 & 0.8515 & 0.2631 & 0.8673 & 0.2399 & 0.7224 & 0.5461 & 0.6669 & 0.7863 & 0.8715 & 0.9250 & 0.3184 & 0.8022 \\
\rowcolor{mixedrow} 0.9200 & 0.9594 & 0.2530 & 0.8685 & 0.2296 & 0.7272 & 0.5485 & 0.6689 & 0.7880 & 0.8729 & 0.9261 & 0.3026 & 0.8105 \\
0.9400 & 1.0000 & 0.2358 & 0.8713 & 0.2321 & 0.7238 & 0.5347 & 0.6455 & 0.7609 & 0.8483 & 0.9068 & 0.3021 & 0.7951 \\
0.9600 & 1.0000 & 0.1936 & 0.8744 & 0.2530 & 0.7022 & 0.5226 & 0.6198 & 0.7256 & 0.8120 & 0.8754 & 0.2887 & 0.7707 \\
0.9800 & 1.0000 & 0.1768 & 0.8771 & 0.2692 & 0.6885 & 0.5152 & 0.6009 & 0.6974 & 0.7804 & 0.8453 & 0.2940 & 0.7499 \\
1.0000 & 1.0000 & 0.1660 & 0.8792 & 0.2811 & 0.6788 & 0.5127 & 0.5892 & 0.6767 & 0.7549 & 0.8190 & 0.2998 & 0.7341 \\
\bottomrule
\end{tabular}
\medskip{}
\caption{Equilibrium simulations for $q=0.55$, $C=3$, and  different values of $K$ and $\lambda$. The first row of each table, labeled $\lambda \uparrow \lambda^*_K$, is the left limit as virality converges from below to the critical virality weight $\lambda^*_K$ for sample  size $K$. Column $p$ shows the limit equilibrium's probability of following private signal upon observing a sample with $k+2$ signals on one side and $k$ signals on the other side. Columns $x_M^*$ and $x_I^*$ are the misleading and informative steady states, while $\pi_M$ is the probability of converging to the misleading steady state. $E[x^*]$ is the expected steady-state viral accuracy (which is also a normalized version of expected social welfare). Columns $b(\text{maj}=\Delta)$ show beliefs in $\omega=1$ when there are $\Delta$ more positive than negative signals in the sample. $E[b \mid x^*_M]$ and $E[b \mid x^*_I]$ are expected beliefs in the correct state after observing the sample, conditional on the misleading and informative steady states respectively. Mixed-equilibrium rows ($0 < p < 1$) have a light-gray background.}
\label{tab:eqm-all}
\end{table}

\begin{figure}[H]
\centering
\includegraphics[width=\linewidth]{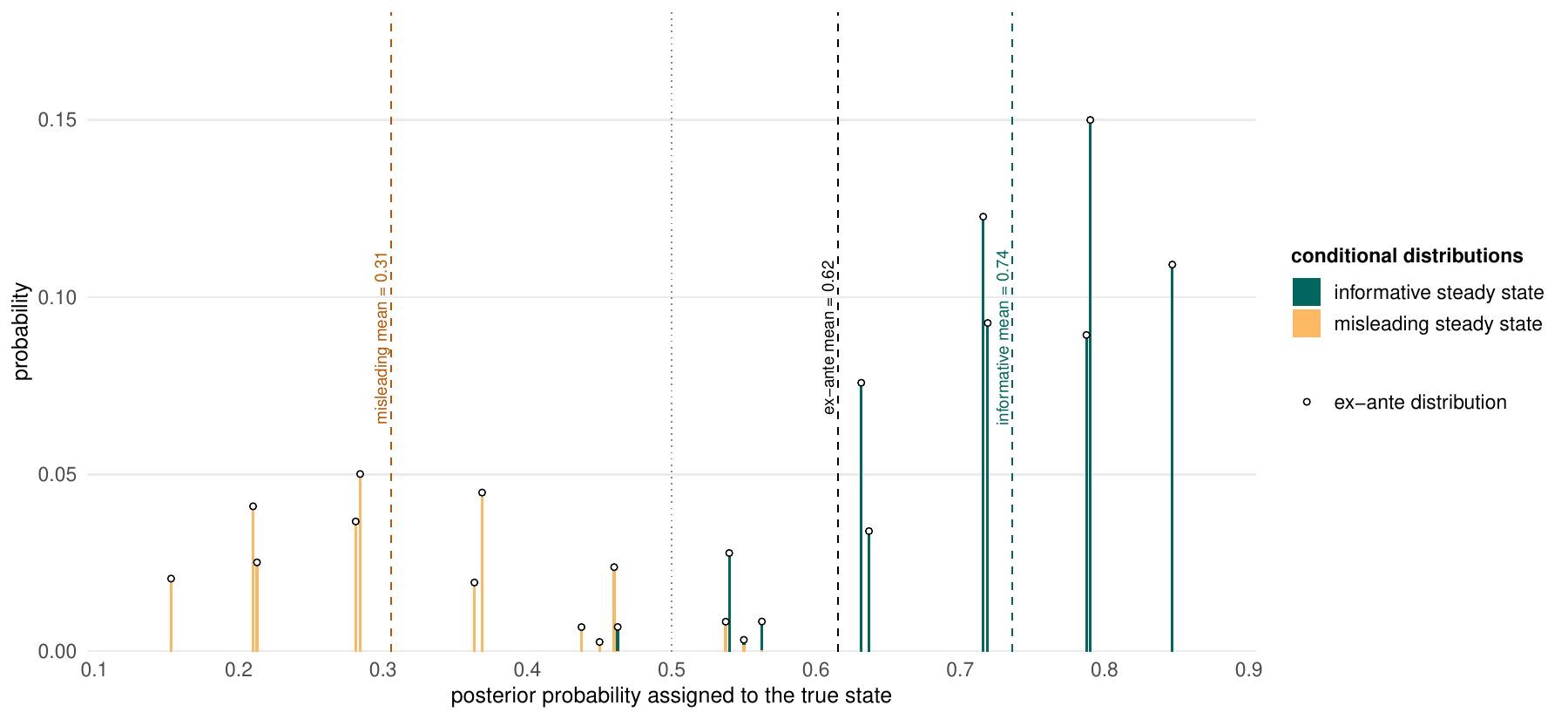}
\caption{Total bar height shows the ex-ante distribution of beliefs in the true state based  on sampled signals and private signals, for $K=10$ and $\lambda=1$. Each color is the conditional distribution of beliefs given one steady state, based on sampled signals and private signals. The two conditional distributions are stacked but they are almost non-overlapping.}
\label{fig:belief_dist_with_signal}
\end{figure}

\end{document}